\documentclass[letterpaper,11pt]{article}

\usepackage{bm}
\usepackage{url}
\usepackage{array}
\usepackage{color}
\usepackage{amsthm}
\usepackage{dsfont}
\usepackage{amsmath}
\usepackage{amssymb}
\usepackage{authblk}
\usepackage{caption}
\usepackage{amsfonts}
\usepackage{booktabs}
\usepackage{graphicx}
\usepackage{mathrsfs}
\usepackage{multicol}
\usepackage{multirow}
\usepackage{cellspace}
\usepackage{enumerate}
\usepackage{enumitem}
\usepackage{makecell}

\usepackage{hyperref}
\hypersetup{
	linktoc      = all,
	colorlinks   = true,
	urlcolor     = blue,
	linkcolor    = blue,
	citecolor    = red
}

\usepackage[backend=biber, isbn=false, style=alphabetic, backref=true, doi=true, url=false, maxcitenames=10, mincitenames=10, maxalphanames=10, maxbibnames=10, minbibnames=5, minalphanames=5, defernumbers=true, sortlocale=en_US]{biblatex}
\usepackage[T1]{fontenc}
\usepackage[dvipsnames]{xcolor}
\usepackage[margin=1in]{geometry}
\usepackage[ruled,vlined,linesnumbered]{algorithm2e}
\usepackage[caption=false,font=normalsize,labelfont=sf,textfont=sf]{subfig}

\usepackage{titlesec}
\titlespacing*{\paragraph}{0pt}{1ex}{1ex}

\setlist[itemize]{leftmargin=*}
\setlist[enumerate]{leftmargin=*}
\setlist[enumerate]{label=(\arabic*)}

\newtheorem{theorem}{Theorem}[section]
\newtheorem{lemma}[theorem]{Lemma}
\newtheorem{claim}[theorem]{Claim}
\newtheorem{corollary}[theorem]{Corollary}

\DeclareMathOperator{\E}{E}
\DeclareMathOperator{\Var}{Var}
\DeclareMathOperator{\poly}{poly}
\DeclareMathOperator{\polylog}{polylog}
\DeclareMathOperator{\spansp}{span}
\DeclareMathOperator{\range}{range}
\DeclareMathOperator{\sgn}{sgn}
\DeclareMathOperator{\nnz}{nnz}
\DeclareMathOperator{\err}{err}

\def\indicator{\mathds{1}}

\def\R{\mathbb{R}}
\def\Z{\mathbb{Z}}

\def\Fro{\mathrm{F}}

\def\tO{\widetilde{O}}
\def\tOmega{\widetilde{\Omega}}
\def\tTheta{\widetilde{\Theta}}

\def\mat{\mathbf}
\def\vec{\boldsymbol}

\def\A{\mat{A}}
\def\D{\mat{D}}
\def\I{\mat{I}}
\def\L{\mat{L}}
\def\M{\mat{M}}
\def\P{\mat{P}}
\def\S{\mat{S}}
\def\X{\mat{X}}

\def\tA{\widetilde{\A}}
\def\tL{\widetilde{\L}}
\def\tM{\widetilde{\M}}
\def\tS{\widetilde{\S}}

\def\b{\vec{b}}
\def\e{\vec{e}}
\def\p{\vec{p}}
\def\r{\vec{r}}
\def\s{\vec{s}}
\def\t{\vec{t}}
\def\u{\vec{u}}
\def\v{\vec{v}}
\def\x{\vec{x}}
\def\y{\vec{y}}
\def\z{\vec{z}}
\def\vzero{\vec{0}}
\def\vone{\vec{1}}
\def\vpi{\vec{\pi}}

\def\pFP{\p_{\mathrm{FP}}}
\def\rFP{\r_{\mathrm{FP}}}

\def\dmax{d_{\max}}

\def\din{d^{-}}
\def\dout{d^{+}}
\def\Deltain{\Delta^{-}}
\def\Deltaout{\Delta^{+}}

\def\deltaout{\delta^{+}}

\def\eps{\varepsilon}
\def\epsr{\eps_{r}}

\def\rmax{r_{\max}}

\def\ns{n_{\mathrm{s}}}

\def\push{\textup{\texttt{Push}}\xspace}
\def\fpush{\textup{\texttt{ForwardPush}}\xspace}
\def\bpush{\textup{\texttt{BackwardPush}}\xspace}
\def\BiPPR{\textup{\texttt{BiPPR}}\xspace}

\def\frow{f_{\mathrm{row}}}
\def\fcol{f_{\mathrm{col}}}

\title{On Solving Asymmetric Diagonally Dominant Linear Systems in Sublinear Time}

\author[1]{Tsz Chiu Kwok}
\author[2]{Zhewei Wei}
\author[2]{Mingji Yang}

\affil[1]{Shanghai University of Finance and Economics \authorcr
	kwok@mail.shufe.edu.cn \vspace{0.8em}}
\affil[2]{Renmin University of China \authorcr
	\{zhewei, kyleyoung\}@ruc.edu.cn}

\date{}

\begin{document}

\maketitle

\begin{abstract}

We initiate a study of solving a row/column diagonally dominant (RDD/CDD) linear system $\M\x = \b$ in sublinear time, with the goal of estimating $\t^{\top}\x^{\ast}$ for a given vector $\t \in \R^n$ and a specific solution $\x^{\ast}$.
This setting naturally generalizes the study of sublinear-time solvers for symmetric diagonally dominant (SDD) systems [Andoni-Krauthgamer-Pogrow, ITCS 2019] to the asymmetric case, which has remained underexplored despite extensive work on nearly-linear-time solvers for RDD/CDD systems.

Our first contributions are characterizations of the problem's mathematical structure.
We express a solution $\x^{\ast}$ via a Neumann series, prove its convergence, and upper bound the truncation error on this series through a novel quantity of $\M$, termed the maximum $p$-norm gap.
This quantity generalizes the spectral gap of symmetric matrices and captures how the structure of $\M$ governs the problem's computational difficulty.

For systems with bounded maximum $p$-norm gap, we develop a collection of algorithmic results for locally approximating $\t^{\top}\x^{\ast}$ under various scenarios and error measures.
We derive these results by adapting the techniques of random-walk sampling, local push, and their bidirectional combination, which have proved powerful for special cases of solving RDD/CDD systems, particularly estimating PageRank and effective resistance on graphs.
Our general framework yields deeper insights, extended results, and improved complexity bounds for these problems.
Notably, our perspective provides a unified understanding of Forward Push and Backward Push, two fundamental approaches for estimating random-walk probabilities on graphs.

Our framework also inherits the hardness results for sublinear-time SDD solvers and local PageRank computation, establishing lower bounds on the maximum $p$-norm gap or the accuracy parameter.
We hope that our work opens the door for further study into sublinear solvers, local graph algorithms, and directed spectral graph theory.

\end{abstract}

\section{Introduction} \label{sec:introduction}

Solving systems of linear equations is one of the most fundamental problems in numerical linear algebra and theoretical computer science.
In the classic version of this problem, we are given a matrix $\M \in \R^{n\times n}$ and a vector $\b \in \R^n$ in the range of $\M$, and the goal is to compute a solution vector $\x \in \R^n$ satisfying $\M\x = \b$.
Beyond general-purpose solvers for arbitrary linear systems (e.g., using fast matrix multiplication or the conjugate gradient method), extensive research has focused on developing efficient solvers for special classes of systems.

In particular, for the important classes of Laplacian and symmetric diagonally dominant (SDD) systems, the breakthrough work of Spielman and Teng~\cite{spielman2004nearly,spielman2014nearly} established the first nearly-linear-time (in $\nnz(\M)$) solvers.
This gave rise to the influential Laplacian Paradigm, which revolutionized algorithmic graph theory and numerical linear algebra, with widespread applications ranging from network science to machine learning (see, e.g., \cite{teng2010laplacian,vishnoi2013Lx}).
Among subsequent efforts to generalize the SDD solvers, a line of work~\cite{cohen2016faster,cohen2017almost,cohen2018solving} developed nearly-linear-time solvers for asymmetric row/column diagonally dominant (RDD/CDD) systems, which significantly expanded the scope of the Laplacian Paradigm.

On the other hand, partly motivated by the advances in quantum algorithms for solving linear systems in sublinear time~\cite{harrow2009quantum}, Andoni, Krauthgamer, and Pogrow~\cite{andoni2019solving} pioneered the study of classical algorithms for approximately solving a single entry of SDD systems in sublinear time.
Under specific access models and error measures, they established:
\begin{itemize}
	\item a $\polylog(n)$-time solver for well-conditioned SDD systems; \footnote{In fact, \cite{andoni2019solving} considers the (effective) condition number of a normalized version of the involved SDD matrix. See Section~\ref{sec:previous_results} for details.}
	\item a $\tOmega\left(\kappa^2\right)$ query lower bound for general SDD systems, where $\kappa$ is the condition number of $\M$.
\end{itemize}
The second result demonstrates the necessity of a quadratic dependence on the condition number (equivalently, the reciprocal of the spectral gap) for sublinear-time SDD solvers.

In light of the previous research for nearly-linear-time RDD/CDD solvers and sublinear-time SDD solvers, it is natural to ask whether the sublinear-time SDD solvers can be extended to the more general RDD/CDD cases.
In this paper, we initiate a study in this direction and give partial positive answers to this question.
We show that the sublinear-time solvers for well-conditioned SDD systems can be extended to ``well-structured'' RDD/CDD systems, provided that we first define an appropriate generalization of the key structural quantity, the spectral gap for symmetric matrices, to asymmetric matrices.
We achieve this by re-characterizing the problem's mathematical structure and introducing a new concept called the \textit{maximum $p$-norm gap}.

Algorithmically, the sublinear SDD solver in \cite{andoni2019solving} works by solely generating \textit{random-walk samplings} based on $\M$ to approximate a truncated Neumann series of the solution.
In contrast, we conduct a deeper investigation of the complexity upper bounds by applying two techniques in addition to random-walk sampling: the \textit{local push method}, which performs local exploration in $\M$; and the \textit{bidirectional method}, which integrates random-walk sampling with local push.
Together, we derive a suite of upper bounds for solving RDD/CDD systems under diverse access models and error measures.
For instance, we extend the algorithmic result in \cite{andoni2019solving} to RDD systems and derive new results for RCDD systems with smaller dependence on some parameters.

Our algorithmic toolkit and investigation of the diverse upper bounds are inspired by recent advances in local algorithms for estimating PageRank~\cite{brin1998anatomy} and effective resistance~\cite{doyle1984random} on graphs~\cite{yang2024efficient,wei2024approximating,wang2024revisiting,wang2024revisitinga,bertram2025estimating,cui2025mixing,yang2025improved,thorup2026pagerank}, which are important special cases of solving RDD/CDD systems.
The techniques of random-walk sampling~\cite{spielman2004nearly,fogaras2005scaling}, local push~\cite{andersen2007pagerank,andersen2008local}, and the bidirectional method~\cite{lofgren2016personalized} have been extensively studied for these problems, and these recent works have further uncovered their new properties and optimality in certain settings.
Nonetheless, previous works typically analyze their applications to PageRank and effective resistance separately.
Our perspective of formulating these problems as solving linear systems, however, provides a more general and unified framework for understanding these techniques and problems, revealing their deeper connections.
As we shall see, this bigger picture yields novel insights, extended results, and improved complexity bounds.

Notably, our perspective reveals a connection between two fundamental local push algorithms on graphs, namely \fpush~\cite{andersen2007pagerank} and \bpush \cite{andersen2008local}.
These algorithms iteratively perform local push operations to explore the graph in opposite directions.
Although both algorithms share similar approaches, they have been treated as distinct methods for different problems, with each analyzed separately.
In contrast, by abstracting both methods as a single algebraic primitive, we demonstrate that \fpush and \bpush are equivalent to applying this primitive to different linear systems.
This characterization helps to explain their distinct properties and enables unified analysis of both approaches.

On the lower-bound side, our framework inherits the hardness result for sublinear-time SDD solvers, establishing the necessity of our assumption on the maximum $p$-norm gap; also, known lower bounds for local PageRank computation imply lower bounds on the accuracy parameter for our setting.
As our work bridges the study of sublinear-time solvers and local graph algorithms, we believe that further investigation could uncover more connections and results for these topics.

In the remainder of this section, we formally define the problem, present our main contributions, and provide a technical overview and discussion of future directions.

\subsection{Basic Notations}

For $n \in \Z^{+}$, we define $[n] := \{1,2,\dots,n\}$.
We call a matrix $\M \in \R^{n\times n}$ \textit{RDD} (\textit{row diagonally dominant}) if it satisfies $\M(j,j) \ge \sum_{k \ne j} \big|\M(j,k)\big|$ for all $j \in [n]$ and its diagonal entries are positive.
We call a matrix \textit{CDD} (\textit{column diagonally dominant}) if its transpose is RDD, and call a matrix \textit{SDD} (\textit{symmetric diagonally dominant}) if it is symmetric and RDD.
It is well-known that any SDD matrix is \textit{PSD (positive semidefinite)}.
We call a square matrix \textit{Z-matrix} if its off-diagonal entries are nonpositive.
We use $\e_k$ to denote the $k$-th canonical unit vector and $\vone$ to denote the all-one vector.
For any matrix or vector, we use $|\cdot|$ to denote taking the entrywise absolute value.

We call two real numbers $p,q > 1$ \textit{H\"older conjugates} (or $q$ is \textit{conjugate to} $p$) if they satisfy $1/p+1/q = 1$.
By convention, we also formally let $\frac{1}{\infty}:=0$ and view $\infty$ and $1$ as H\"older conjugates.
For any $p \in [1,\infty]$, we use $\|\x\|_p$ to denote the \textit{$p$-norm} of a vector $\x$, and $\|\M\|_p$ to denote the \textit{matrix norm induced by vector $p$-norm} of a matrix $\M \in \R^{n \times n}$.

\paragraph{Restriction and Pseudoinverse.}
For a subspace $U \subseteq \mathbb{R}^n$, we use $\M|_U$ to denote the restriction of the linear map $\M$ to $U$, with induced norm $\left\|\M|_U\right\|_p := \max_{\x \in U,\|\x\|_p=1}\|\M\x\|_p$ for any $p \in [1,\infty]$.
We write the \textit{pseudoinverse} (a.k.a. \textit{Moore–Penrose inverse}) of $\M$ as $\M^{+}$.

\paragraph{Spectral Gap.}
For an SDD matrix $\S \in \R^{n \times n}$, we define its \textit{spectral gap} $\gamma(\S)$ as half the smallest nonzero eigenvalue of $\tS := \D_{\S}^{-1/2}\S\D_{\S}^{-1/2}$, where $\D_{\S}$ is the diagonal matrix that satisfies $\D(k,k) = \S(k,k)$ for each $k \in [n]$.
The \textit{(effective) condition number} of $\S$, denoted by $\kappa(\S)$, is defined as the ratio between the largest and smallest nonzero eigenvalues of $\S$.
It holds that $\kappa(\tS) = \Theta\big(1/\gamma(\S)\big)$.

\paragraph{Graphs.}
We consider directed graphs $G = (V,E)$, with $n := |V|$ and $m := |E|$.
We assume that $V = [n]$ for simplicity.
If $(u,v) \in E$, we write $u \to v$.
We denote the (possibly weighted) adjacency matrix as $\A_G \in \R^{n \times n}$.
For each $v \in V$, we define its indegree $\din_G(v) := \sum_{u \to v} \A_G(u,v)$ and outdegree $\dout_G(v) := \sum_{v \to u} \A_G(v,u)$.
The \textit{outdegree matrix} $\D_G \in \R^{n \times n}$ is the diagonal matrix with $\D_G(v,v) = \dout_G(v)$ for each $v \in V$.
We denote $G$'s minimum and maximum outdegree as $\deltaout_G := \min_{v \in V} \big\{\dout_G(v)\big\}$ and $\Deltaout_G := \max_{v \in V} \big\{\dout_G(v)\big\}$, respectively.

\paragraph{Eulerian Graphs and Laplacian.}
We call a graph \textit{Eulerian} if $\din_G(v) = \dout_G(v)$ for all $v \in V$.
On Eulerian graphs, we simply write $d_G(v) := \dout_G(v)$, $\delta_G := \deltaout_G$, and $\Delta_G := \Deltaout_G$.
Undirected graphs constitute a special case of Eulerian graphs where each edge corresponds to two directed edges in opposite directions.
The \textit{directed Laplacian matrix} is defined as $\L_G := \D_G - \A_G^{\top}$, which satisfies $\vone^{\top}\L_G = \vzero^{\top}$ and is CDDZ.
$\L_G$ is RCDDZ for Eulerian graphs and is SDDZ for undirected graphs.

\subsection{Problem Formulation}

We consider a linear system $\M\x = \b$, where $\M \in \R^{n \times n}$ is an RDD/CDD matrix and $\b \in \range(\M)$, and a coefficient vector $\t \in \R^n$.
We assume that $\b,\t \ne \vzero$ and all nonzero entries in $\M$, $\b$, and $\t$ have absolute values in $\big[1/\poly(n), \poly(n)\big]$.
We assume that the algorithms are given the dimension $n$ and have oracle access to $\M$, $\b$, and $\t$ via the following basic queries:
\begin{itemize}
	\item Diagonal queries for $\M$: return $\M(k,k)$ in $O(1)$ time for a given index $k \in [n]$;
	\item Row/column queries for $\M$: return the indices and corresponding values of nonzero entries for a specified row/column of $\M$, in time linear in the number of returned indices;
	\item Entrywise queries for $\b$ and $\t$: return $\b(k)$ or $\t(k)$ in $O(1)$ time for a given index $k \in [n]$.
\end{itemize}
Our results will assume additional access operations, which will be specified in the statements.

As we do not require $\M$ to be invertible, the system may have multiple solutions, and we need to specify what it means for a sublinear-time algorithm to ``solve'' the system.
Following the concept of local computation algorithms~\cite{rubinfeld2011fast} and the previous work~\cite{andoni2019solving}, we consider a fixed solution $\x^{\ast}$ that is determined by $\M$ and $\b$, and require invoking the algorithm with different $\t$ and accuracy parameters returns outputs that are all consistent with the ``global'' solution $\x^{\ast}$.
Our choice of $\x^{\ast}$ will be given in Theorem~\ref{thm:x*}.

Specifically, given an accuracy parameter $\eps > 0$, our goal is to compute an estimate $\hat{x}$ such that $\big|\hat{x}-\t^{\top}\x^{\ast}\big| \le \err(\eps,\M,\b,\t)$ with probability at least $3/4$, where $\err(\eps,\M,\b,\t)$ denotes the problem-specific error bound and $\x^{\ast}$ is a fixed solution determined by $\M$ and $\b$ that satisfies $\M\x^{\ast} = \b$.
Clearly, the problem of approximating a single entry of $\x^{\ast}$ studied in \cite{andoni2019solving} is a special case of our problem, where $\t$ is set to be a canonical unit vector.

We also consider the problems of computing (Personalized) PageRank~\cite{brin1998anatomy} and effective resistance~\cite{doyle1984random} on graphs, viewing them as special cases of our formulation of solving RDD/CDD systems.
For these graph problems, we assume the standard adjacency-list model~\cite{goldreich1997property,goldreich2002property}, where each degree query takes $O(1)$ time and each neighbor query returns a neighbor index along with the edge weight in $O(1)$ time.

For \textit{Personalized PageRank (PPR)}, we consider a directed graph $G$, a \textit{decay factor} $\alpha \in (0,1)$, and a source distribution $\s \in \big\{\y \in \R^n_{\ge 0}: \|\y\|_1 = 1\big\}$.
To ensure that PPR is well-defined, we assume that $\deltaout_G > 0$.
The PPR vector $\vpi_{G,\alpha,\s}$ is defined as the unique solution to the following two equivalent forms of the \textit{PPR equation}:
\begin{align}
	& \left(\I-(1-\alpha)\A_G^{\top}\D_G^{-1}\right)\vpi_{G,\alpha,\s} = \alpha\s, \label{eqn:PPR_I} \\
	& \left(\D_G-(1-\alpha)\A_G^{\top}\right)\left(\D_G^{-1}\vpi_{G,\alpha,\s}\right) = \alpha\s. \label{eqn:PPR_D}
\end{align}
Both equations can be viewed as linear systems of the form $\M\x = \b$, where the coefficient matrices $\I-(1-\alpha)\A_G^{\top}\D_G^{-1}$ and $\D_G-(1-\alpha)\A_G^{\top}$ are both CDDZ and invertible.
Note that for the second form, the solution to the corresponding system is $\D_G^{-1}\vpi_{G,\alpha,\s}$, an outdegree-scaled version of the PPR vector.
We define the PPR value from $s$ to $t$ as $\vpi_{G,\alpha}(s,t) := \vpi_{G,\alpha,\e_s}(t)$, and the \textit{PageRank} vector $\vpi_{G,\alpha}$ as $\vpi_{G,\alpha} := \vpi_{G,\alpha,1/n \cdot \vone}$.
It holds that $\vpi_{G,\alpha}(t) = \frac{1}{n}\sum_{s \in V}\vpi_{G,\alpha}(s,t)$ for all $t \in V$.
We also consider the following two equivalent forms of the \textit{PageRank contribution equation}:
\begin{align}
	& \big(\I-(1-\alpha)\D_G^{-1}\A_G\big)\vpi^{-1}_{G,\alpha,t} = \alpha\e_t, \label{eqn:PageRank_contribution_I} \\
	& \big(\D_G-(1-\alpha)\A_G\big)\vpi^{-1}_{G,\alpha,t} = \alpha\D_G\e_t, \label{eqn:PageRank_contribution_D}
\end{align}
where $t \in V$ is a specified target node and $\vpi^{-1}_{G,\alpha,t}$ is called the \textit{PageRank contribution vector} to $t$.
It holds that $\vpi^{-1}_{G,\alpha,t}(s) = \vpi_{G,\alpha}(s,t)$ for all $s \in V$.
Similarly, both equations can be viewed as linear systems of the form $\M\x = \b$, where the coefficient matrices $\I-(1-\alpha)\D_G^{-1}\A_G$ and $\D_G-(1-\alpha)\A_G$ are both RDDZ and invertible.
Note that for the second form, the corresponding vector $\b$ is $\alpha \D_G\e_t$.

For effective resistance, we consider a connected undirected graph $G$ and two distinct nodes $s,t \in V$.
The \textit{effective resistance} (a.k.a. \textit{resistance distance}) between $s$ and $t$, denoted by $R_G(s,t)$, is defined as the equivalent resistance between $s,t$ if the graph is thought of as an electrical network with each edge $(u,v) \in E$ having resistance $1/\A_G(u,v)$.
Algebraically, $R_G(s,t) = (\e_s-\e_t)^{\top}\L_G^{+}(\e_s-\e_t)$.
As we will establish in Lemma~\ref{lem:ER_quadratic_form}, setting $\M = \L_G$ and $\b = \t = \e_s-\e_t$ in our formulation yields $\t^{\top}\x^{\ast} = R_G(s,t)$.
Here, $\M = \L_G$ is SDDZ and $\b = \e_s-\e_t \in \range(\M)$.

\subsection{Previous Work for SDD Systems} \label{sec:previous_results}

\cite{andoni2019solving} studies the case when the linear system is $\S\x = \b$ for some SDD matrix $\S$.
Define $\D_{\S}$ as the diagonal matrix that satisfies $\D(k,k) = \S(k,k)$ for each $k \in [n]$ and $\tS := \D_{\S}^{-1/2}\S\D_{\S}^{-1/2}$.
They formulate the fixed solution as $\x^{\ast} := \D_{\S}^{-1/2}\tS^{+}\D_{\S}^{-1/2}\b$ and give a Neumann series expansion of $\x^{\ast}$.
For the algorithmic results, they assume that the algorithm is given $\kappa$, an upper bound on $\kappa(\tS)$ (alternatively, $\gamma$ as a lower bound on $\gamma(\S)$) and set a truncation parameter for the Neumann series as $L := \Theta\left(\kappa\log\big(\kappa \cdot \kappa(\D_{\S})\|\b\|_0/\eps\big)\right) = \tTheta(\kappa)$.

Based on the truncated Neumann series, they present a randomized algorithm that, given a coordinate $t \in [n]$, computes an estimate $\hat{x}_t$ satisfying $\big|\hat{x}_t-\x^{\ast}(t)\big| \le \eps\left\|\D_{\S}^{-1}\b\right\|_{\infty}$ with probability at least $3/4$.
The algorithm runs in time $O\left(f(\S) L^3 \log L / \eps^2\right) = \tO\left(f(\S) \kappa^3 \eps^{-2}\right) = \tO\left(f(\S) \gamma^{-3} \eps^{-2}\right)$, where $f(\S)$ is the maximum time cost to simulate one step in the random walk defined by $\S$.
This result is implicit in the proof of \cite[Theorem 5.1]{andoni2018solving}.

On the negative side, \cite{andoni2019solving} proves an $\Omega\left(\kappa(\S)^2/\log^3 n\right) = \tOmega\left(\kappa(\S)^2\right)$ query lower bound (in terms of probing $\b$) for achieving a weaker absolute error bound of $\eps\left\|\x^{\ast}\right\|_{\infty}$, for $\kappa(\S) = O(\sqrt{n}/\log n)$ and $\eps = \Theta(1/\log n)$.
The matrix $\S$ in their hard instance is a Laplacian matrix of a fixed unweighted undirected graph with maximum degree $4$ and thus satisfies $\kappa(\tS) = \Theta\big(\kappa(\S)\big)$.
Therefore, this lower bound can also be written as $\tOmega\left(1/\gamma(\S)^2\right)$.
To our knowledge, no other work has explicitly studied sublinear-time SDD/RDD/CDD solvers.

\subsection{Formulation of $\x^{\ast}$ and the $p$-Norm Gaps} \label{sec:characterization}

Our first contribution is a Neumann-series-based characterization of a solution $\x^{\ast}$ to $\M\x = \b$, which is consistent with the solution considered by \cite{andoni2019solving} for SDD systems.

We decompose $\M$ uniquely as $\M = \D_{\M}-\A_{\M}^{\top}$, where $\D_{\M}$ is a diagonal matrix and all diagonal entries of $\A_{\M}^{\top}$ are $0$.
Define $\tM := \D_{\M}^{-1/2}\M\D_{\M}^{-1/2}$ and $\tA_{\M}^{\top} := \D_{\M}^{-1/2}\A_{\M}^{\top}\D_{\M}^{-1/2}$.

The next theorem gives the definition of $\x^{\ast}$ and its properties.

\begin{theorem} \label{thm:x*}
	For any RDD/CDD $\M$ and $\b \in \range(\M)$, define $\x^{\ast}$ to be
	\begin{align*}
		\x^{\ast} := \frac{1}{2}\sum_{\ell=0}^{\infty}\left(\frac{1}{2}\left(\I+\D_{\M}^{-1}\A_{\M}^{\top}\right)\right)^{\ell}\D_{\M}^{-1}\b.
	\end{align*}
	Then $\x^{\ast}$ is well-defined and satisfies $\M\x^{\ast} = \b$.
	Also, if $\M$ is SDD, then $\x^{\ast} = \D_{\M}^{-1/2}\tM^{+}\D_{\M}^{-1/2}\b$.
\end{theorem}

Next, we study a truncated version of $\x^{\ast}$ and upper bound the truncation error.
As previous analysis based on eigendecomposition is not directly applicable to asymmetric matrices, we introduce a novel concept called the $p$-norm gap of $\M$: for any $p \in [1,\infty]$, we define the \textit{$p$-norm gap} of $\M$ as
\begin{equation}
	\gamma_p(\M) := 1 - \left\|\left.\frac{1}{2}\left(\I+\D_{\M}^{-1/q}\A_{\M}^{\top}\D_{\M}^{-1/p}\right)\right|_{\range\left(\I-\D_{\M}^{-1/q}\A_{\M}^{\top}\D_{\M}^{-1/p}\right)}\right\|_p, \label{eqn:p_norm_gap}
\end{equation}
where $q$ is conjugate to $p$.
We further define the \textit{maximum $p$-norm gap} of $\M$ as $\gamma_{\max}(\M) := \max_{p \in [1,\infty]}\gamma_p(\M)$.\footnote{Alternatively, one can define $\gamma_{\max}(\M) := \max_{p \in \{1,2,\infty\}}\gamma_p(\M)$, which yields a possibly smaller quantity but our main results continue to hold with this definition.}
To our knowledge, no prior work has explicitly studied these quantities.

The following theorem guarantees that for any RDD/CDD matrix $\M$, the maximum $p$-norm gap $\gamma_{\max}(\M)$ lies in $(0,1]$, and when $\M$ is SDD, $\gamma_{\max}(\M)$ coincides with the spectral gap $\gamma(\M)$.
Thus, it serves as a natural generalization of the spectral gap to asymmetric matrices.

\begin{theorem} \label{thm:p_norm_gap}
	If $\M$ is RDD/CDD, then $0 < \gamma_{\max}(\M) \le 1$; if $\M$ is SDD, then $\gamma_{\max}(\M) = \gamma_2(\M) = \gamma(\M)$.
\end{theorem}

We will devise algorithms to estimate the truncated version of $\x^{\ast}$, defined as
\begin{align}
	\x^{\ast}_L := \frac{1}{2}\sum_{\ell=0}^{L-1}\left(\frac{1}{2}\left(\I+\D_{\M}^{-1}\A_{\M}^{\top}\right)\right)^{\ell}\D_{\M}^{-1}\b \label{eqn:x*_L}
\end{align}
for an integer truncation parameter $L$.
The next theorem upper bounds the truncation error in terms of a given lower bound on the maximum $p$-norm gap $\gamma_{\max}(\M)$.

\begin{theorem} \label{thm:truncation_error}
	Suppose $0 < \gamma \le \gamma_{\max}(\M)$.
	To ensure that $\left|\t^{\top}\x^{\ast}_L - \t^{\top}\x^{\ast}\right| \le \frac{1}{2}\eps$, it suffices to set
	\begin{equation}
		L := \Theta\left(\frac{1}{\gamma}\log\left(\frac{1}{\gamma\eps} \cdot \dmax(\M)\|\t\|_0\left\|\D_{\M}^{-1}\t\right\|_{\infty}\|\b\|_0\left\|\D_{\M}^{-1}\b\right\|_{\infty}\right)\right) = \tTheta\left(\frac{1}{\gamma}\right), \footnote{Here and after, we use $\tTheta$ and $\tO$ to hide polylogarithmic factors in $n$, $1/\gamma$, $1/\eps$, and (reciprocals of) quantities in $\M$, $\b$, and $\t$.} \label{eqn:L}
	\end{equation}
	where $\dmax(\M)$ denotes the largest diagonal entry in $\M$.
\end{theorem}

\paragraph{Relationship with the Formulation in \cite{andoni2019solving}.}
Our formulations of $\x^{\ast}$ and the maximum $p$-norm gap generalize the ones in \cite{andoni2019solving}, in the sense that when $\M$ is SDD, $\x^{\ast} = \D_{\M}^{-1/2}\tM^{+}\D_{\M}^{-1/2}\b$ matches their solution and $\gamma_{\max}(\M)$ equals $\gamma(\M)$, the spectral gap of $\M$.
Therefore, the quadratic lower bound on $1/\gamma(\M)$ for SDD systems in that paper translates into a quadratic lower bound on $1/\gamma_{\max}(\M)$ for general RDD/CDD systems.
Also, our setting of the truncation parameter $L$ in Equation~\eqref{eqn:L} matches theirs when $\M$ is SDD and $\t$ is a canonical unit vector.

\subsection{Main Algorithmic Results}

For our algorithmic results, we assume that the algorithm is given as input a quantity $\gamma > 0$ as a lower bound on $\gamma_{\max}(\M)$ and an accuracy parameter $\eps > 0$.
We use $\gamma$, $\eps$, and the specific terms in the accuracy guarantee to set the truncation parameter $L$ according to Equation~\eqref{eqn:L}, where we assume that suitable upper bounds on the quantities in the logarithmic factor are known.

The statements of our results will use the following definition.
For RDD $\M$, we use $\frow(\M)$ to denote the maximum time cost to simulate a single step in the random walk defined by the row substochastic matrix $\frac{1}{2}\left(\I+\D_{\M}^{-1}\left|\A_{\M}^{\top}\right|\right)$.
This quantity depends on the structure and representation of $\M$.
For instance, if each row of $\M$ has at most $d$ nonzero entries, then $\frow(\M) = O(d)$; if the nonzero entries in $\A_{\M}$ have equal absolute values and we are allowed to sample a uniformly random index of the nonzero entries in each column of $\A_{\M}$ in $O(1)$ time, then $\frow(\M) = O(1)$.
For CDD $\M$, we define $\fcol(\M) := \frow\left(\M^{\top}\right)$.

Each of our algorithmic results excels under different system types, access models, and parameter regimes.
We present a selection of our results here, with additional results deferred to the technical sections.
Additionally, a remark at the end of this subsection establishes that most of our results have ``symmetric'' counterparts (e.g., for CDD systems) that can be obtained by exchanging the roles of certain quantities.

We first present our results of using random-walk sampling for RDD systems.

\begin{theorem} \label{thm:MC_cubic}
	Suppose that $\M$ is RDD, we can sample from the distribution $|\t|/\|\t\|_1$ in $O(1)$ time, and $\|\t\|_1$ is known.
	Then there exists a randomized algorithm that computes an estimate $\hat{x}$ such that $\Pr\left\{\left|\hat{x}-\t^{\top}\x^{\ast}\right| \le \eps \left\|\D_{\M}^{-1}\b\right\|_{\infty}\right\} \ge \frac{3}{4}$ in time $O\left(\frow(\M) \|\t\|_1^2 L^3 \eps^{-2}\right) = \tO\left(\frow(\M) \|\t\|_1^2 \gamma^{-3} \eps^{-2}\right)$.
\end{theorem}

This theorem subsumes the main algorithmic result of \cite{andoni2019solving} for SDD systems and extends it to the RDD case while achieving a mild $(\log L)$-factor improvement.
Their original result is recovered as a special case when $\M$ is SDD and $\t$ is a canonical unit vector.
On the other hand, the lower bound in \cite{andoni2019solving} implies that a quadratic dependence on $1/\gamma$ in the compelxity bound is necessary.
Our algorithm is similar to theirs, but we achieve the improved complexity by adopting a different random-walk sampling scheme.

We also show that for RDDZ $\M$ along with nonnegative vectors $\b$ and $\t$, if we allow the relaxed error bound $\eps \|\x^{\ast}\|_{\infty}$, the complexity can be improved to depend quadratically on $L$ via a variance analysis of the random-walk sampling process.
We adopt this error measure since it provides a natural accuracy guarantee and has been previously studied in \cite{andoni2019solving}.
Lemma~\ref{lem:bound_comparison} will establish that $\left\|\D_{\M}^{-1}\b\right\|_{\infty} \le 2\|\x^{\ast}\|_{\infty}$ for RDD systems.

\begin{theorem} \label{thm:MC_square}
	Suppose that $\M$ is RDDZ, $\b,\t \ge \vzero$, we can sample from the distribution $\t/\|\t\|_1$ in $O(1)$ time, and $\|\t\|_1$ is known.
	Then there exists a randomized algorithm that computes a $\hat{x}$ such that $\Pr\left\{\left|\hat{x}-\t^{\top}\x^{\ast}\right| \le \eps \left\|\x^{\ast}\right\|_{\infty}\right\} \ge \frac{3}{4}$ in time $O\left(\frow(\M) \|\t\|_1^2 L^2 \eps^{-2}\right) = \tO\left(\frow(\M) \|\t\|_1^2 \gamma^{-2} \eps^{-2}\right)$.
\end{theorem}

The following theorem further considers the relative error guarantee of $\eps \cdot \t^{\top}\x^{\ast}$.
The complexity depends quadratically on $L$ and linearly on $\|\t\|_1 \left\|\D_{\M}^{-1}\b\right\|_{\infty} \big / \t^{\top}\x^{\ast}$ and is also achieved using random-walk sampling.

\begin{theorem} \label{thm:MC_relative}
	Suppose that $\M$ is RDDZ, $\b,\t \ge \vzero$, we can sample from the distribution $\t/\|\t\|_1$ in $O(1)$ time, $\|\t\|_1$ and $\left\|\D_{\M}^{-1}\b\right\|_{\infty}$ are known, and $\t^{\top}\x^{\ast} > 0$.
	Then there exists a randomized algorithm that computes an estimate $\hat{x}$ such that $\Pr\left\{\left|\hat{x}-\t^{\top}\x^{\ast}\right| \le \eps \cdot \t^{\top}\x^{\ast}\right\} \ge \frac{3}{4}$ in expected time
	\begin{align*}
		O\left(\frac{\frow(\M) \|\t\|_1 \left\|\D_{\M}^{-1}\b \right\|_{\infty} L^2}{\eps^2 \cdot \t^{\top}\x^{\ast}}\right) = \tO\left(\frac{\frow(\M) \|\t\|_1 \left\|\D_{\M}^{-1}\b\right\|_{\infty}}{\gamma^2\eps^2 \cdot \t^{\top}\x^{\ast}}\right).
	\end{align*}
\end{theorem}

Our next result leverages the local push method to derive a deterministic algorithm for special RCDD systems, whose complexity depends linearly on $1/\eps$.

\begin{theorem} \label{thm:push_RCDD}
	Suppose that $\M$ is RCDD, its nonzero entries have absolute values of $\Omega(1)$, and we can scan the nonzero entries of $\b$ in $O\big(\|\b\|_0\big)$ time.
	Then there exists a deterministic algorithm that computes a $\hat{x}$ such that $\left|\hat{x}-\t^{\top}\x^{\ast}\right| \le \eps \|\t\|_1$ in time $O\big(\|\b\|_0\big)$ plus $O\left(\|\b\|_1 L^3 \eps^{-1} \right) = \tO\left(\|\b\|_1 \gamma^{-3}\eps^{-1}\right)$.
\end{theorem}

Lastly, applying the bidirectional method to special RCDD systems yields complexity bounds with improved dependence on $L$ and $\eps$, achieving either $L^{7/3}\eps^{-2/3}$ or $L^{5/2}\eps^{-1}$ using different parameter settings.

\begin{theorem} \label{thm:bidirectional_RCDD}
	Suppose that $\M$ is RCDD, its nonzero entries have absolute values of $\Omega(1)$, we can sample from the distribution $|\t|/\|\t\|_1$ in $O(1)$ time, $\|\t\|_1$ and $\frow(\M)$ are known, and we can scan through the nonzero entries of $\b$ in $O\big(\|\b\|_0\big)$ time.
	Then there exists a randomized algorithm that computes an estimate $\hat{x}$ such that $\Pr\left\{ \left|\hat{x}-\t^{\top}\x^{\ast}\right| \le \eps \right\} \ge \frac{3}{4}$ in time $O\big( \|\b\|_0 \big)$ plus
	\begin{align*}
		& \phantom{{}={}} O\left(\min\left(\frow(\M)^{1/3} \|\t\|_1^{2/3} \|\b\|_1^{2/3} L^{7/3} \eps^{-2/3}, \frow(\M)^{1/2} \|\t\|_1 \|\b\|_1^{1/2} \left\|\D_{\M}^{-1}\b\right\|_{\infty}^{1/2} L^{5/2} \eps^{-1}\right) \right) \\
		& = \tO\left(\min\left(\frac{\frow(\M)^{1/3} \|\t\|_1^{2/3} \|\b\|_1^{2/3}}{\gamma^{7/3} \eps^{2/3}}, \frac{\frow(\M)^{1/2} \|\t\|_1 \|\b\|_1^{1/2} \left\|\D_{\M}^{-1}\b\right\|_{\infty}^{1/2}}{\gamma^{5/2}\eps}\right)\right).
	\end{align*}
\end{theorem}

\paragraph{Remark.}
As we shall see, all theorems in this subsection except Theorem~\ref{thm:MC_square} still hold if we replace RDD/RDDZ by CDD/CDDZ, swap $\b$ and $\t$ (except in $\t^{\top}\x^{\ast}$), and replace $\frow(\M)$ by $\fcol(\M)$ in the statements.
Theorem~\ref{thm:MC_square} is an exception because its proof relies on the property that $\left\|\D_{\M}^{-1}\b\right\|_{\infty} \le 2\|\x^{\ast}\|_{\infty}$ for RDD $\M$, which does not have a simple analog for the CDD case.

\subsection{Connections to PageRank and Effective Resistance Computation}

Our results can be directly applied to the PPR or PageRank contribution equations and the single-pair effective resistance problem to yield complexity bounds for different graph types under various access models and accuracy guarantees.
Notably, as we shall see, for PageRank computation, half the decay factor serves as a lower bound on the maximum $p$-norm gap of the corresponding systems; for effective resistance computation, the spectral gap of the graph Laplacian equals the maximum $p$-norm gap of the system.
Rather than exhaustively presenting all applications of our results, here we highlight selected results that generalize and improve upon the previously best complexity bounds.
Additional applications and comparisons are provided in Sections~\ref{sec:PageRank} and~\ref{sec:effective_resistance}.

The following theorem is derived by applying Theorem~\ref{thm:MC_relative} to the PageRank equation~\eqref{eqn:PPR_D} combined with tighter lower bounds on $\boldsymbol{\pi}_{G,\alpha}(t)$ that we establish for Eulerian graphs.

\begin{theorem} \label{thm:PageRank_Eulerian}
	For any unweighted Eulerian graph $G$, given the decay factor $\alpha$, a target node $t \in V$, and $\delta_G$, there exists a randomized algorithm that estimates the PageRank value $\vpi_{G,\alpha}(t)$ within relative error $\eps$ with probability at least $3/4$ in time
	\begin{align*}
		O \left( \frac{1}{\alpha\eps^2} \cdot \frac{d_G(t)}{\delta_G} \cdot \frac{1}{n\vpi_{G,\alpha}(t)} \right) = O\left(\frac{1}{\eps^2\delta_G} \cdot \min\left(\frac{d_G(t)}{\alpha^2},\frac{m/d_G(t)}{\alpha^2},\frac{\Delta_G}{\alpha},\frac{\sqrt{m}}{\alpha}\right)\right).
	\end{align*}
\end{theorem}

This improves over the previously best upper bounds of $O\left(\frac{1}{\eps^2\delta_G} \cdot \min\left(\frac{d_G(t)}{\alpha^2},\frac{\sqrt{m}}{\alpha^2}\right)\right)$, which was given by \cite{wang2024revisitinga} and stated for unweighted undirected graphs.
Our algorithm is essentially the same as that in \cite{wang2024revisitinga}, which generates random walks from the target node $t$, and the improvement comes from our tighter lower bounds on $\vpi_{G,\alpha}(t)$ for Eulerian graphs.

For effective resistance computation, recall that we consider connected undirected graphs $G$.
Lemma~\ref{lem:ER_quadratic_form} establishes that the value $\t^{\top}\x^{\ast}$ that our algorithms approximate equals $R_G(s,t)$ when setting $\M =\L_G$ and $\b = \t = \e_s - \e_t$.
Thus, Theorems~\ref{thm:MC_cubic} and \ref{thm:bidirectional_RCDD} directly imply the following complexity bounds for estimating effective resistance.

\begin{corollary} \label{cor:ER}
	For any connected unweighted undirected graph $G$, given nodes $s,t \in V$ and $\gamma > 0$ as a lower bound on $\gamma(\L_G)$, there exists a randomized algorithm that estimates the effective resistance $R_G(s,t)$ within absolute error $\eps$ with probability at least $3/4$ in time
	\begin{align*}
		O \left( \min\left( \frac{L^3}{\eps^2 \cdot \min\big(d_G(s),d_G(t)\big)^2}, \frac{L^{7/3}}{\eps^{2/3}}, \frac{L^{5/2}}{\eps \cdot \min\big(d_G(s),d_G(t)\big)^{1/2}} \right) \right),
	\end{align*}
	where $L := \Theta\left(\frac{1}{\gamma} \log\left(\frac{1}{\gamma\eps} \left(\frac{1}{d_G(s)} + \frac{1}{d_G(t)}\right)\right)\right)$.
\end{corollary}

This result subsumes the previously best upper bounds of $O \left( \min\left( \frac{L^3 \log L}{\eps^2 \cdot \min(d_G(s),d_G(t))^2}, \frac{L^{7/3}\log L}{\eps^{2/3}}\right) \right)$ given by \cite{cui2025mixing} with essentially the same setting of $L$.
Moreover, since $R_G(s,t)$ can be lower bounded by $1/2 \big/ \min\big(d_G(s),d_G(t)\big)$~\cite[Corollary 3.3]{lovasz1993random}, by setting the absolute error parameter $\eps$ to be $\epsr \cdot 1/2 \big/ \min\big(d_G(s),d_G(t)\big)$, the last bound in our result implies that an estimate of $R_G(s,t)$ within relative error $\epsr$ can be computed in time $O\left(\min\big(d_G(s),d_G(t)\big)^{1/2} L^{5/2} \epsr^{-1}\right)$.
This improves over the previously best upper bound of $O\left(\min\big(d_G(s),d_G(t)\big)^{1/2} L^{3} \log L \cdot \epsr^{-1}\right)$ given by \cite{yang2025improved}.
Our algorithms are essentially the same as those in \cite{cui2025mixing,yang2025improved}, i.e., using random-walk sampling and the bidirectional method, and the improvements stem from our simpler sampling scheme that avoids using extra data structures and a refined analysis.

On the other hand, known hardness results for local PageRank and effective resistance computation can potentially yield lower bounds for sublinear-time solvers.
We highlight the following result, derived by establishing a reduction from estimating single-node PageRank on undirected graphs to solving SDD systems and applying the lower bound for PageRank computation from \cite{wang2024revisitinga}.

\begin{theorem} \label{thm:lower_bound_eps}
	For any large enough $n$ and $\eps = \Omega(1/n)$, there exist $\b \in \R^n$ and $t \in [n]$ that satisfy the following.
	Every randomized algorithm that, given access to an invertible SDD matrix $\S \in \R^{n \times n}$ whose spectral gap is $\Omega(1)$, succeeds with probability at least $3/4$ to approximate $\x^{\ast}(t)$ within absolute error $\eps\|\x^{\ast}\|_{\infty}$, must probe $\Omega(1/\eps)$ coordinates of $\S$ in the worst case.
	Here, $\x^{\ast} = \S^{-1}\b$.
\end{theorem}

This result gives an $\Omega(1/\eps)$ lower bound for local SDD solvers with accuracy guarantee $\eps\|\x^{\ast}\|_{\infty}$, demonstrating the necessity of a linear dependence on $1/\eps$ in our Theorem~\ref{thm:MC_cubic}.
This lower bound on the accuracy parameter complements the lower bound on the spectral gap given by \cite{andoni2019solving}.

\subsection{Understanding \fpush and \bpush on Graphs}

Inspired by the local push algorithms \fpush and \bpush (see Section~\ref{sec:prelim_push}), we formulate our \push algorithm as a unified primitive for estimating a summation of matrix powers applied to a vector, which is closely related to our approach to solving RDD/CDD systems.
This abstraction reveals that \fpush and \bpush, despite appearing as distinct algorithms for two problems with different propagation strategies, are equivalent to applying \push to different linear systems, modulo variable scaling.
The apparent differences arise from the scaling and the distinct behaviors that \push exhibits on different types of linear systems.

Specifically, for PageRank computation, \fpush from node $s$ corresponds to applying \push to Equation~\eqref{eqn:PPR_D} with $\s = \e_s$ and outdegree scaling, while \bpush from node $t$ applies \push to the contribution equations~\eqref{eqn:PageRank_contribution_I} or \eqref{eqn:PageRank_contribution_D}.
Our analysis demonstrates that \push provides closed-form complexity bounds on certain CDD systems and accuracy guarantees on RDD systems, which explains the known computational properties of \fpush and \bpush on directed graphs.

For RCDD systems, \push inherits both complexity and accuracy advantages, clarifying why \fpush and \bpush perform well on undirected graphs.
This perspective further reveals that on Eulerian graphs, \fpush from any node is equivalent to \bpush from the same node on the transpose graph, establishing a fundamental connection previously unrecognized.

\subsection{Technical Overview}

Our characterizations of the solution $\x^{\ast}$ and the $p$-norm gaps are based on fundamental properties of Neumann series and restricted linear maps.
Since asymmetric matrices may be non-diagonalizable, the eigendecomposition techniques used for symmetric matrices become inapplicable.
We address this challenge by analyzing the operator norms of restricted linear maps to establish series convergence and derive truncation error bounds.

Specifically, by the decomposition $\M = \D_{\M}-\A_{\M}^{\top}$, the equation $\M\x = \b$ is equivalent to $\D_{\M}^{-1}\M\x = \D_{\M}^{-1}\b$, i.e., $\left(\I-\D_{\M}^{-1}\A_{\M}^{\top}\right)\x = \D_{\M}^{-1}\b$.
For certain types of $\M$, $\I-\D_{\M}^{-1}\A_{\M}^{\top}$ is invertible and we can write the unique solution $\x^{\ast}$ as $\x^{\ast} = \left(\I-\D_{\M}^{-1}\A_{\M}^{\top}\right)^{-1}\D_{\M}^{-1}\b = \sum_{\ell=0}^{\infty}\left(\D_{\M}^{-1}\A_{\M}^{\top}\right)^{\ell}\D_{\M}^{-1}\b$.
Our formulation of the specific solution $\x^{\ast}=\frac{1}{2}\sum_{\ell=0}^{\infty}\left(\frac{1}{2}\left(\I+\D_{\M}^{-1}\A_{\M}^{\top}\right)\right)^{\ell}\D_{\M}^{-1}\b$ in Theorem~\ref{thm:x*} is similar to this expression but uses the ``lazified'' matrix $\frac{1}{2}\left(\I+\D_{\M}^{-1}\A_{\M}^{\top}\right)$ instead of $\D_{\M}^{-1}\A_{\M}^{\top}$, which helps to ensure convergence for general RDD/CDD $\M$.

Furthermore, to analyze the truncation error, intuitively, we require a quantity that measures the strength of the diagonal dominance property of $\M$.
If $\M$ is strictly RDD, a natural choice is $1 - \left\|\D_{\M}^{-1}\A_{\M}^{\top}\right\|_{\infty} = \min_{j \in [n]}\left\{\frac{\M(j,j)-\sum_{k \ne j}|\M(j,k)|}{\M(j,j)}\right\} > 0$, which quantifies the strict row diagonal dominance of $\M$.
The larger this quantity, the stronger the diagonal dominance and the faster the Neumann series converges.
For strictly CDD $\M$, a similar quantity is $1 - \left\|\A_{\M}^{\top}\D_{\M}^{-1}\right\|_{1}$.
This hints us to consider the general quantity $1 - \left\|\D_{\M}^{-1/q}\A_{\M}^{\top}\D_{\M}^{-1/p}\right\|_{p}$ for arbitrary $p \in [1,\infty]$.
However, this quantity may be nonpositive for general RDD/CDD matrices $\M$, rendering it unsuitable for bounding the truncation error.
To address this issue, we define the maximum $p$-norm gap by lazifying the involved operator and restricting it to $\range\left(\I-\D_{\M}^{-1/q}\A_{\M}^{\top}\D_{\M}^{-1/p}\right)$ (see Equation~\eqref{eqn:p_norm_gap}), which guarantees positivity while remaining sufficient for bounding the truncation error.

Our algorithms estimate $\t^{\top}\x^{\ast}_L = \frac{1}{2}\t^{\top}\sum_{\ell=0}^{L-1}\left(\frac{1}{2}\left(\I+\D_{\M}^{-1}\A_{\M}^{\top}\right)\right)^{\ell}\D_{\M}^{-1}\b$ by adapting three techniques for random-walk probability estimation on graphs: random-walk sampling~\cite{spielman2004nearly,fogaras2005scaling}, local push~\cite{andersen2007pagerank,andersen2008local}, and the bidirectional method~\cite{lofgren2016personalized}.

When $\M$ is RDD, the matrix $\frac{1}{2}\left(\I+\D_{\M}^{-1}\left|\A_{\M}^{\top}\right|\right)$ is row substochastic, enabling us to interpret $\t^{\top}\x^{\ast}_L$ as an expectation over random walks that start from the distribution $|\t|/\|\t\|_1$ and transition according to $\frac{1}{2}\left(\I+\D_{\M}^{-1}\left|\A_{\M}^{\top}\right|\right)$.
Such random walks may terminate early, and the algorithm needs to record the signs of the entries in $\A_{\M}$ along the walk.
This method extends the approach in \cite{andoni2019solving}, and we adopt a different sampling scheme and conduct variance analysis in some special cases to reduce the dependence on $L$ in the complexity.

Based on local push methods, we formulate a \push primitive that estimates the vector $\x^{\ast}_L = \sum_{\ell=0}^{L-1}\left(\frac{1}{2}\left(\I+\D_{\M}^{-1}\A_{\M}^{\top}\right)\right)^{\ell}\D_{\M}^{-1}\b$ through deterministic local computation.
\push maintains coordinate variables initialized to $\D_{\M}^{-1}\b$ and iteratively applies the linear operator $\frac{1}{2}\left(\I+\D_{\M}^{-1}\A_{\M}^{\top}\right)$ to selected coordinates.
Our analysis relies on invariant properties preserved by these operations, including an inequality variant that helps to handle negative entries.
This approach yields closed-form accuracy guarantees for RDD systems and complexity bounds for special CDD systems.
Combining these two aspects yield our result for special RCDD systems.

The bidirectional method combines random-walk sampling and local push from two directions.
We adapt the \BiPPR framework~\cite{lofgren2016personalized}, performing \push from $\D_{\M}^{-1}\b$ and exploiting the invariant property to construct an unbiased estimator for $\t^{\top}\x^{\ast}_L$, which can be sampled via random walks from $|\t|/\|\t\|_1$ when $\M$ is RDD.
By balancing the computational costs of both components, we achieve improved dependence on $L$ and $\eps$, particularly for RCDD systems.

Crucially, we can transpose the expression for $\t^{\top}\x^{\ast}_L$ to obtain an equivalent summation with $\A_{\M}^{\top}$ replaced by $\A_{\M}$ and the roles of $\b$ and $\t$ exchanged.
This allows us to alternatively apply \push from $\D_{\M}^{-1}\t$ and random-walk sampling from $|\b|/\|\b\|_1$ when $\M$ is CDD, leading to symmetric algorithmic procedures and complexity results.

\subsection{Future Directions}

Our work bridges the study of sublinear-time solvers and local graph algorithms, opening several avenues for future research.

First, improving our complexity bounds remains an important direction from both upper and lower bound perspectives.
For example, our upper bound in Theorem~\ref{thm:MC_cubic} contains factors of $\gamma^{-3}\varepsilon^{-2}$, while only lower bounds of $\tOmega\left(\gamma^{-2}\right)$ and $\Omega\left(\eps^{-1}\right)$ are known.
Achieving optimal bounds across various settings perhaps requires substantial future investigation, even for the symmetric case.

Second, establishing stronger lower bounds through unified techniques represents a promising direction.
Current spectral gap lower bound for SDD solvers and graph parameter lower bounds for PageRank computation rely on different proofs, establishing lower bounds for the number of probes into $\b$ and into the graph structure, respectively.
Combining these techniques may yield generalized and tighter lower bounds.
Notably, for local PageRank and effective resistance computation, the dependence on the decay factor and the spectral gap is crucial~\cite{fountoulakis2019variational,fountoulakis2022open,cai2023effective,cui2025mixing}.
However, surprisingly, no lower bounds on these parameters are known to the best of our knowledge.
In this regard, the spectral gap lower bound for SDD solvers may provide inspiration for proving such lower bounds.

Third, the $p$-norm gaps may be of independent interest in directed spectral graph theory.
Further study may reveal their connections to combinatorial properties of matrices and graphs, which may lead to complexity lower bounds on $\gamma_{\infty}(\M)$ for solving RDD systems and $\gamma_1(\M)$ for CDD systems.

Finally, it would be meaningful to explore additional applications of local RDD/CDD solvers and extend our techniques to other linear system classes.
In fact, although \cite{andoni2019solving} gives a qualitative separation between solving SDD and PSD systems in sublinear time, it does not rule out the possibility of sublinear-time PSD solvers.

\subsection{Paper Organization}

The remainder of this paper is organized as follows.
Section~\ref{sec:prelim} introduces some preliminary concepts and the \fpush and \bpush algorithms, and Section~\ref{sec:related_work} introduces more related work.
Next, Section~\ref{sec:algebra} elaborates on our formulation of $\x^{\ast}$ and $p$-norm gaps and proves Theorems~\ref{thm:x*} to \ref{thm:truncation_error}.
In Sections~\ref{sec:MC}, \ref{sec:push}, and \ref{sec:bidirectional}, we present our algorithms based on random-walk sampling, local push, and the bidirectional method, respectively, and prove Theorems~\ref{thm:MC_cubic} to \ref{thm:bidirectional_RCDD} along with some additional results.
After that, Sections~\ref{sec:PageRank} and \ref{sec:effective_resistance} relate our problem and results to local computation of PageRank and effective resistance, respectively.
Detailed discussion about the relationship between our \push primitive and the \fpush and \bpush algorithms is given in Section~\ref{sec:push_relationship}.
Finally, the appendix contains the deferred proofs and some additional results.

\section{Preliminaries} \label{sec:prelim}

\subsection{More Notations and Facts}

This subsection introduces some additional notations and basic facts from linear algebra.

For a matrix $\M \in \R^{n \times n}$ and an index $k \in [n]$, we use $\M(k,\cdot)$ to denote the $k$-th row vector of $\M$ and $\M(\cdot,k)$ to denote the $k$-th column vector of $\M$; we use $d_{\M}(k)$ to denote the diagonal entry $\M(k,k)$.
It holds that $\|\M\|_1 = \max_{k \in [n]}\sum_{j=1}^{n}\big|\M(j,k)\big|$, which is the maximum absolute column sum of $\M$, and $\|\M\|_{\infty} = \left\|\M^{\top}\right\|_{1}$, which is the maximum absolute row sum of $\M$.
When we use $\le$ and $\ge$ between two vectors, we mean entrywise comparison.

The pseudoinverse of $\M \in \R^{n \times n}$ satisfies $\M^{+} = \left(\M|_{\ker(\M)^{\perp}}\right)^{-1}\P_{\range(\M)}$, where $\M|_{\ker(\M)^{\perp}}$ is an invertible linear map from $\ker(\M)^{\perp}$ to $\range(\M)$ and $\P_{\range(\M)}$ is the orthogonal projection from $\R^n$ onto $\range(\M)$ (see, e.g., \cite[Section 6C]{axler2024linear}).

For a linear operator $\M: U \to U$ on a finite-dimensional complex vector space $U$, we use $\rho(\M)$ to denote its \textit{spectral radius}, which is the maximum modulus of its eigenvalues.
For any induced operator norm $\|\cdot\|$, it holds that $\rho(\M)\le\|\M\|$.
A \textit{Neumann series} is a series of the form $\sum_{\ell=0}^{\infty}\M^{\ell}$.
It is well-known that if $\rho(\M)<1$, then $\I-\M$ is invertible and the Neumann series $\sum_{\ell=0}^{\infty}\M^{\ell}$ converges to its inverse, i.e., $(\I-\M)^{-1} = \sum_{\ell=0}^{\infty}\M^{\ell}$.

We use $\oplus$ to denote the direct sum of subspaces.
One can verify that for any matrix $\M \in \R^{n \times n}$, $\ker(\M) \oplus \range(\M) = \R^n$ holds iff $\ker\left(\M^2\right) = \ker(\M)$.
For any real symmetric matrix $\S$, it holds that $\ker(\S) = \range(\S)^{\perp}$, $\rho(\S) = \|\S\|_2$, and $\S^{+}$ is real symmetric.

\subsection{\fpush and \bpush on Graphs} \label{sec:prelim_push}

\fpush~\cite{andersen2006local,andersen2007pagerank} and \bpush~\cite{andersen2007local,andersen2008local} are two local exploration algorithms for estimating random-walk probabilities on graphs.
They are originally developed for computing PPR and PageRank contribution vectors, respectively.\footnote{Their original names are \texttt{ApproximatePageRank} and \texttt{ApproxContributions}, respectively.}
For our purpose, we describe multi-level variants of both algorithms for PPR computation, which are similar to their original forms in spirit, but can be readily adapted to estimate multi-step random-walk probabilities and align better with our more general framework.
We refer interested readers to \cite{yang2024efficient} for a more detailed treatment of the original version of \fpush and \bpush.

\fpush (Algorithm~\ref{alg:FP}) performs forward exploration from a given source node $s \in V$ and estimates the PPR values from $s$ to other nodes.
It receives an extra parameter $L$ and maintains vectors $\p^{(\ell)}$ and $\r^{(\ell)}$ for $\ell = 0,1,\dots,L-1$, where $\p^{(\ell)}$'s are called \textit{reserve} vectors and $\r^{(\ell)}$'s are called \textit{residue} vectors.
The reserve vectors together serve as an estimate of $\vpi_{G,\alpha,\e_s}$, while the approximation errors are captured by the residues.

\fpush initializes the reserve and residue vectors to $\vzero$, except that $\r^{(0)}(s)$ is set to be the decay factor $\alpha$.
Next, the main loop iterates over levels $\ell$ from $0$ to $L-2$.
At each level $\ell$, the algorithm performs a local push operation on each node $v$ whose residue $\r^{(\ell)}(v)$ divided by its outdegree $\dout_G(v)$ exceeds the threshold $\rmax$ in absolute value.
The push operation on $v$ at level $\ell$ sets the reserve $\p^{(\ell)}(v)$ to $\r^{(\ell)}(v)$, increments $\r^{(\ell+1)}(v)$ by $\frac{1}{2}\r^{(\ell)}(v)$, increments $\r^{(\ell+1)}(u)$ by $\frac{1}{2}(1-\alpha) \cdot \frac{\A_G(v,u)}{\dout_G(v)} \cdot \r^{(\ell)}(v)$ for each out-neighbor $u$ of $v$, and sets $\r^{(\ell)}(v)$ to $0$.
Note that the third step in the push operation can be combinatorially interpreted as distributing the probability mass $\r^{(\ell)}(v)$ to the out-neighbors of $v$ according to the transition probability in a lazy random walk, which is the origin of the name ``local push.''

\begin{algorithm}[ht]
    \DontPrintSemicolon
    \caption{$\fpush(G,s,\alpha,L,\rmax)$~\cite{andersen2006local,andersen2007pagerank}} \label{alg:FP}
    \KwIn{oracle access to graph $G$, source node $s \in V$, decay factor $\alpha$, number of levels $L$, threshold $\rmax$}
    \KwOut{dictionaries $\p^{(\ell)}$ for reserves and $\r^{(\ell)}$ for residues, for $\ell = 0,1,\dots,L-1$}
    $\r^{(\ell)} \gets \vzero$, $\p^{(\ell)} \gets \vzero$ for $\ell=0,1,\dots,L-1$ \;
    $\r^{(0)}(s) \gets \alpha$ \;
    \For{$\ell$ \textup{from} $0$ \textup{to} $L-2$}{
		\For{\textup{each} $v$ \textup{with} $\frac{\r^{(\ell)}(v)}{\dout_G(v)} > \rmax$}{
			$\p^{(\ell)}(v) \gets \r^{(\ell)}(v)$ \;
            $\r^{(\ell+1)}(v) \gets \r^{(\ell+1)}(v) + \frac{1}{2} \r^{(\ell)}(v)$ \;
            \For{\textbf{\textup{each}} \textup{out-neighbor} $u$ \textup{of} $v$}
            {
                $\r^{(\ell+1)}(u) \gets \r^{(\ell+1)}(u) + \frac{1}{2}(1-\alpha) \cdot \frac{\A_G(v,u)}{\dout_G(v)} \cdot \r^{(\ell)}(v)$ \;
            }
			$\r^{(\ell)}(v) \gets 0$ \;
		}
	}
	\Return $\p^{(\ell)}$ and $\r^{(\ell)}$ for $\ell = 0,1,\dots,L-1$ \;
\end{algorithm}

\bpush (Algorithm~\ref{alg:BP}) adopts a similar framework to \fpush, but it performs backward exploration from a given target node $t \in V$ and estimates the PageRank contributions of other nodes to $t$.
That is, the reserve vectors together serve as an estimate of $\vpi^{-1}_{G,\alpha,t}$, while the approximation errors are captured by the residues.
The key differences lie in the conditions for performing the push operations and the way of distributing the residue mass.
Specifically, \bpush performs the push operation on a node $v$ at level $\ell$ if its residue $\r^{(\ell)}(v)$ exceeds the threshold $\rmax$, and in the push operation, it increments $\r^{(\ell+1)}(u)$ by $\frac{1}{2}(1-\alpha) \cdot \frac{\A_G(u,v)}{\dout_G(u)} \cdot \r^{(\ell)}(v)$ for each in-neighbor $u$ of $v$.

\begin{algorithm}[ht]
    \DontPrintSemicolon
    \caption{$\bpush(G,t,\alpha,L,\rmax)$~\cite{andersen2007local,andersen2008local}} \label{alg:BP}
    \KwIn{oracle access to graph $G$, target node $t \in V$, decay factor $\alpha$, number of levels $L$, threshold $\rmax$}
    \KwOut{dictionaries $\p^{(\ell)}$ for reserves and $\r^{(\ell)}$ for residues, for $\ell = 0,1,\dots,L-1$}
    $\r^{(\ell)} \gets \vzero$, $\p^{(\ell)} \gets \vzero$ for $\ell=0,1,\dots,L-1$ \;
    $\r^{(0)}(t) \gets \alpha$ \;
    \For{$\ell$ \textup{from} $0$ \textup{to} $L-2$}{
		\For{\textup{each} $v$ \textup{with} $\r^{(\ell)}(v) > \rmax$}{
			$\p^{(\ell)}(v) \gets \r^{(\ell)}(v)$ \;
            $\r^{(\ell+1)}(v) \gets \r^{(\ell+1)}(v) + \frac{1}{2} \r^{(\ell)}(v)$ \;
            \For{\textbf{\textup{each}} \textup{in-neighbor} $u$ \textup{of} $v$}
            {
                $\r^{(\ell+1)}(u) \gets \r^{(\ell+1)}(u) + \frac{1}{2}(1-\alpha) \cdot \frac{\A_G(u,v)}{\dout_G(u)} \cdot \r^{(\ell)}(v)$ \;
            }
			$\r^{(\ell)}(v) \gets 0$ \;
		}
	}
	\Return $\p^{(\ell)}$ and $\r^{(\ell)}$ for $\ell = 0,1,\dots,L-1$ \;
\end{algorithm}

We note that the original \fpush and \bpush algorithms do not partition reserve and residue vectors into multiple levels, which relies on the memoryless property of PPR.
Additionally, they can consider standard random walks instead of lazy random walks, which exploits the invertibility of the PPR and PageRank contribution equations.
Apart from these differences, our variants differ from the originals by at most a factor of $\alpha$ in the variable scaling.
Additionally, if one removes the factors of $\alpha$ and $1-\alpha$ in the two algorithms, they can be used to approximate multi-step random-walk transition probabilities~\cite{banerjee2015fast,cui2025mixing,yang2025improved}.

Although \fpush and \bpush share similar methodologies, they have been treated as two algorithms for different problems and analyzed separately.
A key component of their analysis involves invariant properties preserved by the push operations, which characterize the relationship among the reserve vectors, residue vectors, and the PPR or PageRank contribution vectors.
These invariants enable the derivation of bounds on their approximation errors and running times.

However, as summarized in \cite{yang2024efficient}, meaningful worst-case bounds on unweighted directed graphs are known only for \fpush's running time and \bpush's approximation error.
On unweighted undirected graphs, \fpush also provides degree-normalized accuracy guarantees, while \bpush enjoys worst-case running time bounds.
We do not elaborate on these properties since our unified analysis of the general \push primitive subsumes them.

As noted in \cite{yang2024efficient}, a unified understanding of \fpush and \bpush remains elusive despite their extensive study and applications.
In particular, it is unclear why these algorithms use different conditions for performing the push operations and exhibit different behaviors on directed graphs regarding accuracy and running time bounds.
In this work, we provide a unified perspective by viewing both algorithms as instances of the general \push primitive, explaining their different push strategies and properties.

\section{Other Related Work} \label{sec:related_work}

A vast literature exists on nearly-linear-time Laplacian solvers and their extensions, including solvers for undirected Laplacian/SDD systems (e.g., \cite{spielman2004nearly,spielman2014nearly,jambulapati2021ultrasparse}) and directed Laplacian/RDD/CDD systems (e.g., \cite{cohen2016faster,cohen2018solving,jambulapati2025eulerian}).
These solvers achieve nearly linear time complexity in $\nnz(\M)$ with polylogarithmic dependence on $1/\eps$ and the condition number.
The development of global RDD/CDD solvers relies on reductions from solving RDD/CDD systems to solving Eulerian Laplacian systems, combined with efficient methods for computing PPR vectors with small $\alpha$ and the stationary distribution of random walks on graphs~\cite{cohen2016faster}.
However, the global techniques from algorithmic linear algebra and known reductions to the Eulerian case do not directly apply to the local setting, revealing a fundamental distinction between global and local RDD/CDD solvers.
In fact, the lower bounds established in~\cite{andoni2018solving} and this work demonstrate that local SDD solvers require polynomial dependence on $1/\eps$ and $1/\gamma$, indicating a separation between global and local SDD solvers.

\cite{doron2017probabilistic} develops probabilistic logspace solvers for certain classes of directed Laplacian systems.
Their method also relies on approximating truncated Neumann series, but they bound the truncation error using spectral radius and Jordan normal form, which yield truncation parameters of at least $n^2$.
Such a huge truncation parameter makes their algorithm and analysis inapplicable to the sublinear-time setting.
As an aside, in the quantum regime, \cite{ta-shma2013inverting} presents algorithms for inverting well-conditioned matrices in quantum logspace, but their approaches are not directly applicable to our classical sublinear-time framework.

The idea of using random-walk sampling to solve linear systems dates back to the von Neumann-Ulam algorithm for approximating matrix inversion~\cite{forsythe1950matrix,wasow1952note}.
The bidirectional method for estimating random-walk probabilities on graphs is first proposed in \cite{lofgren2014fast}, which is inspired by property testing techniques~\cite{goldreich2011testing,kale2013noise} and later simplified by the \BiPPR framework~\cite{lofgren2016personalized}.

The bidirectional idea has been widely applied to compute PageRank~\cite{lofgren2015bidirectional,lofgren2016personalized,bressan2018sublinear,bressan2023sublinear,wei2024approximating,wang2024revisiting,bertram2025estimating,thorup2026pagerank}, effective resistance~\cite{cui2025mixing,yang2025improved}, heat kernel~\cite{bressan2018sublinear,bressan2023sublinear}, and Markov Chain transition probability~\cite{banerjee2015fast}.
Among them, \cite{wang2024revisiting} proves that the simple \BiPPR framework computes single-node PageRank on unweighted directed graphs in optimal time complexity (in terms of $n$ and $m$).
Their analysis relies on a new complexity bound of \bpush, which is parameterized by the PageRank value of the target node.
Recently, \cite{yang2025improved} shows that the bidirectional technique can yield faster algorithms for constructing effective resistance sketch (as defined in \cite{dwaraknath2023optimal}) on expander graphs, and \cite{thorup2026pagerank} combines random-walk sampling with a novel randomized local push technique to improve the complexity of estimating single-node PageRank on directed graphs with bounded in-degree.

\cite{shyamkumar2016sublinear} uses the bidirectional method to estimate a single element in the product of a matrix power and a vector, which relates to our estimation of $\frac{1}{2}\t^{\top}\sum_{\ell=0}^{L-1}\left(\frac{1}{2}\left(\I+\D_{\M}^{-1}\A_{\M}^{\top}\right)\right)^{\ell}\D_{\M}^{-1}\b$.
However, they only derive an average-case complexity bound under some bounded-norm conditions and discuss its applications to solving PSD systems.
In contrast, we conduct a more comprehensive study of this problem and apply it to solving RDD/CDD systems.

PageRank and PPR have been extensively studied and widely applied; we refer interested readers to the surveys \cite{langville2003survey,gleich2015pagerank,yang2024efficient}.
More lower bounds for PageRank computation can be found in \cite{bressan2023sublinear,yang2024efficient,bertram2025estimating} and references therein.
The recent work~\cite{bertram2025estimating} conducts a comprehensive study of various types of PPR estimation problems using different graph access queries under both worst-case and average-case settings.
For constant decay factors, they provide nearly tight complexity upper and lower bounds for achieving constant relative error guarantees when the target value is above a given threshold.

Effective resistance is ubiquitous in spectral graph theory~\cite{doyle1984random,lovasz1993random,spielman2011graph,jambulapati2025eulerian}.
A line of work~\cite{peng2021local,yang2023efficient,cui2025mixing,yang2025improved} focuses on locally estimating single-pair effective resistance through multi-step random-walk probabilities.
\cite{cai2023effective} studies this problem on non-expander graphs and establishes strong complexity lower bounds, though it does not explicitly give a lower bound on the spectral gap of the graph.
Recently, \cite{yang2025improved} provides a complexity lower bound on the relative error parameter for this problem.
Besides, a line of work \cite{li2023new,dwaraknath2023optimal,yang2025improved} studies the construction of effective resistance sketches.
To this end, \cite{li2023new} leverages random-walk sampling, \cite{dwaraknath2023optimal} uses count sketches and SDD solvers, and \cite{yang2025improved} employs a bidirectional approach.

\section{Formulation of $\x^{\ast}$ and the $p$-Norm Gaps} \label{sec:algebra}

In this section, we prove Theorems~\ref{thm:x*}, \ref{thm:p_norm_gap}, and \ref{thm:truncation_error} and give further explanations on our formulations of $\x^{\ast}$ and the $p$-norm gaps.

First, we prove the following lemma that upper bounds $\left\|\D_{\M}^{-1/q}\A_{\M}^{\top}\D_{\M}^{-1/p}\right\|_p$ for certain $\M$ and H\"older conjugates $p,q$.

\begin{lemma} \label{lem:norm_upper_bound}
	The following hold:
	\begin{enumerate}
		\item If $\M$ is RDD, then $\left\|\D_{\M}^{-1}\A_{\M}^{\top}\right\|_{\infty} \le 1$.
		\item If $\M$ is CDD, then $\left\|\A_{\M}^{\top}\D_{\M}^{-1}\right\|_1 \le 1$.
		\item If $\M$ is RCDD, then $\left\|\D_{\M}^{-1/q}\A_{\M}^{\top}\D_{\M}^{-1/p}\right\|_p \le 1$ for any H\"older conjugates $p,q$.
	\end{enumerate}
\end{lemma}

\begin{proof}[Proof of Lemma~\ref{lem:norm_upper_bound}]
First consider the case when $\M$ is RDD.
By our decomposition $\M = \D_{\M} - \A_{\M}^{\top}$ and the characterization of the $\|\cdot\|_{\infty}$ norm as the maximum absolute row sum, we have $\left\|\D_{\M}^{-1}\A_{\M}^{\top}\right\|_{\infty} = \max_{j \in [n]}\left\{\sum_{k \ne j}\frac{|\M(j,k)|}{d_{\M}(j)}\right\}$.
As $\M$ is RDD, for any $j \in [n]$, $\sum_{k \ne j}|\M(j,k)| \le d_{\M}(j)$, so $\sum_{k \ne j}\frac{|\M(j,k)|}{d_{\M}(j)} \le 1$, implying that $\left\|\D_{\M}^{-1}\A_{\M}^{\top}\right\|_{\infty} \le 1$.
If $\M$ is CDD, then $\M^{\top}$ is RDD, and the first statement implies that $\left\|\A_{\M}^{\top}\D_{\M}^{-1}\right\|_1 = \left\|\D_{\M}^{-1}\A_{\M}\right\|_{\infty} \le 1$.
This proves the first two statements.

Next, we focus on the third statement.
For any $1 < p < \infty$ and $\x \in \R^n$, we have
\begin{align*}
	\left\|\D_{\M}^{-1/q}\A_{\M}^{\top}\D_{\M}^{-1/p}\x\right\|_p^p & = \sum_{j=1}^{n}\left|\sum_{k=1}^{n}\frac{\A_{\M}^{\top}(j,k)\x(k)}{d_{\M}(j)^{1/q}d_{\M}(k)^{1/p}}\right|^p.
\end{align*}
By H\"older's inequality, for any $j \in [n]$,
\begin{align*}
	\left|\sum_{k=1}^{n}\frac{\A_{\M}^{\top}(j,k)\x(k)}{d_{\M}(j)^{1/q}d_{\M}(k)^{1/p}}\right|
	& = \left|\sum_{k=1}^{n}\frac{\A_{\M}^{\top}(j,k)^{1/q}}{d_{\M}(j)^{1/q}}\cdot\frac{\A_{\M}^{\top}(j,k)^{1/p}\x(k)}{d_{\M}(k)^{1/p}}\right| \\
	& \le \left(\sum_{k=1}^{n}\frac{\left|\A_{\M}^{\top}(j,k)\right|}{d_{\M}(j)}\right)^{1/q}\left(\sum_{k=1}^{n}\frac{\left|\A_{\M}^{\top}(j,k)\right|\big|\x(k)\big|^p}{d_{\M}(k)}\right)^{1/p} \\
	& \le \left(\sum_{k=1}^{n}\frac{\left|\A_{\M}^{\top}(j,k)\right|\big|\x(k)\big|^p}{d_{\M}(k)}\right)^{1/p},
\end{align*}
where we used $\sum_{k=1}^{n}\frac{\left|\A_{\M}^{\top}(j,k)\right|}{d_{\M}(j)} \le 1$.
Substituting, we obtain
\begin{align*}
	& \phantom{{}={}} \left\|\D_{\M}^{-1/q}\A_{\M}^{\top}\D_{\M}^{-1/p}\x\right\|_p^p \le \sum_{j=1}^{n}\sum_{k=1}^{n}\frac{\left|\A_{\M}^{\top}(j,k)\right|\big|\x(k)\big|^p}{d_{\M}(k)} \\
	& = \sum_{k=1}^{n}\big|\x(k)\big|^p\sum_{j=1}^{n}\frac{\left|\A_{\M}^{\top}(j,k)\right|}{d_{\M}(k)} \le \sum_{k=1}^{n}\big|\x(k)\big|^p = \|\x\|_p^p,
\end{align*}
where we used $\sum_{j=1}^{n}\frac{\left|\A_{\M}^{\top}(j,k)\right|}{d_{\M}(k)} \le 1$.
This result implies that $\left\|\D_{\M}^{-1/q}\A_{\M}^{\top}\D_{\M}^{-1/p}\right\|_p \le 1$.
\end{proof}

To establish our formulation of $\x^{\ast}$ in Theorem~\ref{thm:x*}, we use the following lemma, which characterizes the operator norm and Neumann series of certain restricted linear maps.

\begin{lemma} \label{lem:restriction_norm}
	Suppose $\X \in \R^{n\times n}$ and $\|\cdot\|$ is the operator norm induced by some vector norm $\|\cdot\|$.
	If $\|\X\| \le 1$, then, for $\bar{\X} := \frac{1}{2}(\I+\X)$, we have $\left\|\left.\bar{\X}\right|_{\range(\I-\X)}\right\| < 1$ and
	\begin{align*}
		\left(\left.(\I-\X)\right|_{\range(\I-\X)}\right)^{-1} = \frac{1}{2}\left(\left.\left(\I-\bar{\X}\right)\right|_{\range(\I-\X)}\right)^{-1} = \frac{1}{2}\sum_{\ell=0}^{\infty}\left(\left.\bar{\X}\right|_{\range(\I-\X)}\right)^{\ell}.
	\end{align*}
\end{lemma}

\begin{proof}
We first prove that $\R^n = \ker(\I-\X)\oplus\range(\I-\X)$, or equivalently $\ker\left((\I-\X)^2\right) = \ker(\I-\X)$.
Assume for contradiction that $\ker\left((\I-\X)^2\right) \supsetneq \ker(\I-\X)$, so there exists a vector $\u \in \R^n$ such that $(\I-\X)^2\u = \vzero$ and $(\I-\X)\u \ne \vzero$.
Letting $\v := (\I-\X)\u$, we have $\X\u = \u-\v$ and $\X\v = \v$.
By induction, for each integer $k \ge 1$, it holds that $\X^k\u = \u-k\v$.
Therefore,
\begin{align*}
	\lim_{k\to\infty}\left\|\X^k\u\right\| = \lim_{k\to\infty}\left\|\u-k\v\right\| \ge \lim_{k\to\infty}\big(k\|\v\|-\|\u\|\big) = \infty,
\end{align*}
where we used the reverse triangle inequality and $\v\ne\vzero$.
However, since $\|\X\| \le 1$ it holds that $\left\|\X^k\u\right\| \le \|\X\|^k\|\u\|\le\|\u\|$ for all integer $k \ge 1$, which means that $\lim_{k\to\infty}\left\|\X^k\u\right\|$ cannot be unbounded, yielding a contradiction.

Now, to prove that $\left\|\left.\bar{\X}\right|_{\range(\I-\X)}\right\| < 1$, since $\|\bar{\X}\|\le\frac{1}{2}\big(\|\I\|+\|\X\|\big) \le 1$, it suffices to show that no nonzero vector $\v \in \range(\I-\X)$ satisfies $\|\bar{\X}\v\| = \|\v\|$.
Assume for contradiction that such a $\v$ exists.
It holds that
\begin{align*}
	\|\bar{\X}\v\| = \left\|\frac{1}{2}\left(\I+\X\right)\v\right\| = \frac{1}{2}\left\|\v+\X\v\right\| \le \frac{1}{2}\big(\|\v\|+\|\X\v\|\big) \le \|\v\|,
\end{align*}
so the inequalities above must be equalities.
Hence, $\v$ must satisfy $\X\v = \lambda\v$ for some $\lambda \ge 0$ (by the equality condition for the triangle inequality) and $\|\X\v\| = \|\v\|$.
This forces $\lambda = 1$ and thus $\X\v = \v$, or $\v \in \ker(\I-\X)$.
However, we have proved that $\R^n = \ker(\I-\X)\oplus\range(\I-\X)$, so no such nonzero $\v \in \range(\I-\X)$ exists, yielding a contradiction.

To prove the second statement, observe that $\range(\I-\bar{\X}) = \range\left(\frac{1}{2}(\I-\X)\right) = \range(\I-\X)$ and $\range(\I-\bar{\X})$ is a $\bar{\X}$-invariant subspace, which implies that $\range(\I-\X)$ is a $\bar{\X}$-invariant subspace, i.e., $\left.\bar{\X}\right|_{\range(\I-\X)}$ maps $\range(\I-\X)$ to itself.
Thus, we have $\rho\left(\left.\bar{\X}\right|_{\range(\I-\X)}\right) \le \left\|\left.\bar{\X}\right|_{\range(\I-\X)}\right\| < 1$, so $\left.\left(\I-\bar{\X}\right)\right|_{\range(\I-\X)}$ is invertible and its inverse is given by the Neumann series $\sum_{\ell=0}^{\infty}\left(\left.\bar{\X}\right|_{\range(\I-\X)}\right)^{\ell}$.
Now the lemma follows since $\I-\X = 2(\I-\bar{\X})$.
\end{proof}

\begin{proof}[Proof of Theorem~\ref{thm:x*}]
First assume that $\M$ is RDD.
Applying Lemmas~\ref{lem:norm_upper_bound} and \ref{lem:restriction_norm}, we have $\left\|\D_{\M}^{-1}\A_{\M}^{\top}\right\|_{\infty} \le 1$ and
\begin{align*}
	\left(\left.\left(\I-\D_{\M}^{-1}\A_{\M}^{\top}\right)\right|_{\range\left(\I-\D_{\M}^{-1}\A_{\M}^{\top}\right)}\right)^{-1} = \frac{1}{2}\sum_{\ell=0}^{\infty}\left(\left.\frac{1}{2}\left(\I+\D_{\M}^{-1}\A_{\M}^{\top}\right)\right|_{\range\left(\I-\D_{\M}^{-1}\A_{\M}^{\top}\right)}\right)^{\ell}.
\end{align*}
As $\b \in \range(\M) = \range\left(\D_{\M}-\A_{\M}^{\top}\right)$, we have $\D_{\M}^{-1}\b \in \range\left(\I-\D_{\M}^{-1}\A_{\M}^{\top}\right)$, so $\x^{\ast}$ converges to
\begin{align*}
	\frac{1}{2}\sum_{\ell=0}^{\infty}\left(\frac{1}{2}\left(\I+\D_{\M}^{-1}\A_{\M}^{\top}\right)\right)^{\ell}\D_{\M}^{-1}\b & = \frac{1}{2}\sum_{\ell=0}^{\infty}\left(\left.\frac{1}{2}\left(\I+\D_{\M}^{-1}\A_{\M}^{\top}\right)\right|_{\range\left(\I-\D_{\M}^{-1}\A_{\M}^{\top}\right)}\right)^{\ell}\D_{\M}^{-1}\b \\
	& = \left(\left.\left(\I-\D_{\M}^{-1}\A_{\M}^{\top}\right)\right|_{\range\left(\I-\D_{\M}^{-1}\A_{\M}^{\top}\right)}\right)^{-1}\D_{\M}^{-1}\b.
\end{align*}
Thus, we can check that
\begin{align*}
	\M\x^{\ast} = \D_{\M}\left(\I-\D_{\M}^{-1}\A_{\M}^{\top}\right)\left(\left.\left(\I-\D_{\M}^{-1}\A_{\M}^{\top}\right)\right|_{\range\left(\I-\D_{\M}^{-1}\A_{\M}^{\top}\right)}\right)^{-1}\D_{\M}^{-1}\b = \D_{\M}\D_{\M}^{-1}\b = \b.
\end{align*}

Next, observe that for any H\"older conjugates $p,q$, we have
\begin{align*}
	& \phantom{{}={}} \D_{\M}^{-1/p}\left(\frac{1}{2}\left(\I+\D_{\M}^{-1/q}\A_{\M}^{\top}\D_{\M}^{-1/p}\right)\right)^{\ell}\D_{\M}^{-1/q} \\
	& = \D_{\M}^{-1/p}\left(\frac{1}{2}\D_{\M}^{1/p}\left(\I+\D_{\M}^{-1}\A_{\M}^{\top}\right)\D_{\M}^{-1/p}\right)^{\ell}\D_{\M}^{-1/q} = \left(\frac{1}{2}\left(\I+\D_{\M}^{-1}\A_{\M}^{\top}\right)\right)^{\ell}\D_{\M}^{-1}
\end{align*}
for each $\ell \ge 0$, so
\begin{align*}
	\x^{\ast} = \frac{1}{2}\D_{\M}^{-1/p}\sum_{\ell=0}^{\infty}\left(\frac{1}{2}\left(\I+\D_{\M}^{-1/q}\A_{\M}^{\top}\D_{\M}^{-1/p}\right)\right)^{\ell}\D_{\M}^{-1/q}\b
\end{align*}
for any H\"older conjugates $p,q$.

When $\M$ is CDD, we consider the expression $\x^{\ast} = \frac{1}{2}\D_{\M}^{-1}\sum_{\ell=0}^{\infty}\left(\frac{1}{2}\left(\I+\A_{\M}^{\top}\D_{\M}^{-1}\right)\right)^{\ell}\b$.
We have $\left\|\A_{\M}^{\top}\D_{\M}^{-1}\right\|_1 \le 1$ by Lemma~\ref{lem:norm_upper_bound}, so applying the above arguments with $\b \in \range(\M) = \range\left(\I-\A_{\M}^{\top}\D_{\M}^{-1}\right)$ shows that $\x^{\ast}$ converges to
\begin{align*}
	\D_{\M}^{-1}\left(\left.\left(\I-\A_{\M}^{\top}\D_{\M}^{-1}\right)\right|_{\range\left(\I-\A_{\M}^{\top}\D_{\M}^{-1}\right)}\right)^{-1}\b
\end{align*}
and
\begin{align*}
	\M\x^{\ast} = \left(\I-\A_{\M}^{\top}\D_{\M}^{-1}\right) \D_{\M} \D_{\M}^{-1} \left(\left.\left(\I-\A_{\M}^{\top}\D_{\M}^{-1}\right)\right|_{\range\left(\I-\A_{\M}^{\top}\D_{\M}^{-1}\right)}\right)^{-1}\b = \b.
\end{align*}
So we have proved that $\x^{\ast}$ is well-defined and satisfies $\M\x^{\ast} = \b$ when $\M$ is either RDD or CDD.

Next, assume that $\M$ is SDD.
Then $\M$ is RCDD and Lemma~\ref{lem:norm_upper_bound} implies that $\left\|\tA_{\M}^{\top}\right\|_2 \le 1$.
Repeating the above arguments with $\D_{\M}^{-1/2}\b \in \range\left(\I-\tA_{\M}^{\top}\right)$ gives
\begin{align*}
	\x^{\ast} = \D_{\M}^{-1/2}\left(\left.\left(\I-\tA_{\M}^{\top}\right)\right|_{\range\left(\I-\tA_{\M}^{\top}\right)}\right)^{-1}\D_{\M}^{-1/2}\b = \D_{\M}^{-1/2}\left(\left.\tM\right|_{\range(\tM)}\right)^{-1}\D_{\M}^{-1/2}\b.
\end{align*}
Since $\M$ is symmetric, $\tM$ is also symmetric, so $\range(\tM) = \ker(\tM)^{\perp}$.
By the property of the pseudoinverse, $\x^{\ast} = \D_{\M}^{-1/2}\left(\left.\tM\right|_{\ker(\tM)^{\perp}}\right)^{-1}\D_{\M}^{-1/2}\b = \D_{\M}^{-1/2}\tM^{+}\D_{\M}^{-1/2}\b$, finishing the proof.
\end{proof}

Recall that if $\M$ is RDD, then $1 - \left\|\D_{\M}^{-1}\A_{\M}^{\top}\right\|_{\infty} = \min_{j \in [n]}\left\{\frac{d_{\M}(j)-\sum_{k \ne j}|\M(j,k)|}{d_{\M}(j)}\right\} \ge 0$, and this quantity measures how strongly the diagonal entries dominate the off-diagonal entries in each row.
However, this quantity can equal zero, making it unsuitable as a useful notion of ``gap.''
In contrast, our Theorem~\ref{thm:p_norm_gap} shows that the maximum $p$-norm gap $\gamma_{\max}(\M)$ is always strictly positive when $\M$ is RDD/CDD.
As an example, $\gamma_{\infty}(\M) = 1 - \left\|\left.\frac{1}{2}\left(\I+\D_{\M}^{-1}\A_{\M}^{\top}\right)\right|_{\range\left(\I-\D_{\M}^{-1}\A_{\M}^{\top}\right)}\right\|_{\infty}$, which refines the quantity $1 - \left\|\D_{\M}^{-1}\A_{\M}^{\top}\right\|_{\infty}$ by replacing $\D_{\M}^{-1}\A_{\M}^{\top}$ with $\frac{1}{2}\left(\I+\D_{\M}^{-1}\A_{\M}^{\top}\right)$ and restricting the operator to $\range\left(\I-\D_{\M}^{-1}\A_{\M}^{\top}\right)$.
As Lemma~\ref{lem:restriction_norm} has characterized the advantages of this refinement, we can prove Theorem~\ref{thm:p_norm_gap} as follows.

\begin{proof}[Proof of Theorem~\ref{thm:p_norm_gap}]
By Lemmas~\ref{lem:norm_upper_bound} and \ref{lem:restriction_norm} and the definition of the $p$-norm gaps (Equation~\eqref{eqn:p_norm_gap}), we immediately have the following:
\begin{enumerate}
	\item If $\M$ is RDD, then $\gamma_{\infty}(\M) > 0$.
	\item If $\M$ is CDD, then $\gamma_1(\M) > 0$.
	\item If $\M$ is RCDD, then $\gamma_p(\M) > 0$ for any $p \in [1,\infty]$.
\end{enumerate}
Also, the definition implies that $\gamma_p(\M) \le 1$ for any $p \in [1,\infty]$.
This proves the first statement.

Next, assume that $\M$ is SDD.
For any H\"older conjugates $p,q$, it holds that
\begin{align*}
	& \phantom{{}={}} \rho\left(\left.\frac{1}{2}\left(\I+\tA_{\M}^{\top}\right)\right|_{\range\left(\I-\tA_{\M}^{\top}\right)}\right) = \rho\left(\left.\frac{1}{2}\left(\I+\D_{\M}^{-1/q}\A_{\M}^{\top}\D_{\M}^{-1/p}\right)\right|_{\range\left(\I-\D_{\M}^{-1/q}\A_{\M}^{\top}\D_{\M}^{-1/p}\right)}\right) \\
	& \le \left\|\left.\frac{1}{2}\left(\I+\D_{\M}^{-1/q}\A_{\M}^{\top}\D_{\M}^{-1/p}\right)\right|_{\range\left(\I-\D_{\M}^{-1/q}\A_{\M}^{\top}\D_{\M}^{-1/p}\right)}\right\|_p = 1-\gamma_p(\M)
\end{align*}
where the first equality uses the property of similar operators and the last equality uses the definition of $\gamma_p(\M)$.
We also have
\begin{align*}
	\rho\left(\left.\frac{1}{2}\left(\I+\tA_{\M}^{\top}\right)\right|_{\range\left(\I-\tA_{\M}^{\top}\right)}\right) = \left\|\left.\frac{1}{2}\left(\I+\tA_{\M}^{\top}\right)\right|_{\range\left(\I-\tA_{\M}^{\top}\right)}\right\|_2 = 1-\gamma_2(\M)
\end{align*}
by the symmetry of $\tA_{\M}^{\top}$.
These yield $\gamma_2(\M) \ge \gamma_p(\M)$ for all $p \in [1,\infty]$, so $\gamma_{\max}(\M) = \gamma_2(\M)$.

As $\rho\left(\tA_{\M}^{\top}\right) = \left\|\tA_{\M}^{\top}\right\|_2 \le 1$, all eigenvalues of $\frac{1}{2}\left(\I+\tA_{\M}^{\top}\right)$ lie in $[0,1]$.
Also, $\range\left(\I-\tA_{\M}^{\top}\right) = \ker\left(\I-\tA_{\M}^{\top}\right)^{\perp}$ and $\ker\left(\I-\tA_{\M}^{\top}\right)$ is the eigenspace of $\tA_{\M}^{\top}$ associated with eigenvalue $1$.
Combining these facts, we have
\begin{align*}
	\gamma_2(\M) = 1-\rho\left(\left.\frac{1}{2}\left(\I+\tA_{\M}^{\top}\right)\right|_{\ker\left(\I-\tA_{\M}^{\top}\right)^{\perp}}\right) = 1-\frac{1}{2}(1+\lambda) = \frac{1}{2}(1-\lambda),
\end{align*}
where $\lambda$ is the largest eigenvalue of $\tA_{\M}^{\top}$ other than $1$.
As $\tM = \I - \tA_{\M}^{\top}$, $1-\lambda$ equals the smallest nonzero eigenvalue of $\tM$, i.e., $1-\lambda = 2\gamma(\M)$.
This gives $\gamma_2(\M) = \gamma(\M)$ and finishes the proof.
\end{proof}

Although the $p$-norm gaps only involve operator norms for the restricted linear maps, the following proof of Theorem~\ref{thm:truncation_error} shows that they are sufficient to bound the truncation error between $\t^{\top}\x^{\ast}_L$ and $\t^{\top}\x^{\ast}$.

\begin{proof}[Proof of Theorem~\ref{thm:truncation_error}]
Suppose $p \in [1,\infty]$ satisfies $\gamma_p(\M) = \gamma_{\max}(\M)$ and $q$ is conjugate to $p$.
Then by the definitions of $\x^{\ast}$ and $\x^{\ast}_L$, the truncation error $\left|\t^{\top}\x^{\ast}_L - \t^{\top}\x^{\ast}\right|$ can be upper bounded by
\begin{align*}
	& \phantom{{}={}} \left|\t^{\top}\x^{\ast}_L - \t^{\top}\x^{\ast}\right| = \left|\frac{1}{2}\t^{\top}\D_{\M}^{-1/p}\sum_{\ell=L}^{\infty}\left(\frac{1}{2}\left(\I+\D_{\M}^{-1/q}\A_{\M}^{\top}\D_{\M}^{-1/p}\right)\right)^{\ell}\D_{\M}^{-1/q}\b\right| \\
	& \le \frac{1}{2}\left\|\D_{\M}^{-1/p}\t\right\|_q \left\|\sum_{\ell=L}^{\infty}\left(\frac{1}{2}\left(\I+\D_{\M}^{-1/q}\A_{\M}^{\top}\D_{\M}^{-1/p}\right)\right)^{\ell}\D_{\M}^{-1/q}\b\right\|_p \\
	& \le \frac{1}{2}\left\|\D_{\M}^{-1/p}\t\right\|_q \sum_{\ell=L}^{\infty}\left\|\left(\frac{1}{2}\left(\I+\D_{\M}^{-1/q}\A_{\M}^{\top}\D_{\M}^{-1/p}\right)\right)^{\ell}\D_{\M}^{-1/q}\b\right\|_p \\
	& \le \frac{1}{2}\left\|\D_{\M}^{-1/p}\t\right\|_q \sum_{\ell=L}^{\infty}\left\|\left.\frac{1}{2}\left(\I+\D_{\M}^{-1/q}\A_{\M}^{\top}\D_{\M}^{-1/p}\right)\right|_{\range\left(\I-\D_{\M}^{-1/q}\A_{\M}^{\top}\D_{\M}^{-1/p}\right)}\right\|_p^{\ell} \left\|\D_{\M}^{-1/q}\b\right\|_p \\
	& = \frac{1}{2}\left\|\D_{\M}^{-1/p}\t\right\|_q \sum_{\ell=L}^{\infty}\big(1-\gamma_p(\M)\big)^{\ell}\left\|\D_{\M}^{-1/q}\b\right\|_p,
\end{align*}
where we used H\"older's inequality, Minkowski's inequality, $\D_{\M}^{-1/q}\b \in \range\left(\I-\D_{\M}^{-1/q}\A_{\M}^{\top}\D_{\M}^{-1/p}\right)$, and the definition of $\gamma_p(\M)$.
As Theorem~\ref{thm:p_norm_gap} establishes that $\gamma_p(\M) = \gamma_{\max}(\M) \in (0,1]$, we have $\sum_{\ell=L}^{\infty}\big(1-\gamma_p(\M)\big)^{\ell} = \frac{1}{\gamma_p(\M)}\big(1-\gamma_p(\M)\big)^L \le \frac{1}{\gamma_p(\M)} \cdot e^{-\gamma_p(\M)\cdot L}$.
On the other hand, $\left\|\D_{\M}^{-1/p}\t\right\|_q$ can be upper bounded by $\left\|\D_{\M}^{1/q}\right\|_q \left\|\D_{\M}^{-1}\t\right\|_q \le \dmax(\M)^{1/q} \|\t\|_0^{1/q} \left\|\D_{\M}^{-1}\t\right\|_{\infty} \le \dmax(\M)^{1/q} \|\t\|_0 \left\|\D_{\M}^{-1}\t\right\|_{\infty}$.
Similarly, $\left\|\D_{\M}^{-1/q}\b\right\|_p \le \dmax(\M)^{1/p} \|\b\|_0^{1/p} \left\|\D_{\M}^{-1}\b\right\|_{\infty} \le \dmax(\M)^{1/p} \|\b\|_0 \left\|\D_{\M}^{-1}\b\right\|_{\infty}$.
Combining these bounds with $0 < \gamma \le \gamma_p(\M)$ yields
\begin{align*}
	\left|\t^{\top}\x^{\ast}_L - \t^{\top}\x^{\ast}\right| & \le \frac{1}{2\gamma} \cdot e^{-\gamma L} \cdot \dmax(\M) \|\t\|_0 \left\|\D_{\M}^{-1}\t\right\|_{\infty} \|\b\|_0 \left\|\D_{\M}^{-1}\b\right\|_{\infty}.
\end{align*}
Hence, setting $L := \Theta\left(\frac{1}{\gamma}\log\left(\frac{1}{\gamma\eps} \cdot \dmax(\M)\|\t\|_0\left\|\D_{\M}^{-1}\t\right\|_{\infty}\|\b\|_0\left\|\D_{\M}^{-1}\b\right\|_{\infty}\right)\right)$ with proper constants guarantees that $\left|\t^{\top}\x^{\ast}_L - \t^{\top}\x^{\ast}\right| \le \frac{1}{2}\eps$, proving the theorem.
\end{proof}

\paragraph{Remark.}
It follows from the proof that, if the condition of Theorem~\ref{thm:truncation_error} is modified to $0 < \gamma \le \gamma_p(\M)$ for $p = 1$ or $\infty$, we can eliminate the $\|\t\|_0$ or $\|\b\|_0$ term from the setting of $L$, respectively.
Also, for certain forms of $\b$ and $\t$, it is possible to improve the setting of $L$ via a more careful analysis, as we will do for effective resistance computation.

\section{Random-Walk Sampling} \label{sec:MC}

This section presents a Monte Carlo algorithm for estimating $\t^{\top}\x^{\ast}$ via random-walk sampling.
All our algorithms in this paper aim to estimate $\t^{\top}\x^{\ast}$ by approximating $\t^{\top}\x^{\ast}_L$.
By Theorem~\ref{thm:truncation_error} and our setting of $L$, it suffices to ensure that the estimate $\hat{x}$ satisfies that $\left|\hat{x}-\t^{\top}\x^{\ast}_L\right|$ is at most half the desired accuracy guarantee with probability at least $3/4$, and we will omit this matter in the following proofs.

We first focus on RDD systems and transfer the results to CDD systems by transposing $\M$ and the expression of $\t^{\top}\x^{\ast}_L$ at the end of this section.
When $\M$ is RDD, $\frac{1}{2}\left(\I+\D_{\M}^{-1}\left|\A_{\M}^{\top}\right|\right)$ is row substochastic and we can estimate
\begin{align*}
	\t^{\top}\x^{\ast}_L = \frac{1}{2}\t^{\top}\sum_{\ell=0}^{L-1}\left(\frac{1}{2}\left(\I+\D_{\M}^{-1}\A_{\M}^{\top}\right)\right)^{\ell}\D_{\M}^{-1}\b
\end{align*}
by generating random-walk sampling from $|\t|/\|\t\|_1$.
See Algorithm~\ref{alg:RDD_MC_truncated} for a pseudocode.
In each sample, we first sample a walk length $\ell$ from $[0,L-1]$ uniformly at random and sample a source coordinate $v$ from the distribution $|\t|/\|\t\|_1$.
Then we simulate a lazy random walk for $\ell$ steps starting from $v$ according to the transition matrix $\frac{1}{2}\left(\I+\D_{\M}^{-1}\left|\A_{\M}^{\top}\right|\right)$.
Specifically, at each step from $v$, with probability $\frac{1}{2}$ the walk stays put at $v$, with probability $\frac{|\A_{\M}(u,v)|}{2d_{\M}(v)}$ the walk moves to each $u \in [n]$, and with the remaining probability $\frac{1}{2}-\sum_{u \in [n]}\frac{|\A_{\M}(u,v)|}{2d_{\M}(v)}$ the walk terminates.
Note that these probabilities lie in $[0,1]$ since $\M$ is RDD.
Additionally, we keep track of the product of the signs of the initial entry in $\t$ and the entries in $\A_{\M}$ along the walk (where stay-put steps have sign $1$), which is denoted as $\sigma$ in the pseudocode.
After $\ell$ steps, if the walk has not terminated and is at coordinate $v$, we take the value $\sigma \cdot \frac{1}{2} \|\t\|_1 \cdot \frac{\b(v)}{d_{\M}(v)} \cdot L$ as the estimate of this sample.
We repeat this process for $\ns$ independent samples and return the average as the final estimate $\hat{x}$.

We emphasize that our sampling scheme is different from the framework in \cite{andoni2019solving}, in that we first sample the walk length $\ell$ and then perform $\ell$ steps of the random walk, while they perform $L$ steps in each sample and take the quantities obtained in each step into account.
As it turns out, our scheme is easier to analyze and will save a factor of $\log L$ in the number of samples.

\begin{algorithm}[ht]
	\DontPrintSemicolon
	\caption{Random-Walk Sampling for RDD Systems} \label{alg:RDD_MC_truncated}
	\KwIn{oracle access to $\M$, $\b$, and $\t$, truncation parameter $L$, number of samples $\ns$}
	\KwOut{estimate $\hat{x}$ of $\t^{\top}\x^{\ast}$}
	$\hat{x} \gets 0$ \;
	\For{$j$ \textup{from} $1$ \textup{to} $\ns$}{
		$\ell \gets \text{uniformly random sample from }[0,L-1]$ \;
		$v \gets \text{random sample from the distribution }|\t|/\|\t\|_1$ \;
		$\sigma \gets \sgn\big(\t(v)\big) $ \;
		\For{$k$ \textup{from} $1$ \textup{to} $\ell$}{
			simulate one step of the random walk from one of the following three possibilities: \;
			$\quad$ 1. w.p. $\frac{1}{2}$, $v' \gets v$ \textcolor{gray}{// stays put at $v$} \;
			$\quad$ 2. w.p. $\frac{|\A_{\M}(u,v)|}{2d_{\M}(v)}$ for each $u \in [n]$, $v' \gets u$, $\sigma \gets \sigma \cdot \text{sgn}\big(\A_{\M}(u,v)\big)$ \textcolor{gray}{// moves to $u$} \;
			$\quad$ 3. w.p. $\frac{1}{2} - \sum_{u \in [n]}\frac{|\A_{\M}(u,v)|}{2d_{\M}(v)}$, $\sigma \gets 0$, break the loop over $k$ \textcolor{gray}{// terminates} \;
			$v \gets v'$ \;
		}
		$\hat{x} \gets \hat{x} + \frac{1}{\ns} \cdot \sigma \cdot \frac{1}{2} \|\t\|_1 \cdot \frac{\b(v)}{d_{\M}(v)} \cdot L$ \;
	}
	\Return $\hat{x}$ \;
\end{algorithm}

We first establish the unbiasedness of the sampling scheme.

\begin{lemma} \label{lem:MC_unbiased}
	Each sample described above gives an unbiased estimate of $\t^{\top}\x^{\ast}_L$.
\end{lemma}

\begin{proof}
By our sampling scheme, the result of each sample equals
\begin{align*}
	\sum_{\ell=0}^{L-1} \sum_{v \in [n]} \sum_{\sigma \in \{-1,1\}} X(\ell,v,\sigma) \cdot \sigma \cdot \frac{1}{2} \|\t\|_1 \cdot \frac{\b(v)}{d_{\M}(v)} \cdot L,
\end{align*}
where $X(\ell,v,\sigma)$ is the indicator random variable of the event that ``$\ell$ is chosen as the walk length, the lazy random-walk sampling is at node $v$ after $\ell$ steps, and $\sigma$ equals the product of $\sgn\big(\t(u)\big)$ and the sign product along the walk path, where $u$ is the sampled source coordinate'' for each $\ell \in [0,L-1]$, $v \in [n]$, and $\sigma \in \{-1,1\}$.
By our random-walk transition rule, we have
\begin{align*}
	\E\left[\sum_{\sigma \in \{-1,1\}} X(\ell,v,\sigma) \cdot \sigma\right] = \frac{1}{L} \left(\frac{\t}{\|\t\|_1}\right)^{\top} \left(\frac{1}{2}\left(\I+\D_{\M}^{-1}\A_{\M}^{\top}\right)\right)^{\ell}\e_v
\end{align*}
for each $\ell \in [0,L-1]$ and $v \in [n]$.
Thus, the expectation of each sampled value equals
\begin{align*}
	& \phantom{{}={}} \sum_{\ell=0}^{L-1} \sum_{v \in [n]} \frac{1}{L} \left(\frac{\t}{\|\t\|_1}\right)^{\top} \left(\frac{1}{2}\left(\I+\D_{\M}^{-1}\A_{\M}^{\top}\right)\right)^{\ell}\e_v \cdot \frac{1}{2} \|\t\|_1 \cdot \frac{\b(v)}{d_{\M}(v)} \cdot L \\
	& = \frac{1}{2} \t^{\top} \sum_{\ell=0}^{L-1} \left(\frac{1}{2}\left(\I+\D_{\M}^{-1}\A_{\M}^{\top}\right)\right)^{\ell} \D_{\M}^{-1}\b = \t^{\top}\x^{\ast}_L,
\end{align*}
which proves the lemma.
\end{proof}

We can now prove Theorem~\ref{thm:MC_cubic} by applying the Hoeffding bound.

\begin{proof}[Proof of Theorem~\ref{thm:MC_cubic}]
We use random-walk sampling as in Algorithm~\ref{alg:RDD_MC_truncated}.
The algorithm returns $\hat{x}$ as the average of $\ns$ independent samples, where each sampled value has absolute value at most $\frac{1}{2} \|\t\|_1 \left\|\D_{\M}^{-1}\b\right\|_{\infty} L$.
Therefore, by Lemma~\ref{lem:MC_unbiased} and the Hoeffding bound (Theorem~\ref{thm:hoeffding}), $\Pr\left\{\left|\hat{x} - \t^{\top}\x^{\ast}_L\right| \ge \frac{1}{2} \eps \left\|\D_{\M}^{-1}\b\right\|_{\infty}\right\}$ is upper bounded by
\begin{align*}
	2\exp\left(-\frac{2\ns\left(\frac{1}{2}\eps \left\|\D_{\M}^{-1}\b\right\|_{\infty}\right)^2}{\left(\|\t\|_1 \left\|\D_{\M}^{-1}\b\right\|_{\infty} L\right)^2}\right) = 2\exp\left(-\frac{\ns \cdot \eps^2}{2 \|\t\|_1^2 L^2}\right).
\end{align*}
To guarantee that this probability is at most $1/4$, we set $\ns := \Theta\left(\|\t\|_1^2 L^2/\eps^2\right)$.

As each sample simulates at most $L$ steps of random walk, and each step takes $O\big(\frow(\M)\big)$ time, the time complexity is $O\big(\frow(\M) L \cdot \ns\big) = O\left(\frow(\M) \|\t\|_1^2 L^3 / \eps^2\right)$, as desired.
\end{proof}

Next, we consider the accuracy guarantee of $\eps \|\x^{\ast}\|_{\infty}$.
The following lemma given by \cite{andoni2019solving} shows that this is a weaker requirement than $\eps \left\|\D_{\M}^{-1}\b\right\|_{\infty}$ up to a factor of $2$.

\begin{lemma} \label{lem:bound_comparison}
	For any RDD $\M$ and $\b \in \range(\M)$, it holds that $\left\|\D_{\M}^{-1}\b\right\|_{\infty} \le 2\|\x^{\ast}\|_{\infty}$.
\end{lemma}

\begin{proof}
Since $\M\x^{\ast} = \b$ by Theorem~\ref{thm:x*}, we have
\begin{align*}
	\left\|\D_{\M}^{-1}\b\right\|_{\infty} = \left\|\D_{\M}^{-1}\M\x^{\ast}\right\|_{\infty} = \left\|\left(\I-\D_{\M}^{-1}\A_{\M}^{\top}\right)\x^{\ast}\right\|_{\infty} \le \|\x^{\ast}\|_{\infty}+\left\|\D_{\M}^{-1}\A_{\M}^{\top}\x^{\ast}\right\|_{\infty} \le 2\|\x^{\ast}\|_{\infty},
\end{align*}
where we used the triangle inequality and $\left\|\D_{\M}^{-1}\A_{\M}^{\top}\right\|_{\infty} \le 1$ for RDD $\M$.
\end{proof}

The proof of Theorem~\ref{thm:MC_square} below relies on a variance analysis when $\A_{\M}$, $\b$, and $\t$ are nonnegative.

\begin{proof}[Proof of Theorem~\ref{thm:MC_square}]
We use Algorithm~\ref{alg:RDD_MC_truncated}.
Since $\M$ is RDDZ, $\A_{\M}$ is nonnegative.
As $\t$ is also nonnegative, the product of signs considered in each sample is always $1$, and each sampled value is
\begin{align*}
	\sum_{\ell=0}^{L-1} \sum_{v \in [n]} X(\ell,v,1) \cdot \frac{1}{2} \|\t\|_1 \cdot \frac{\b(v)}{d_{\M}(v)} \cdot L.
\end{align*}
The $X(\ell,v,1)$'s are nonpositively correlated since they are indicators of mutually exclusive events.
Thus, we can upper bound the variance of $\hat{x}$ as
\begin{align*}
	\Var[\hat{x}] & \le \frac{1}{\ns}\sum_{\ell=0}^{L-1} \sum_{v \in [n]} \Var\left[\frac{1}{2} \|\t\|_1 \cdot \frac{\b(v)}{d_{\M}(v)} \cdot L \cdot X(\ell,v,1)\right] \\
	& \le \frac{1}{\ns}\sum_{\ell=0}^{L-1} \sum_{v \in [n]}\frac{1}{4} \|\t\|_1^2 \cdot \frac{\b(v)^{2}}{d_{\M}(v)^{2}} \cdot L^2 \cdot \frac{1}{L\|\t\|_1} \cdot \t^{\top}\left( \frac{1}{2}\left(\I+\D_{\M}^{-1}\A_{\M}^{\top}\right) \right)^{\ell} \e_{v} \\
	& = \frac{\|\t\|_1 L}{4\ns}\sum_{\ell=0}^{L-1} \sum_{v \in [n]}\frac{\b(v)}{d_{\M}(v)} \cdot \t^{\top}\left( \frac{1}{2}\left(\I+\D_{\M}^{-1}\A_{\M}^{\top}\right) \right)^{\ell} \D_{\M}^{-1}\big(\b(v)\e_v\big).
\end{align*}
Using $\b(v) \ge 0$ and $\frac{\b(v)}{d_{\M}(v)} \le \left\| \D_{\M}^{-1}\b \right\|_{\infty}$ for all $v \in [n]$, we can further upper bound the variance by
\begin{align*}
	\Var[\hat{x}] & \le \frac{\|\t\|_1 L}{4\ns} \left\| \D_{\M}^{-1}\b \right\|_{\infty} \t^{\top}\sum_{\ell=0}^{L-1}\left(\frac{1}{2}\left(\I+\D_{\M}^{-1}\A_{\M}^{\top}\right)\right)^{\ell} \D_{\M}^{-1}\b \\
	& \le \frac{\|\t\|_1 L}{4\ns} \left\| \D_{\M}^{-1}\b \right\|_{\infty} \t^{\top}\sum_{\ell=0}^{\infty}\left(\frac{1}{2}\left(\I+\D_{\M}^{-1}\A_{\M}^{\top}\right)\right)^{\ell} \D_{\M}^{-1}\b \\
	& = \frac{\|\t\|_1 L}{2\ns} \left\| \D_{\M}^{-1}\b \right\|_{\infty} \t^{\top}\x^{\ast}.
\end{align*}
By Lemma~\ref{lem:MC_unbiased} and Chebyshev's inequality, we have
\begin{align*}
	\Pr\left\{\left|\hat{x}-\t^{\top}\x_L^{\ast}\right| \ge \frac{1}{2} \eps \|\x^{\ast}\|_{\infty}\right\} \le \frac{\Var[\hat{x}_{t}]}{\frac{1}{4} \eps^2 \|\x^{\ast}\|_{\infty}^2} \le \frac{2 \|\t\|_1 \left\| \D_{\M}^{-1}\b \right\|_{\infty} L \cdot \t^{\top}\x^{\ast}}{\ns \cdot \eps^2 \|\x^{\ast}\|_{\infty}^2} \le \frac{4 \|\t\|_1^2 L}{\ns \cdot \eps^2},
\end{align*}
where in the last step we used $\t^{\top}\x^{\ast} \le \|\t\|_1 \|\x^{\ast}\|_{\infty}$ and $\left\| \D_{\M}^{-1}\b \right\|_{\infty}\leq 2\left\|\x^{\ast}\right\|_{\infty}$ by Lemma~\ref{lem:bound_comparison}.
Hence, to guarantee that this deviation probability is at most $1/4$, we set $\ns := \Theta\left(\|\t\|_1^2 L/\eps^2\right)$.
This leads to the time complexity of $O\big(\frow(\M) L \cdot \ns\big) = O\left(\frow(\M) \|\t\|_1^2 L^2/\eps^2\right)$.
\end{proof}

Our last result for solving RDD systems using random-walk sampling is Theorem~\ref{thm:MC_relative}, which assumes that $\left\|\D_{\M}^{-1}\b\right\|_{\infty}$ is known and provides a relative error guarantee.
Its proof below relies on the Stopping Rule Theorem given in \cite{dagum2000optimal} for adaptively setting the number of samples in Monte Carlo estimation.

\begin{proof}[Proof of Theorem~\ref{thm:MC_relative}]
We use Algorithm~\ref{alg:RDD_MC_truncated}, but we need to modify the algorithm to set $\ns$ adaptively and divide the summation of the sampled values by the final value of $\ns$ after a certain stopping condition is met.
Let $Z$ denote the random variable of the quantity obtained by one random-walk sampling.
Then, $\E[Z] = \t^{\top}\x^{\ast}_L$ by Lemma~\ref{lem:MC_unbiased} and $0 \le Z \le \frac{1}{2} \|\t\|_1 \left\|\D_{\M}^{-1}\b\right\|_{\infty} L$.
Now consider modifying the algorithm to sample random variables $Z \big/ \left(\frac{1}{2} \|\t\|_1 \left\|\D_{\M}^{-1}\b\right\|_{\infty} L\right)$ (recall that $\|\t\|_1$ and $\left\|\D_{\M}^{-1}\b\right\|_{\infty}$ are known).
Then each sample lies in $[0,1]$ and by the Stopping Rule Theorem~\cite{dagum2000optimal}, we can obtain a $(1 \pm \eps/2)$-multiplicative approximation of $\t^{\top}\x^{\ast}_L \big/ \left(\frac{1}{2} \|\t\|_1 \left\|\D_{\M}^{-1}\b\right\|_{\infty} L\right)$ with probability at least $3/4$ using $O\left(\|\t\|_1 \left\|\D_{\M}^{-1}\b\right\|_{\infty} L \big/ \left(\eps^2 \cdot \t^{\top}\x^{\ast}_L\right)\right)$ samples in expectation.
Such an approximation yields an $\hat{x}$ that satisfies $\left|\hat{x}-\t^{\top}\x^{\ast}_L\right| \le \frac{1}{2}\eps \cdot \t^{\top}\x^{\ast}_L \le \frac{1}{2}\eps \cdot \t^{\top}\x^{\ast}$, as desired.
On the other hand, our setting of $L$ ensures that $\t^{\top}\x^{\ast}_L \ge \frac{1}{2}\t^{\top}\x^{\ast}$ by Theorem~\ref{thm:truncation_error}, so the expected number of samples is $O\left(\|\t\|_1\left\|\D_{\M}^{-1}\b\right\|_{\infty} L \big/ \left(\eps^2 \cdot \t^{\top}\x^{\ast}\right)\right)$.
Multiplying this bound by $O\big(\frow(\M) L\big)$ gives the result.
\end{proof}

We remark that Theorem~\ref{thm:MC_relative} remains valid if we replace the exact knowledge of $\left\|\D_{\M}^{-1}\b\right\|_{\infty}$ with an upper bound of it and substitute this upper bound for $\left\|\D_{\M}^{-1}\b\right\|_{\infty}$ in the time complexity.

On the other hand, if $\M$ is CDD, then $\M^{\top}$ is RDD and $\frac{1}{2}\left(\I+\D_{\M}^{-1}\left|\A_{\M}\right|\right)$ is row substochastic.
As we can transpose the expression of $\t^{\top}\x^{\ast}_L$ to obtain
\begin{align*}
	\t^{\top}\x^{\ast}_L = \frac{1}{2}\b^{\top}\D_{\M}^{-1}\sum_{\ell=0}^{L-1}\left(\frac{1}{2}\left(\I+\A_{\M}\D_{\M}^{-1}\right)\right)^{\ell}\t = \frac{1}{2}\b^{\top}\sum_{\ell=0}^{L-1}\left(\frac{1}{2}\left(\I+\D_{\M}^{-1}\A_{\M}\right)\right)^{\ell}\D_{\M}^{-1}\t,
\end{align*}
our algorithms and results for RDD systems (except Theorem~\ref{thm:MC_square}) applies to CDD systems by interchanging $\t$ with $\b$ and replacing $\A_{\M}^{\top}$ with $\A_{\M}$.
This justifies our claim that we can derive symmetric results by replacing RDD/RDDZ by CDD/CDDZ, swapping $\b$ and $\t$ (except in $\t^{\top}\x^{\ast}$), and replacing $\frow(\M)$ by $\fcol(\M)$ in the theorem statements.
This argument also straightforwardly applies to our subsequent algorithms and results based on local push and the bidirectional method.

\section{The Local Push Method} \label{sec:push}

In this section, we adapt the local push methods to estimate $\t^{\top}\x^{\ast}_L$.
We first describe our \push algorithm as a primitive that can be applied to both RDD and CDD systems.
After that, we establish different properties of \push for RDD and CDD systems.
This leads to our proofs for Theorem~\ref{thm:push_RCDD} and some additional results.

\subsection{Description and Properties of the \push Algorithm}

For both RDD and CDD systems $\M\x = \b$, we describe our \push algorithm as a primitive that can be used for approximating the vector $2\x^{\ast}_L = \sum_{\ell=0}^{L-1}\left(\frac{1}{2}\left(\I+\D_{\M}^{-1}\A_{\M}^{\top}\right)\right)^{\ell}\D_{\M}^{-1}\b$.
The pseudocode is given in Algorithm~\ref{alg:push}, where the basic settings are similar to \fpush and \bpush.

In the \push algorithm, the initialization step sets the reserve and residue vectors to $\vzero$, except that $\r^{(0)}$ is set to be $\D_{\M}^{-1}\b$, which requires $O\big(\|\b\|_0\big)$ time if we assume that we can scan through the nonzero entries of $\b$ in $O\big(\|\b\|_0\big)$ time.
Next, the main loop iterates over levels $\ell$ from $0$ to $L-2$.
At each level $\ell$, the algorithm performs a local push operation on each coordinate $v$ whose residue $\r^{(\ell)}(v)$ exceeds the threshold $\rmax$ in absolute value.
The push operation on $v$ at level $\ell$ sets the reserve $\p^{(\ell)}(v)$ to $\r^{(\ell)}(v)$, increments $\r^{(\ell+1)}$ by $\frac{1}{2}\left(\I+\D_{\M}^{-1}\A_{\M}^{\top}\right)\left(\r^{(\ell)}(v)\e_v\right)$, and sets $\r^{(\ell)}(v)$ to $0$.
In effect, the second step in the push operation increases $\r^{(\ell+1)}(v)$ by $\frac{1}{2}\r^{(\ell)}(v)$ and increases $\r^{(\ell+1)}(u)$ by $\frac{\A(v,u)}{2d_{\M}(u)} \cdot \r^{(\ell)}(v)$ for each $u \in [n]$ with $\A(v,u) \ne 0$, which can be done in $\big\|\M(\cdot,v)\big\|_0$ time given oracle row/column access to $\M$.

\begin{algorithm}[ht]
	\DontPrintSemicolon
	\caption{$\push(\M,\b,L,\rmax)$} \label{alg:push}
	\KwIn{oracle access to $\M$ and $\b$, truncation parameter $L$, threshold $\rmax$}
	\KwOut{dictionaries $\p^{(\ell)}$ for reserves and $\r^{(\ell)}$ for residues, for $\ell = 0,1,\dots,L-1$}
	$\r^{(\ell)} \gets \vzero$, $\p^{(\ell)} \gets \vzero$ for $\ell=0,1,\dots,L-1$ \;
	$\r^{(0)} \gets \D_{\M}^{-1}\b$ \;
	\For{$\ell$ \textup{from} $0$ \textup{to} $L-2$}{
		\For{\textup{each} $v$ \textup{with} $\left|\r^{(\ell)}(v)\right| > \rmax$}{
			$\p^{(\ell)}(v) \gets \r^{(\ell)}(v)$ \;
			$\r^{(\ell+1)} \gets \r^{(\ell+1)} + \frac{1}{2}\left(\I+\D_{\M}^{-1}\A_{\M}^{\top}\right)\left(\r^{(\ell)}(v)\e_v\right)$ \;
			$\r^{(\ell)}(v) \gets 0$ \;
		}
	}
	\Return $\p^{(\ell)}$ and $\r^{(\ell)}$ for $\ell = 0,1,\dots,L-1$ \;
\end{algorithm}

The following lemma gives the key invariant property of the \push algorithm.

\begin{lemma} \label{lem:invariant}
	The push operations preserve the following invariant:
	\begin{align*}
		\sum_{\ell=0}^{L-1}\left(\frac{1}{2}\left(\I+\D_{\M}^{-1}\A_{\M}^{\top}\right)\right)^{\ell}\D_{\M}^{-1}\b = \sum_{\ell=0}^{L-1} \p^{(\ell)} + \sum_{\ell=0}^{L-1} \sum_{\ell'=0}^{L-\ell-1} \left(\frac{1}{2}\left(\I+\D_{\M}^{-1}\A_{\M}^{\top}\right)\right)^{\ell} \r^{(\ell')}.
	\end{align*}
\end{lemma}

\begin{proof}
We prove the lemma by induction on the number of push operations performed.
As the base case, initially only $\r^{(0)} = \D_{\M}^{-1}\b$ contains nonzero entries, so the right-hand side of the invariant equation equals $\sum_{\ell=0}^{L-1}\left(\frac{1}{2}\left(\I+\D_{\M}^{-1}\A_{\M}^{\top}\right)\right)^{\ell}\D_{\M}^{-1}\b$, which equals the left-hand side of the equation, as desired.
For the inductive step, we examine the alteration of the right-hand side of the equation incurred by a single push operation on $v \in [n]$ at level $\ell'$, and it suffices to check that the alteration is $\vzero$.
The first step changes $\p^{(\ell')}(v)$ from $0$ to $\r^{(\ell')}(v)$, so $\sum_{\ell=0}^{L-1}\p^{(\ell)}$ is increased by $\r^{(\ell')}(v)\e_v$.
The second step increases $\r^{(\ell'+1)}$ by $\frac{1}{2}\left(\I+\D_{\M}^{-1}\A_{\M}^{\top}\right)\left(\r^{(\ell')}(v)\e_v\right)$, which in turn increases the right-hand side of the invariant equation by $\r^{(\ell)}(v)$ times
\begin{align*}
	\sum_{\ell=0}^{L-\ell'-2}\left(\frac{1}{2}\left(\I+\D_{\M}^{-1}\A_{\M}^{\top}\right)\right)^{\ell} \left(\frac{1}{2}\left(\I+\D_{\M}^{-1}\A_{\M}^{\top}\right)\right)\e_v = \sum_{\ell=0}^{L-\ell'-2}\left(\frac{1}{2}\left(\I+\D_{\M}^{-1}\A_{\M}^{\top}\right)\right)^{\ell+1} \e_v.
\end{align*}
The third step sets $\r^{(\ell')}(v)$ to $0$, decreasing the right-hand side of the invariant equation by $\r^{(\ell')}(v)$ times $\sum_{\ell=0}^{L-\ell'-1}\left(\frac{1}{2}\left(\I+\D_{\M}^{-1}\A_{\M}^{\top}\right)\right)^{\ell}\e_v$.
Consequently, the alteration equals $\r^{(\ell')}(v)$ times
\begin{align*}
	& \phantom{{}={}} \e_v + \sum_{\ell=0}^{L-\ell'-2}\left(\frac{1}{2}\left(\I+\D_{\M}^{-1}\A_{\M}^{\top}\right)\right)^{\ell+1} \e_v - \sum_{\ell=0}^{L-\ell'-1}\left(\frac{1}{2}\left(\I+\D_{\M}^{-1}\A_{\M}^{\top}\right)\right)^{\ell}\e_v \\
	& = \left(\I + \sum_{\ell=1}^{L-\ell'-1}\left(\frac{1}{2}\left(\I+\D_{\M}^{-1}\A_{\M}^{\top}\right)\right)^{\ell} - \sum_{\ell=0}^{L-\ell'-1}\left(\frac{1}{2}\left(\I+\D_{\M}^{-1}\A_{\M}^{\top}\right)\right)^{\ell}\right)\e_v = \vzero,
\end{align*}
as desired.
\end{proof}

In light of this invariant, we use $\frac{1}{2}\t^{\top}\left(\sum_{\ell=0}^{L-1}\p^{(\ell)} + \r^{(L-1)}\right)$ as an estimate of $\t^{\top}\x^{\ast}_L$.
Note that this quantity can be computed during the push process.
The next lemma shows that if $\M$ is RDD, then the absolute error between this quantity and $\t^{\top}\x^{\ast}_L$ can be bounded.

\begin{lemma} \label{lem:push_error}
	If $\M$ is RDD, then $\left| \frac{1}{2}\t^{\top}\left(\sum_{\ell=0}^{L-1}\p^{(\ell)} + \r^{(L-1)}\right) - \t^{\top}\x^{\ast}_L \right| \le \frac{1}{2} \|\t\|_1 L^2 \cdot \rmax$.
\end{lemma}

\begin{proof}
By Lemma~\ref{lem:invariant} and Equation~\eqref{eqn:x*_L}, we have
\begin{align*}
	& \phantom{{}={}} \left| \frac{1}{2}\t^{\top}\left(\sum_{\ell=0}^{L-1}\p^{(\ell)} + \r^{(L-1)}\right) - \t^{\top}\x^{\ast}_L \right| = \frac{1}{2}\left|\sum_{\ell'=0}^{L-2}\sum_{\ell=0}^{L-\ell'-1}\t^{\top}\left(\frac{1}{2}\left(\I+\D_{\M}^{-1}\A_{\M}^{\top}\right)\right)^{\ell}\r^{(\ell')} \right| \\
	& \le \frac{1}{2}\sum_{\ell'=0}^{L-2}\sum_{\ell=0}^{L-\ell'-1} \|\t\|_1 \left\|\left(\frac{1}{2}\left(\I+\D_{\M}^{-1}\A_{\M}^{\top}\right)\right)^{\ell}\r^{(\ell')}\right\|_{\infty} \le \frac{1}{2} \|\t\|_1 L \sum_{\ell'=0}^{L-2}\left\|\r^{(\ell')}\right\|_{\infty} \le \frac{1}{2} \|\t\|_1 L^2 \cdot \rmax,
\end{align*}
where we used $\left\|\frac{1}{2}\left(\I+\D_{\M}^{-1}\A_{\M}^{\top}\right)\right\|_{\infty} \le \frac{1}{2}\left(1+\left\|\D_{\M}^{-1}\A_{\M}^{\top}\right\|_{\infty}\right) \le 1$ for RDD $\M$ and $\left\|\r^{(\ell')}\right\|_{\infty} \le \rmax$ for any $\ell' \in [0, L-2]$ as guaranteed by the process of \push.
\end{proof}

We will also use the following inequality version of the invariant to bound the running time of the \push algorithm, which shares a similar proof with the invariant equation.

\begin{lemma} \label{lem:invariant_inequality}
	The push operations preserve the following inequality:
	\begin{align*}
		\sum_{\ell=0}^{L-1}\left(\frac{1}{2}\left(\I+\D_{\M}^{-1}\left|\A_{\M}^{\top}\right|\right)\right)^{\ell}\D_{\M}^{-1}|\b| \ge \sum_{\ell=0}^{L-1} \left|\p^{(\ell)}\right| + \sum_{\ell=0}^{L-1} \sum_{\ell'=0}^{L-\ell-1} \left(\frac{1}{2}\left(\I+\D_{\M}^{-1}\left|\A_{\M}^{\top}\right|\right)\right)^{\ell} \left|\r^{(\ell')}\right|.
	\end{align*}
\end{lemma}

\begin{proof}
	We prove the lemma by induction again.
	As the base case, initially the right-hand side of the inequality equals $\sum_{\ell=0}^{L-1}\left(\frac{1}{2}\left(\I+\D_{\M}^{-1}\left|\A_{\M}^{\top}\right|\right)\right)^{\ell}\D_{\M}^{-1}|\b|$, which equals the left-hand side of the inequality, so the inequality holds.
	For the inductive step, we examine the alteration of the right-hand side of the equation incurred by a single push operation on $v \in [n]$ at level $\ell'$, and it suffices to check that the alteration is elementwise nonpositive.
	The first step changes $\p^{(\ell')}(v)$ from $0$ to $\r^{(\ell')}(v)$, so $\sum_{\ell=0}^{L-1}\left|\p^{(\ell)}\right|$ is increased by $\left|\r^{(\ell')}(v)\right|\e_v$.
	In the second step, $\r^{(\ell'+1)}$ is increased by $\frac{1}{2}\left(\I+\D_{\M}^{-1}\A_{\M}^{\top}\right)\left(\r^{(\ell')}(v)\e_v\right)$, so $\left|\r^{(\ell'+1)}\right|$ is increased by at most $\frac{1}{2}\left(\I+\D_{\M}^{-1}\left|\A_{\M}^{\top}\right|\right)\left(\left|\r^{(\ell')}(v)\right|\e_v\right)$.
	Third, $\r^{(\ell')}(v)$ is set to $0$, which means that $\left|\r^{(\ell')}(v)\right|$ is decreased by $\left|\r^{(\ell')}(v)\right|$.
	Together, the alteration is elementwisely smaller than or equal to $\left|\r^{(\ell')}(v)\right|$ times
	\begin{align*}
		& \phantom{{}={}} \e_v + \sum_{\ell=0}^{L-\ell'-2}\left(\frac{1}{2}\left(\I+\D_{\M}^{-1}\left|\A_{\M}^{\top}\right|\right)\right)^{\ell}\left(\frac{1}{2}\left(\I+\D_{\M}^{-1}\left|\A_{\M}^{\top}\right|\right)\right)\e_v - \sum_{\ell=0}^{L-\ell'-1}\left(\frac{1}{2}\left(\I+\D_{\M}^{-1}\left|\A_{\M}^{\top}\right|\right)\right)^{\ell}\e_v \\
		& = \left(\I + \sum_{\ell=1}^{L-\ell'-1}\left(\frac{1}{2}\left(\I+\D_{\M}^{-1}\left|\A_{\M}^{\top}\right|\right)\right)^{\ell+1} - \sum_{\ell=0}^{L-\ell'-1}\left(\frac{1}{2}\left(\I+\D_{\M}^{-1}\left|\A_{\M}^{\top}\right|\right)\right)^{\ell}\right)\e_v = \vzero,
	\end{align*}
	as desired.
\end{proof}

The next lemma bounds the complexity of the \push algorithm by a convoluted expression.
We will shortly see how to simplify this expression for some special systems and settings.

\begin{lemma} \label{lem:push_cost}
	Suppose that we can scan through the nonzero entries of $\b$ in $O\big(\|\b\|_0\big)$ time.
	Then the complexity of the \push algorithm is bounded by
	\begin{align*}
		O\left(\|\b\|_0 + \frac{1}{\rmax}\sum_{v \in [n]}\big\|\M(\cdot,v)\big\|_0 \cdot \e_v^{\top}\sum_{\ell=0}^{L-1}\left(\frac{1}{2}\left(\I+\D_{\M}^{-1}\left|\A_{\M}^{\top}\right|\right)\right)^{\ell}\D_{\M}^{-1}|\b|\right).
	\end{align*}
\end{lemma}

\begin{proof}
Observe that each time a push operation is performed on a coordinate $v \in [n]$, the value of $\e_v^{\top}\sum_{\ell=0}^{L-1}\left|\p^{(\ell)}\right|$ is increased by at least $\rmax$.
However, by Lemma~\ref{lem:invariant_inequality}, $\e_v^{\top}\sum_{\ell=0}^{L-1}\left|\p^{(\ell)}\right|$ is always upper bounded by $\e_v^{\top}\sum_{\ell=0}^{L-1}\left(\frac{1}{2}\left(\I+\D_{\M}^{-1}\left|\A_{\M}^{\top}\right|\right)\right)^{\ell}\D_{\M}^{-1}|\b|$.
Therefore, the total number of push operations performed on $v$ is at most $\frac{1}{\rmax} \cdot \e_v^{\top}\sum_{\ell=0}^{L-1}\left(\frac{1}{2}\left(\I+\D_{\M}^{-1}\left|\A_{\M}^{\top}\right|\right)\right)^{\ell}\D_{\M}^{-1}|\b|$.
Since each push operation on $v$ takes $O\left(\big\|\M(\cdot,v)\big\|_0\right)$ time, the lemma follows by summing the cost of the push operations over all $v \in [n]$ and adding the $O\big(\|\b\|_0\big)$ time for initialization.
\end{proof}

If $\M$ is CDD and the nonzero entries of $\M$ have absolute values of $\Omega(1)$, the next lemma shows that the complexity of \push can be simplified.

\begin{lemma} \label{lem:push_cost_CDD}
	Suppose that $\M$ is CDD, the nonzero entries of $\M$ have absolute values of $\Omega(1)$, and we can scan through the nonzero entries of $\b$ in $O\big(\|\b\|_0\big)$ time.
	Then the complexity of the \push algorithm is $O\big(\|\b\|_0 + \|\b\|_1 L/\rmax\big)$.
\end{lemma}

\begin{proof}
It suffices to bound the complexity of the push operations.
By Lemma~\ref{lem:push_cost}, the push cost can be upper bounded by big-$O$ of
\begin{align*}
	& \phantom{{}={}} \frac{1}{\rmax}\sum_{v \in [n]}\big\|\M(\cdot,v)\big\|_0 \cdot \e_v^{\top}\sum_{\ell=0}^{L-1}\left(\frac{1}{2}\left(\I+\D_{\M}^{-1}\left|\A_{\M}^{\top}\right|\right)\right)^{\ell}\D_{\M}^{-1}|\b| \\
	& = \frac{1}{\rmax}\sum_{v \in [n]}\big\|\M(\cdot,v)\big\|_0 \cdot \e_v^{\top}\D_{\M}^{-1}\sum_{\ell=0}^{L-1}\left(\frac{1}{2}\left(\I+\left|\A_{\M}^{\top}\right|\D_{\M}^{-1}\right)\right)^{\ell}|\b|.
\end{align*}
We can derive that
\begin{align*}
	& \phantom{{}={}} \sum_{v \in [n]}\big\|\M(\cdot,v)\big\|_0 \cdot \e_v^{\top}\D_{\M}^{-1} = \sum_{v \in [n]}\frac{\|\M(\cdot,v)\|_0}{d_{\M}(v)} \cdot \e_v^{\top} \\
	& \le \sum_{v \in [n]}\max_{u \in [n]}\left\{\frac{\|\M(\cdot,u)\|_0}{d_{\M}(u)}\right\} \cdot \e_v^{\top} = \max_{u \in [n]}\left\{\frac{\|\M(\cdot,u)\|_0}{d_{\M}(u)}\right\} \cdot \vone^{\top}.
\end{align*}
Consequently, the push cost can be upper bounded by big-$O$ of
\begin{align*}
	& \phantom{{}={}} \frac{1}{\rmax} \cdot \max_{u \in [n]}\left\{\frac{\|\M(\cdot,u)\|_0}{d_{\M}(u)}\right\} \cdot \vone^{\top} \sum_{\ell=0}^{L-1}\left(\frac{1}{2}\left(\I+\left|\A_{\M}^{\top}\right|\D_{\M}^{-1}\right)\right)^{\ell}|\b| \\
	& \le \frac{1}{\rmax} \cdot \max_{u \in [n]}\left\{\frac{\|\M(\cdot,u)\|_0}{d_{\M}(u)}\right\} \cdot \|\vone\|_{\infty} \left\|\sum_{\ell=0}^{L-1}\left(\frac{1}{2}\left(\I+\left|\A_{\M}^{\top}\right|\D_{\M}^{-1}\right)\right)^{\ell}|\b|\right\|_1 \\
	& \le \frac{\|\b\|_1 L}{\rmax} \cdot \max_{u \in [n]}\left\{\frac{\|\M(\cdot,u)\|_0}{d_{\M}(u)}\right\},
\end{align*}
where we used $\left\|\frac{1}{2}\left(\I+\left|\A_{\M}^{\top}\right|\D_{\M}^{-1}\right)\right\|_1 \le 1$ for CDD $\M$.
Since the nonzero entries of $\M$ have absolute values of $\Omega(1)$, we have
\begin{align*}
	\max_{u \in [n]}\left\{\frac{\|\M(\cdot,u)\|_0}{d_{\M}(u)}\right\} \le \max_{u \in [n]}\left\{\frac{\|\M(\cdot,u)\|_0}{\sum_{v \in [n], v \ne u, \M(v,u) \ne 0}|\M(v,u)|}\right\} = O(1),
\end{align*}
so the push cost is upper bounded by $O\big(\|\b\|_1 L/\rmax\big)$, proving the lemma.
\end{proof}

\subsection{Proof of Theorem~\ref{thm:push_RCDD} and Additional Results}

Having established the properties of the \push algorithm for RDD and CDD systems, we can readily combine them to prove Theorem~\ref{thm:push_RCDD} for RCDD systems.

\begin{proof}[Proof of Theorem~\ref{thm:push_RCDD}]
We run \push with $\rmax := \eps/L^2$ and use $\frac{1}{2}\t^{\top}\left(\sum_{\ell=0}^{L-1}\p^{(\ell)} + \r^{(L-1)}\right)$ as the result.
Since $\M$ is RDD, by Lemma~\ref{lem:push_error}, the absolute error between the estimate and $\t^{\top}\x^{\ast}_L$ is upper bounded by $\frac{1}{2} \|\t\|_1 L^2 \cdot \rmax \le \frac{1}{2} \eps \|\t\|_1$.
Since $\M$ is CDD and its nonzero entries have absolute values of $\Omega(1)$, by Lemma~\ref{lem:push_cost_CDD}, the time complexity of the \push algorithm is $O\big(\|\b\|_0 + \|\b\|_1 L/\rmax\big) = O\big(\|\b\|_0 + \|\b\|_1 L^3/\eps\big)$.
This finishes the proof.
\end{proof}

Additionally, for RDD systems, if $\b$ is a canonical unit vector, we can derive a closed-form complexity bound by averaging over all possible choices of $\b$, as stated below.
Notably, this result does not require the nonzero entries of $\M$ to have absolute values of $\Omega(1)$.

\begin{theorem} \label{thm:push_average}
	Suppose that $\M$ is RDD and $\b = \e_w$ for a given $w \in [n]$.
	Then there exists a deterministic algorithm that computes an estimate $\hat{x}$ such that $\left|\hat{x}-\t^{\top}\x^{\ast}\right| \le \eps \|\t\|_1 \left\|\D_{\M}^{-1}\b\right\|_{\infty}$ with average time complexity $O\left(\nnz(\M)/n \cdot L^3 \eps^{-1}\right) = \tO\left(\nnz(\M)/n \cdot \gamma^{-3}\eps^{-1}\right)$, where the average is taken over all choices of $w \in [n]$.
\end{theorem}

\begin{proof}[Proof of Theorem~\ref{thm:push_average}]
We run \push with $\rmax := \eps\big/\big(d_{\M}(w) L^2\big)$.
By Lemma~\ref{lem:push_error}, the absolute error is upper bounded by $\frac{1}{2} \|\t\|_1 L^2 \cdot \rmax = \frac{1}{2} \eps \|\t\|_1 / d_{\M}(w) = \frac{1}{2} \eps \|\t\|_1 \left\|\D_{\M}^{-1}\b\right\|_{\infty}$.
Using Lemma~\ref{lem:push_cost}, we can upper bound the average push cost over $w$ by
\begin{align*}
	O\left(\frac{1}{n}\sum_{w \in [n]} \frac{d_{\M}(w) L^2}{\eps} \sum_{v \in [n]}\big\|\M(\cdot,v)\big\|_0 \cdot \e_v^{\top}\sum_{\ell=0}^{L-1}\left(\frac{1}{2}\left(\I+\D_{\M}^{-1}\left|\A_{\M}^{\top}\right|\right)\right)^{\ell}\D_{\M}^{-1}\e_w\right).
\end{align*}
Since $\sum_{w \in [n]}d_{\M}(w) \D_{\M}^{-1}\e_w = \vone$, the the average push cost can be upper bounded by big-$O$ of
\begin{align*}
	& \phantom{{}={}} \frac{L^2}{n\eps}\sum_{v \in [n]}\big\|\M(\cdot,v)\big\|_0 \cdot \e_v^{\top} \sum_{\ell=0}^{L-1}\left(\frac{1}{2}\left(\I+\D_{\M}^{-1}\left|\A_{\M}^{\top}\right|\right)\right)^{\ell} \vone \\
	& \le \frac{L^2}{n\eps}\sum_{v \in [n]}\big\|\M(\cdot,v)\big\|_0 \left\|\sum_{\ell=0}^{L-1}\left(\frac{1}{2}\left(\I+\D_{\M}^{-1}\left|\A_{\M}^{\top}\right|\right)\right)^{\ell} \vone\right\|_{\infty} \\
	& \le \frac{L^2}{n\eps}\sum_{v \in [n]}\big\|\M(\cdot,v)\big\|_0 \cdot L = \frac{\nnz(\M)/n \cdot L^3}{\eps}.
\end{align*}
On the other hand, since $\b = \e_w$ for a given $w \in [n]$, we can perform the initialization in $O(1)$ time.
This finishes the proof.
\end{proof}

\section{The Bidirectional Method} \label{sec:bidirectional}

This section combines the techniques of random-walk sampling and local push to develop bidirectional algorithms for estimating $\t^{\top}\x^{\ast}$, which leads to Theorem~\ref{thm:bidirectional_RCDD} and some additional results.

The framework of the bidirectional method for RDD systems is presented in Algorithm~\ref{alg:bidirectional_RDD}.
First, we invoke the \push algorithm to obtain the reserve and residue vectors.
Recall that the invariant equation of \push (Lemma~\ref{lem:invariant}) implies that
\begin{align*}
	& \phantom{{}={}} \t^{\top}\x^{\ast}_L - \frac{1}{2}\t^{\top}\left(\sum_{\ell=0}^{L-1}\p^{(\ell)} + \r^{(L-1)}\right) = \frac{1}{2} \t^{\top} \sum_{\ell'=0}^{L-2}\sum_{\ell=0}^{L-\ell-1}\left(\frac{1}{2}\left(\I+\D_{\M}^{-1}\A_{\M}^{\top}\right)\right)^{\ell} \r^{(\ell')} \\
	& = \frac{1}{2} \t^{\top} \sum_{\ell=0}^{L-1} \left(\frac{1}{2}\left(\I+\D_{\M}^{-1}\A_{\M}^{\top}\right)\right)^{\ell} \left(\sum_{\ell'=0}^{\min(L-\ell-1,L-2)} \r^{(\ell')}\right).
\end{align*}
So, instead of directly using $\frac{1}{2}\t^{\top}\left(\sum_{\ell=0}^{L-1}\p^{(\ell)} + \r^{(L-1)}\right)$ as an estimate of $\t^{\top}\x^{\ast}_L$, we estimate the right-hand side of the above equation using random-walk samplings from $|\t|/\|\t\|_1$ to reduce approximation error.
We employ the same sampling scheme as in Algorithm~\ref{alg:RDD_MC_truncated} to obtain a walk length $\ell$ and the coordinate $v$ reached by the random walk after $\ell$ steps, but each sampled value now involves the summation $\sum_{\ell'=0}^{\min(L-\ell-1,L-2)} \r^{(\ell')}(v)$.
We take the average of the estimates across $\ns$ independent samples and add it to $\frac{1}{2}\t^{\top}\left(\sum_{\ell=0}^{L-1}\p^{(\ell)} + \r^{(L-1)}\right)$ to obtain the final estimate $\hat{x}$.

We note that, compared to the sampling scheme in \cite{cui2025mixing}, our approach does not need additional data structures to maintain the prefix sums of residues, since we can directly compute $\sum_{\ell'=0}^{\min(L-\ell-1,L-2)} \r^{(\ell')}(v)$ in $O(L)$ time per sample without increasing the asymptotic complexity.

\begin{algorithm}[ht]
	\DontPrintSemicolon
	\caption{bidirectional method for RDD systems} \label{alg:bidirectional_RDD}
	\KwIn{oracle access to $\M$, $\b$, and $\t$, truncation parameter $L$, threshold $\rmax$, number of samples $\ns$}
	\KwOut{estimate $\hat{x}$ of $\t^{\top}\x^{\ast}$}
	$\r^{(\ell)}$ and $\p^{(\ell)}$ for $\ell = 0,1,\dots,L-1 \gets \push(\M,\b,L,\rmax)$ \;
	$\hat{x} \gets $$\frac{1}{2}\t^{\top}\left(\sum_{\ell=0}^{L-1}\p^{(\ell)} + \r^{(L-1)}\right)$ \;
	\For{$j$ \textup{from} $1$ \textup{to} $\ns$}{
		$\ell \gets \text{uniformly random sample from }[0,L-1]$ \;
		$v \gets \text{random sample according to distribution }|\t|/\|\t\|_1$ \;
		$\sigma \gets \sgn\big(\t(v)\big)$ \;
		\For{$k$ \textup{from} $1$ \textup{to} $\ell$}{
			simulate one step of the random walk from one of the following three possibilities: \;
			$\quad$ 1. w.p. $\frac{1}{2}$, $v' \gets v$ \textcolor{gray}{// stays put at $v$} \;
			$\quad$ 2. w.p. $\frac{|\A_{\M}(u,v)|}{2d_{\M}(v)}$ for each $u \in [n]$, $v' \gets u$, $\sigma \gets \sigma \cdot \text{sgn}\big(\A_{\M}(u,v)\big)$ \textcolor{gray}{// moves to $u$} \;
			$\quad$ 3. w.p. $\frac{1}{2} - \sum_{u \in [n]}\frac{|\A_{\M}(u,v)|}{2d_{\M}(v)}$, $\sigma \gets 0$, break the loop over $k$ \textcolor{gray}{// terminates} \;
			$v \gets v'$ \;
		}
		$\hat{x} \gets \hat{x} + \frac{1}{\ns} \cdot \sigma \cdot \frac{1}{2} \|\t\|_1 L \sum_{\ell'=0}^{\min(L-\ell-1,L-2)}\r^{(\ell')}(v)$ \;
	}
	\Return $\hat{x}$ \;
\end{algorithm}

The next lemma establishes the unbiasedness of the bidirectional estimator.

\begin{lemma}
	The sum of $\frac{1}{2}\t^{\top}\left(\sum_{\ell=0}^{L-1}\p^{(\ell)} + \r^{(L-1)}\right)$ and each sampled value in the bidirectional method described above gives an unbiased estimate of $\t^{\top}\x^{\ast}_L$.
\end{lemma}

\begin{proof}
Following the proof of Lemma~\ref{lem:MC_unbiased}, we can show that the expectation of each sampled value is $\frac{1}{2} \t^{\top} \sum_{\ell=0}^{L-1} \left(\frac{1}{2}\left(\I+\D_{\M}^{-1}\A_{\M}^{\top}\right)\right)^{\ell} \left(\sum_{\ell'=0}^{\min(L-\ell-1,L-2)} \r^{(\ell')}\right)$.
Combining this with the invariant equation in Lemma~\ref{lem:invariant} completes the proof.
\end{proof}

Next, we prove Theorem~\ref{thm:bidirectional_RCDD} by proving the two stated complexity bounds separately in the following two lemmas.
Their proofs are partly inspired by \cite{cui2025mixing} and \cite{yang2025improved}, respectively.

\begin{lemma} \label{lem:bidirectional_RCDD_Hoeffding}
	Suppose the same assumptions as in Theorem~\ref{thm:bidirectional_RCDD}.
	Then there exists a randomized algorithm that computes an estimate $\hat{x}$ such that $\Pr\left\{ \left|\hat{x}-\t^{\top}\x^{\ast}\right| \le \eps \right\} \ge \frac{3}{4}$ in time $O\big(\|\b\|_0\big)$ plus
	\begin{align*}
		O\left( \frow(\M)^{1/3} \|\t\|_1^{2/3} \|\b\|_1^{2/3} L^{7/3} \eps^{-2/3}\right).
	\end{align*}
\end{lemma}

\begin{proof}
We use the bidirectional method as in Algorithm~\ref{alg:bidirectional_RDD}.
Note that each sampled value equals $\sigma \cdot \frac{1}{2} \|\t\|_1 L \sum_{\ell'=0}^{\min(L-\ell-1,L-2)} \r^{(\ell')}(v)$ for some $\sigma \in \{0,\pm 1\}$ and $v \in [n]$ and the \push algorithm ensures that $\left\|\r^{(\ell')}\right\|_{\infty} \le \rmax$ for each $\ell' \in [0,L-2]$.
Thus, the absolute value of each sampled value is at most $\frac{1}{2} \|\t\|_1 L^2 \cdot \rmax$.
Using the Hoeffding bound, it follows that
\begin{align*}
	& \phantom{{}={}} \Pr\left\{ \left|\hat{x} - \t^{\top}\x^{\ast}_L\right| \ge \frac{1}{2} \eps \right\} \le 2\exp\left(-\frac{2\ns (\frac{1}{2}\eps)^2}{(\|\t\|_1 L^2 \cdot \rmax)^2}\right) = 2\exp\left(-\frac{\ns \cdot \eps^2}{2 \|\t\|_1^2 L^4 \cdot \rmax^2}\right).
\end{align*}
To guarantee that this probability is at most $1/4$, we set $\ns := \Theta\left(\|\t\|_1^2 L^4 \cdot \rmax^2 / \eps^2\right)$, where $\rmax$ will be determined shortly.
The cost for random-walk sampling is thus $O\left(\frow(\M) \|\t\|_1^2 L^5 \cdot \rmax^2 / \eps^2\right)$.

By Lemma~\ref{lem:push_cost_CDD}, the push cost is $O\left(\|\b\|_1 L/\rmax\right)$.
We set $\rmax := \frac{\eps^{2/3} \|\b\|_1^{1/3}}{\frow(\M)^{1/3} \|\t\|_1^{2/3} L^{4/3}}$, where $\|\b\|_1$ can be computed in $O(\|\b\|_0)$ time given our assumptions.
Consequently, the cost for random-walk sampling and push both becomes $O\left( \frow(\M)^{1/3} \|\t\|_1^{2/3} \|\b\|_1^{2/3} L^{7/3} \eps^{-2/3}\right)$, completing the proof.
\end{proof}

\begin{lemma} \label{lem:bidirectional_RCDD_variance}
	Suppose the same assumptions as in Theorem~\ref{thm:bidirectional_RCDD}.
	Then there exists a randomized algorithm that computes an estimate $\hat{x}$ such that $\Pr\left\{ \left|\hat{x}-\t^{\top}\x^{\ast}\right| \le \eps \right\} \ge \frac{3}{4}$ in time $O\big(\|\b\|_0\big)$ plus
	\begin{align*}
		O\left(\frow(\M)^{1/2} \|\t\|_1 \|\b\|_1^{1/2} \left\|\D_{\M}^{-1}\b\right\|_{\infty}^{1/2} L^{5/2} \eps^{-1}\right).
	\end{align*}
\end{lemma}

\begin{proof}
We use Algorithm~\ref{alg:bidirectional_RDD}.
Recall that the final estimate $\hat{x}$ equals $\frac{1}{2}\t^{\top}\left(\sum_{\ell=0}^{L-1}\p^{(\ell)} + \r^{(L-1)}\right)$ plus the average of $\ns$ independent samples of
\begin{align*}
	\sum_{\ell=0}^{L-1} \sum_{v \in [n]} \sum_{\sigma \in \{-1,1\}} X(\ell,v,\sigma) \cdot \sigma \cdot \frac{1}{2} \|\t\|_1 L \sum_{\ell'=0}^{\min(L-\ell-1,L-2)}\r^{(\ell')}(v),
\end{align*}
where $X(\ell,v,\sigma)$'s are the indicator random variables defined in Section~\ref{sec:MC}.
It follows that, for each $\ell \in [0,L-1]$ and $v \in [n]$,
\begin{align*}
	\E\big[X(\ell,v,-1) + X(\ell,v,1)\big] = \frac{1}{L} \left(\frac{|\t|}{\|\t\|_1}\right)^{\top}\left( \frac{1}{2}\left(\I+\D_{\M}^{-1}\left|\A_{\M}^{\top}\right|\right) \right)^{\ell} \e_{v}.
\end{align*}
By nonpositive correlation, we can upper bound the variance as
\begin{align*}
	\Var[\hat{x}] & \le \frac{1}{\ns} \sum_{\ell=0}^{L-1} \sum_{v \in [n]} \sum_{\sigma \in \{-1,1\}} \left(\sigma \cdot \frac{1}{2} \|\t\|_1 L \sum_{\ell'=0}^{\min(L-\ell-1,L-2)}\r^{(\ell')}(v)\right)^2 \E\big[X(\ell,v,\sigma)\big] \\
	& = \frac{\|\t\|_1^2 L^2}{4\ns} \sum_{\ell=0}^{L-1} \sum_{v \in [n]} \left(\sum_{\ell'=0}^{\min(L-\ell-1,L-2)}\r^{(\ell')}(v)\right)^2 \E\big[X(\ell,v,-1) + X(\ell,v,1)\big] \\
	& = \frac{\|\t\|_1 L}{4\ns} \sum_{\ell=0}^{L-1} \sum_{v \in [n]} \left|\sum_{\ell'=0}^{\min(L-\ell-1,L-2)}\r^{(\ell')}(v)\right|^2 |\t|^{\top}\left( \frac{1}{2}\left(\I+\D_{\M}^{-1}\left|\A_{\M}^{\top}\right|\right) \right)^{\ell} \e_{v}.
\end{align*}
For each $\ell \in [0,L-1]$ and $v \in [n]$, we have
\begin{align*}
	& \phantom{{}={}} \left|\sum_{\ell'=0}^{\min(L-\ell-1,L-2)}\r^{(\ell')}(v)\right|^2 \le \left(\sum_{\ell'=0}^{\min(L-\ell-1,L-2)}\left|\r^{(\ell')}(v)\right|\right)^2 \\
	& \le L \cdot \rmax \left(\sum_{\ell'=0}^{\min(L-\ell-1,L-2)}\left|\r^{(\ell')}(v)\right|\right) \le L \cdot \rmax \sum_{\ell'=0}^{L-\ell-1}\left|\r^{(\ell')}(v)\right|.
\end{align*}
By substituting, we obtain
\begin{align*}
	\Var[\hat{x}] & \le \frac{\|\t\|_1 L^2}{4\ns} \cdot \rmax \sum_{\ell=0}^{L-1} \sum_{v \in [n]} \sum_{\ell'=0}^{L-\ell-1}\left|\r^{(\ell')}(v)\right| \cdot |\t|^{\top}\left( \frac{1}{2}\left(\I+\D_{\M}^{-1}\left|\A_{\M}^{\top}\right|\right) \right)^{\ell} \e_{v} \\
	& = \frac{\|\t\|_1 L^2}{4\ns} \cdot \rmax \sum_{\ell=0}^{L-1} \sum_{\ell'=0}^{L-\ell-1} |\t|^{\top}\left( \frac{1}{2}\left(\I+\D_{\M}^{-1}\left|\A_{\M}^{\top}\right|\right) \right)^{\ell} \left|\r^{(\ell')}\right|.
\end{align*}
By Lemma~\ref{lem:invariant_inequality}, the summation $\sum_{\ell=0}^{L-1} \sum_{\ell'=0}^{L-\ell-1} \left( \frac{1}{2}\left(\I+\D_{\M}^{-1}\left|\A_{\M}^{\top}\right|\right) \right)^{\ell} \left|\r^{(\ell')}\right|$ can be upper bounded by $\sum_{\ell=0}^{L-1}\left(\frac{1}{2}\left(\I+\D_{\M}^{-1}\left|\A_{\M}^{\top}\right|\right)\right)^{\ell}\D_{\M}^{-1}|\b|$.
This leads to
\begin{align}
	& \phantom{{}={}} \Var[\hat{x}] \le \frac{\|\t\|_1 L^2}{4\ns} \cdot \rmax \cdot |\t|^{\top}\sum_{\ell=0}^{L-1}\left(\frac{1}{2}\left(\I+\D_{\M}^{-1}\left|\A_{\M}^{\top}\right|\right)\right)^{\ell}\D_{\M}^{-1}|\b| \label{eqn:var_bound} \\
	& \le \frac{\|\t\|_1 L^2}{4\ns} \cdot \rmax \cdot \|\t\|_1 \left\|\sum_{\ell=0}^{L-1}\left(\frac{1}{2}\left(\I+\D_{\M}^{-1}\left|\A_{\M}^{\top}\right|\right)\right)^{\ell}\D_{\M}^{-1}|\b|\right\|_{\infty} \le \frac{1}{4\ns} \|\t\|_1^2 \left\|\D_{\M}^{-1}\b\right\|_{\infty} L^3 \cdot \rmax. \nonumber
\end{align}
By Chebyshev's inequality, it follows that
\begin{align*}
	\Pr\left\{\left|\hat{x}-\t^{\top}\x^{\ast}_L\right| \ge \eps\right\} \le \frac{\|\t\|_1^2 \left\|\D_{\M}^{-1}\b\right\|_{\infty} L^3 \cdot \rmax}{4\ns \cdot \eps^2}.
\end{align*}
To guarantee that this probability is at most $1/4$, we set $\ns := \Theta\left(\|\t\|_1^2 \left\|\D_{\M}^{-1}\b\right\|_{\infty} L^3 \cdot \rmax / \eps^{2}\right)$, so the cost for random-walk sampling is $O\left(\frow(\M) \|\t\|_1^2 \left\|\D_{\M}^{-1}\b\right\|_{\infty} L^4 \cdot \rmax / \eps^{2}\right)$.

By Lemma~\ref{lem:push_cost_CDD}, the push cost is $O\big(\|\b\|_1 L/\rmax\big)$.
We set $\rmax := \frac{\eps \|\b\|_1^{1/2}}{\frow(\M)^{1/2} \|\t\|_1 \left\|\D_{\M}^{-1}\b\right\|_{\infty}^{1/2} L^{3/2}}$, where $\|\b\|_1$ and $\left\|\D_{\M}^{-1}\b\right\|_{\infty}$ can be computed in $O(\|\b\|_0)$ time given our assumptions.
Consequently, the cost for random-walk sampling and push is $O\left(\frow(\M)^{1/2} \|\t\|_1 \|\b\|_1^{1/2} \left\|\D_{\M}^{-1}\b\right\|_{\infty}^{1/2} L^{5/2} \eps^{-1}\right)$, completing the proof.
\end{proof}

\begin{proof}[Proof of Theorem~\ref{thm:bidirectional_RCDD}]
The theorem follows from Lemmas~\ref{lem:bidirectional_RCDD_Hoeffding} and \ref{lem:bidirectional_RCDD_variance}.
\end{proof}

Additionally, similar to Theorem~\ref{thm:push_average}, when $\M$ is RDD and $\b$ is a canonical unit vector, we can obtain the following average-case result.

\begin{theorem} \label{thm:bidirectional_average}
	Suppose that $\M$ is RDD, we can sample from the distribution $|\t|/\|\t\|_1$ in $O(1)$ time, $\|\t\|_1,\nnz(\M),\frow(\M)$ are known, and $\b = \e_w$ for a given $w \in [n]$.
	Then there exists a randomized algorithm that computes an estimate $\hat{x}$ such that $\Pr\left\{\left|\hat{x}-\t^{\top}\x^{\ast}\right| \le \eps \left\|\D_{\M}^{-1}\b\right\|_{\infty}\right\} \ge \frac{3}{4}$ with average time complexity
	\begin{align*}
		O\left(\frow(\M)^{1/3} \big(\nnz(\M)/n\big)^{2/3} \|\t\|_1^{2/3} L^{7/3} \eps^{-2/3}\right) = \tO\left(\frac{\frow(\M)^{1/3}(\nnz(\M)/n)^{2/3} \|\t\|_1^{2/3}}{\gamma^{7/3} \eps^{2/3}} \right),
	\end{align*}
	where the average is taken over all choices of $w \in [n]$.
\end{theorem}

\begin{proof}[Proof of Theorem~\ref{thm:bidirectional_average}]
We use Algorithm~\ref{alg:bidirectional_RDD}.
Recall that the absolute value of each sampled value is at most $\frac{1}{2} \|\t\|_1 L^2 \cdot \rmax$.
Applying the Hoeffding bound gives
\begin{align*}
	& \phantom{{}={}} \Pr\left\{ \left|\hat{x} - \t^{\top}\x^{\ast}_L\right| \ge \frac{1}{2} \eps \left\|\D_{\M}^{-1}\b\right\|_{\infty} \right\} = \Pr\left\{ \left|\hat{x} - \t^{\top}\x^{\ast}_L\right| \ge \frac{1}{2} \eps \cdot \frac{1}{d_{\M}(w)} \right\} \\
	& \le 2\exp\left(-\frac{2\ns (\frac{1}{2}\eps \cdot \frac{1}{d_{\M}(w)})^2}{(\|\t\|_1 L^2 \cdot \rmax)^2}\right) = 2\exp\left(-\frac{\ns \cdot \eps^2}{2 \|\t\|_1^2 d_{\M}(w)^2 L^4 \cdot \rmax^2}\right).
\end{align*}
To guarantee that this probability is at most $1/4$, we set $\ns := \Theta\left(\|\t\|_1^2 d_{\M}(w)^2 L^4 \cdot \rmax^2 / \eps^2\right)$, so the cost for random-walk sampling is $O\left(\frow(\M)\|\t\|_1^2 d_{\M}(w)^2 L^5 \cdot \rmax^2 / \eps^2\right)$.

We set $\rmax := \frac{(\nnz(\M)/n)^{1/3} \eps^{2/3}}{\frow(\M)^{1/3} \|\t\|_1^{2/3} d_{\M}(w) L^{4/3}}$.
By Lemma~\ref{lem:push_cost}, the average push cost over all choices of $w \in [n]$ can be upper bounded by big-$O$ of
\begin{align*}
	& \phantom{{}={}} \frac{1}{n} \sum_{w \in [n]}\frac{\frow(\M)^{1/3} \|\t\|_1^{2/3} d_{\M}(w) L^{4/3}}{(\nnz(\M)/n)^{1/3}\eps^{2/3}}\sum_{v \in [n]} \big\|\M(\cdot,v)\big\|_0 \cdot \e_v^{\top} \sum_{\ell=0}^{L-1}\left(\frac{1}{2}\left(\I+\D_{\M}^{-1}\left|\A_{\M}^{\top}\right|\right)\right)^{\ell} \D_{\M}^{-1}\e_w.
\end{align*}
Similar to the derivation in the proof of Theorem~\ref{thm:push_average}, we have
\begin{align*}
	\sum_{w \in [n]}d_{\M}(w) \sum_{v\in [n]} \big\|\M(\cdot,v)\big\|_0 \cdot \e_v^{\top} \sum_{\ell=0}^{L-1}\left(\frac{1}{2}\left(\I+\D_{\M}^{-1}\left|\A_{\M}^{\top}\right|\right)\right)^{\ell} \D_{\M}^{-1}\e_w \le \nnz(\M) \cdot L.
\end{align*}
By substituting, the average push cost becomes $O\left(\frow(\M)^{1/3} \big(\nnz(\M)/n\big)^{2/3} \|\t\|_1^{2/3} L^{7/3} \eps^{-2/3}\right)$.
By our setting of $\rmax$, the random-walk sampling cost is bounded by the same expression, completing the proof.
\end{proof}

Furthermore, we can prove results analogous to Theorem~\ref{thm:MC_relative} when $\M$ is RDDZ and $\b,\t \ge \vzero$.
We present an average-case complexity result for the RDDZ case and a worst-case result for the RCDDZ case in the next two theorems.

\begin{theorem} \label{thm:bidirectional_relative_average}
	Suppose that $\M$ is RDDZ, $\b,\t \ge \vzero$, we can sample from the distribution $\t/\|\t\|_1$ in $O(1)$ time, $\|\t\|_1,\nnz(\M),\frow(\M)$ are known, $\b = \e_w$ for a given $w \in [n]$, and $\t^{\top}\x^{\ast} \ge \eta$ for a given $\eta > 0$.
	Then there exists a randomized algorithm that computes an estimate $\hat{x}$ such that $\Pr\left\{\left|\hat{x}-\t^{\top}\x^{\ast}\right| \le \eps \left\|\D_{\M}^{-1}\b\right\|_{\infty}^{1/2} \cdot \t^{\top}\x^{\ast}\right\} \ge \frac{3}{4}$ with average time complexity
	\begin{align*}
		O\left(\frow(\M)^{1/2} \big(\nnz(\M)/n\big)^{1/2} \|\t\|_1^{1/2} L^2 \eps^{-1} \eta^{-1/2}\right) = \tO\left(\frac{\frow(\M)^{1/2} (\nnz(\M)/n)^{1/2} \|\t\|_1^{1/2}}{\gamma^2 \eps \eta^{1/2}} \right),
	\end{align*}
	where the average is taken over all choices of $w \in [n]$.
\end{theorem}

\begin{proof}
We use Algorithm~\ref{alg:bidirectional_RDD}.
By Inequality~\eqref{eqn:var_bound} derived in the proof of Lemma~\ref{lem:bidirectional_RCDD_variance}, we have
\begin{align*}
	\Var[\hat{x}] & \le \frac{\|\t\|_1 L^2}{4\ns} \cdot \rmax \cdot \t^{\top}\sum_{\ell=0}^{L-1}\left(\frac{1}{2}\left(\I+\D_{\M}^{-1}\A_{\M}^{\top}\right)\right)^{\ell}\D_{\M}^{-1}\b \le \frac{\|\t\|_1 L^2}{2\ns} \cdot \rmax \cdot \t^{\top}\x^{\ast}.
\end{align*}
By Chebyshev's inequality, it follows that
\begin{align*}
	& \phantom{{}={}} \Pr\left\{\left|\hat{x}-\t^{\top}\x^{\ast}\right| \ge \eps \left\|\D_{\M}^{-1}\b\right\|_{\infty}^{1/2} \cdot \t^{\top}\x^{\ast}\right\} \le \frac{\|\t\|_1 L^2 \cdot \rmax \cdot \t^{\top}\x^{\ast}}{2\ns \cdot \eps^2 \left\|\D_{\M}^{-1}\b\right\|_{\infty} \left(\t^{\top}\x^{\ast}\right)^2} \\
	& = \frac{\|\t\|_1 d_{\M}(w) L^2 \cdot \rmax}{2\ns \cdot \eps^2 \cdot \t^{\top}\x^{\ast}} \le \frac{\|\t\|_1 d_{\M}(w) L^2 \cdot \rmax}{2\ns \cdot \eps^2 \eta}.
\end{align*}
To guarantee that this probability is at most $1/4$, we set $\ns := \Theta\left(\|\t\|_1 d_{\M}(w) L^2 \cdot \rmax \cdot \eps^{-2} \eta^{-1}\right)$, so the cost for random-walk sampling is $O\left(\frow(\M)\|\t\|_1 d_{\M}(w) L^3 \cdot \rmax \cdot \eps^{-2} \eta^{-1}\right)$.

We set $\rmax := \frac{(\nnz(\M)/n)^{1/2} \eps \eta^{1/2}}{\frow(\M)^{1/2} \|\t\|_1^{1/2} d_{\M}(w) L}$.
Following the derivation in the proof of Theorem~\ref{thm:bidirectional_average}, the average push cost over all choices of $w \in [n]$ can be upper bounded by big-$O$ of
\begin{align*}
	\frac{1}{n} \cdot \frac{\frow(\M)^{1/2} \|\t\|_1^{1/2} L}{(\nnz(\M)/n)^{1/2} \eps \eta^{1/2}} \cdot \nnz(\M) \cdot L = \frow(\M)^{1/2} \big(\nnz(\M)/n\big)^{1/2} \|\t\|_1^{1/2} L^2 \eps^{-1}\eta^{-1/2}.
\end{align*}
By our setting of $\rmax$, the random-walk sampling cost is bounded by the same expression, completing the proof.
\end{proof}

\begin{theorem} \label{thm:bidirectional_relative}
	Suppose that $\M$ is RCDDZ, its nonzero entries have absolute values of $\Omega(1)$, $\b,\t \ge \vzero$, we can sample from the distribution $\t/\|\t\|_1$ in $O(1)$ time, $\|\t\|_1,\frow(\M)$ are known, we can scan through the nonzero entries of $\b$ in $O\big(\|\b\|_0\big)$ time, and $\t^{\top}\x^{\ast} \ge \eta$ for a given $\eta > 0$.
	Then there exists a randomized algorithm that computes an estimate $\hat{x}$ such that $\Pr\left\{\left|\hat{x}-\t^{\top}\x^{\ast}\right| \le \eps \cdot \t^{\top}\x^{\ast}\right\} \ge \frac{3}{4}$ in time $O\big(\|\b\|_0\big)$ plus
	\begin{align*}
		O\left( \frow(\M)^{1/2} \|\t\|_1^{1/2} \|\b\|_1^{1/2} L^2 \eps^{-1} \eta^{-1/2}\right) = \tO\left(\frac{\frow(\M)^{1/2} \|\t\|_1^{1/2} \|\b\|_1^{1/2}}{\gamma^2 \eps \eta^{1/2}} \right).
	\end{align*}
\end{theorem}

\begin{proof}
We use Algorithm~\ref{alg:bidirectional_RDD}.
According to the preivous proof, we have $\Var[\hat{x}] \le \frac{\|\t\|_1 L^2}{2\ns} \cdot \rmax \cdot \t^{\top}\x^{\ast}$, so applying Chebyshev's inequality gives
\begin{align*}
	\Pr\left\{\left|\hat{x}-\t^{\top}\x^{\ast}\right| \ge \eps \cdot \t^{\top}\x^{\ast}\right\} \le \frac{\|\t\|_1 L^2 \cdot \rmax \cdot \t^{\top}\x^{\ast}}{2\ns \cdot \eps^2 \left(\t^{\top}\x^{\ast}\right)^2} = \frac{\|\t\|_1 L^2 \cdot \rmax}{2\ns \cdot \eps^2 \cdot \t^{\top}\x^{\ast}} \le \frac{\|\t\|_1 L^2 \cdot \rmax}{2\ns \cdot \eps^2 \eta}.
\end{align*}
To guarantee that this probability is at most $1/4$, we set $\ns := \Theta\left(\|\t\|_1 L^2 \cdot \rmax \cdot \eps^{-2} \eta^{-1}\right)$, so the cost for random-walk sampling is $O\left(\frow(\M)\|\t\|_1 L^3 \cdot \rmax \cdot \eps^{-2} \eta^{-1}\right)$.

Since $\M$ is CDD and its nonzero entries have absolute values of $\Omega(1)$, by Lemma~\ref{lem:push_cost_CDD}, the push cost is $O\big(\|\b\|_1 L/\rmax\big)$.
We set $\rmax := \frac{\|\b\|_1^{1/2} \eps \eta^{1/2}}{\frow(\M)^{1/2} \|\t\|_1^{1/2} L}$.
Now, the cost for random-walk sampling and push both becomes $O\left( \frow(\M)^{1/2} \|\t\|_1^{1/2} \|\b\|_1^{1/2} L^2 \eps^{-1} \eta^{-1/2}\right)$, completing the proof.
\end{proof}

\section{Connections with PageRank Computation} \label{sec:PageRank}

This section discusses the connections between our framework and PageRank computation, establishing a unified understanding of the \fpush and \bpush algorithms and presenting proofs for Theorems~\ref{thm:PageRank_Eulerian} and \ref{thm:lower_bound_eps} along with some additional results.

As mentioned in Section~\ref{sec:introduction}, the PPR equations~\eqref{eqn:PPR_I} and \eqref{eqn:PPR_D} and PageRank contribution equations~\eqref{eqn:PageRank_contribution_I} and \eqref{eqn:PageRank_contribution_D} can be formulated as RDD/CDD systems.
For example, by Equations~\eqref{eqn:PPR_I} and \eqref{eqn:PageRank_contribution_I}, for a node $t \in V$, setting $\M = \I-(1-\alpha)\A_G^{\top}\D_G^{-1}, \b = \frac{\alpha}{n}\vone, \t = \e_t$ or $\M = \I-(1-\alpha)\D_G^{-1}\A_G, \b = \alpha\e_t, \t = \frac{1}{n}\vone$ in our formulation both yield $\t^{\top}\x^{\ast} = \vpi_{G,\alpha}(t)$; by Equation~\eqref{eqn:PPR_D}, for nodes $s,t \in V$, setting $\M = \D_G-(1-\alpha)\A_G^{\top}, \b = \alpha\e_s, \t = \e_t$ yields $\t^{\top}\x^{\ast} = \vpi_{G,\alpha}(s,t) / \dout_G(t)$.
It is worth noting that on Eulerian graphs, Equations~\eqref{eqn:PPR_D} and \eqref{eqn:PageRank_contribution_D} are RCDD systems, and on undirected graphs, they are SDD systems.

A subtle issue is that PageRank computation often involves graphs with self-loops, making the decomposition of matrices like $\D_G - (1-\alpha)\A_G^{\top}$ into diagonal and off-diagonal parts potentially inconsistent with the apparent decomposition into $\D_G$ and $(1-\alpha)\A_G^{\top}$.
Nevertheless, we enforce a decomposition of these matrices that yields the ``diagonal'' part as $\D_G$ or $\I$ (so the corresponding matrix $\A_{\M}$ may have nonzero diagonal entries), and one can verify that using this type of decomposition for PageRank computation preserves our theoretical results.

By the definition of the $p$-norm gaps, we have
\begin{align*}
	& \phantom{{}={}} \gamma_1\left(\I-(1-\alpha)\A_G^{\top}\D_G^{-1}\right) = 1 - \left\|\left.\frac{1}{2}\left(\I+(1-\alpha)\A_G^{\top}\D_G^{-1}\right)\right|_{\range\left(\I-(1-\alpha)\A_G^{\top}\D_G^{-1}\right)}\right\|_1 \\
	& = 1 - \left\|\frac{1}{2}\left(\I+(1-\alpha)\A_G^{\top}\D_G^{-1}\right)\right\|_1 = 1-\frac{1}{2}\big(1+(1-\alpha)\big) = \frac{1}{2}\alpha,
\end{align*}
where we used $\range\left(\I-(1-\alpha)\A_G^{\top}\D_G^{-1}\right) = \R^n$ since $\I-(1-\alpha)\A_G^{\top}\D_G^{-1}$ is invertible.
Similarly, we have $\gamma_1\left(\D_G-(1-\alpha)\A_G^{\top}\right) = \gamma_{\infty}\left(\I-(1-\alpha)\D_G^{-1}\A_G\right) = \gamma_{\infty}\left(\D_G-(1-\alpha)\A_G\right) = \frac{1}{2}\alpha$.
Thus, $\frac{1}{2}\alpha$ serves as a lower bound on the maximum $p$-norm gap of all these matrices involved in the PPR and PageRank contribution equations.

\subsection{Understanding \fpush and \bpush} \label{sec:push_relationship}

Now we introduce a unified understanding of the \fpush and \bpush algorithms described in Section~\ref{sec:prelim_push} through the lens of our \push primitive (Algorithm~\ref{alg:push}).

We first show that \fpush (Algorithm~\ref{alg:FP}) and \bpush (Algorithm~\ref{alg:BP}) are equivalent to applying \push to the PPR equation and PageRank contribution equation, respectively.
As stated in the following claim, executing \fpush is equivalent to applying \push to Equation~\eqref{eqn:PPR_D} with a degree-scaled version of the residue and reserve vectors.
This scaling stems from the fact that the solution to Equation~\eqref{eqn:PPR_D} is the degree-normalized PPR vector $\D_G^{-1}\vpi_{G,\alpha,\s}$ instead of $\vpi_{G,\alpha,\s}$ itself.

\begin{claim}
	Executing \fpush from a source node $s \in V$ is equivalent to applying \push to Equation~\eqref{eqn:PPR_D} with $\s = \e_s$, where \fpush maintains degree-scaled versions of the variables: $\pFP^{(\ell)}(v) = \dout_G(v) \cdot \p^{(\ell)}(v)$ and $\rFP^{(\ell)}(v) = \dout_G(v) \cdot \r^{(\ell)}(v)$ for each $\ell \in [0,L-1]$ and $v \in V$.
	Here, for comparison, $\pFP$ and $\rFP$ denote the reserve and residue vectors maintained by \fpush, while $\p$ and $\r$ denote those maintained by \push.
	In other words, $\fpush(G,s,\alpha,L,\rmax)$ operates on these scaled variables and is equivalent to $\push\left(\D_G-(1-\alpha)\A_G^{\top},\alpha\e_s,L,\rmax\right)$ with the corresponding variable transformations.
\end{claim}

\begin{proof}
The result follows by verifying that the described variable transformations convert the procedure of $\push\left(\D_G-(1-\alpha)\A_G^{\top},\alpha\e_s,L,\rmax\right)$ into the procedure of $\fpush(G,s,\alpha,L,\rmax)$ exactly.
For the initialization, \push sets $\r^{(0)}(s) = \alpha/\dout_G(s)$, so $\rFP^{(\ell)}(v) = \alpha$, matching the initialization in \fpush.
The condition for performing the push operation on a coordinate $v$ at level $\ell$ in \push is $\left|\r^{(\ell)}(v)\right| > \rmax$, which is equivalent to $\frac{\rFP^{(\ell)}(v)}{\dout_G(v)} > \rmax$ since the residues in \fpush are nonnegative, matching the condition in \fpush.
In the push operation on $v$ at level $\ell$, \push increments $\r^{(\ell+1)}(u)$ by $\frac{(1-\alpha)\A_G(v,u)}{2\dout_G(u)} \cdot \r^{(\ell)}(v)$ for each $u$ with $\A_G(v,u) > 0$.
This is equivalent to incrementing $\rFP^{(\ell+1)}(u)$ by $\dout_G(u) \cdot \frac{(1-\alpha)\A_G(v,u)}{2\dout_G(u)} \cdot \r^{(\ell)}(v) = \frac{1}{2}(1-\alpha)\A_G(v,u) \cdot \r^{(\ell)}(v) = \frac{1}{2}(1-\alpha) \cdot \frac{\A_G(v,u)}{\dout_G(v)} \cdot \rFP^{(\ell)}(v)$ for each out-neighbor $u$ of $v$, matching the corresponding step in \fpush.
This verifies the equivalence of the two algorithms.
\end{proof}

Similarly, one can readily verify the next claim for \bpush.

\begin{claim}
	Executing \bpush from a target node $t \in V$ is equivalent to applying \push to either Equation~\eqref{eqn:PageRank_contribution_I} or \eqref{eqn:PageRank_contribution_D}.
	In other words, the process of $\bpush(G,t,\alpha,L,\rmax)$ is equivalent to both $\push\left(\I-(1-\alpha)\D_G^{-1}\A_G,\alpha\e_t,L,\rmax\right)$ and $\push\big(\D_G-(1-\alpha)\A_G,\alpha\D_G\e_t,L,\rmax\big)$.
\end{claim}

\begin{proof}
Since there is no variable transformation, the verification is straightforward.
We only need to note that in the push operation on $v$ at level $\ell$, \push increments $\r^{(\ell+1)}(u)$ by $\frac{(1-\alpha)\A_G(u,v)}{2\dout_G(u)} \cdot \r^{(\ell)}(v)$ for each $u$ with $\A_G(u,v) > 0$, which is exactly the procedure performed by \bpush for each in-neighbor $u$ of $v$.
\end{proof}

It is constructive to ask why \fpush is formulated as applying \push to Equation~\eqref{eqn:PPR_D} instead of Equation~\eqref{eqn:PPR_I}.
An explanation is that Equation~\eqref{eqn:PPR_D} has superior properties on some special graphs.
As mentioned, when $G$ is Eulerian, Equation~\eqref{eqn:PPR_D} is an RCDD system, which yields better algorithmic guarantees for \fpush according to our analysis of \push on RCDD systems.

These characterizations explain the distinct behaviors of \fpush and \bpush.
Recall that by Lemmas~\ref{lem:push_error} and \ref{lem:push_cost_CDD}, \push provides closed-form accuracy guarantees for RDD systems and running time bounds for special CDD systems.
As \fpush applies \push to the CDD system \eqref{eqn:PPR_D} and \bpush applies \push to the RDD system \eqref{eqn:PageRank_contribution_D}, this understanding aligns well with the known properties of \fpush and \bpush on unweighted directed graphs.
Our Theorem~\ref{thm:push_average} also corresponds to the average-case analysis of \bpush.
Moreover, on unweighted undirected graphs, both algorithms correspond to applying \push to SDD systems, so they both enjoy accuracy guarantees and worst-case running time bounds, where the accuracy guarantee for \fpush takes a degree-normalized form due to variable transformations.
From this viewpoint, the fundamental property underlying these combined advantages is that the system becomes RCDD, which occurs when the graph is Eulerian.

These connections also reveal that on Eulerian graphs, performing \fpush from a node is equivalent to performing \bpush from the same node on the transpose graph $G^{\top}$ (which reverses all edge directions), with appropriate variable scaling and threshold adjustment.
We state this result in the next claim.

\begin{claim}
	For any Eulerian graph $G$ and $s \in V$, executing $\fpush(G,s,\alpha,L,\rmax)$ is equivalent to executing $\bpush\left(G^{\top},s,\alpha,L,\rmax \cdot d_G(s)\right)$, up to variable transformations.
\end{claim}

\begin{proof}
When $G$ is Eulerian, the outdegree matrix of $G^{\top}$ equals $\D_G$.
Thus, by the previous claims, \fpush on $G$ corresponds to $\push\left(\D_G-(1-\alpha)\A_G^{\top},\alpha\e_s,L,\rmax\right)$, while \bpush on $G^{\top}$ with threshold $\rmax \cdot d_G(s)$ is equivalent to $\push\left(\D_G-(1-\alpha)\A_G^{\top},\alpha\D_G\e_s,L,\rmax \cdot d_G(s)\right)$.
In the two sets of parameters for \push, changing $\b = \alpha\e_s$ to $\b = \alpha\D_G\e_s$ scales the initial residue $\r^{(0)}(s)$ by $d_G(s)$, which in turn scales all the reserves and residues by $d_G(s)$.
As the threshold $\rmax$ is also scaled by $d_G(s)$, the procedures of the two algorithms remain equivalent.
\end{proof}

As a corollary, on undirected graphs, performing \fpush and \bpush from the same node are algorithmically equivalent.
This indicates that for local PageRank and effective resistance computation on undirected graphs~\cite{lofgren2015bidirectional,cui2025mixing,yang2025improved}, the choice between \fpush and \bpush is essentially immaterial.

\subsection{Results for PageRank Computation when $\D_G-(1-\alpha)\A_G^{\top}$ is RCDD}

Since the PPR and PageRank contribution equations can be formulated as special RDD/CDD systems, our results for solving general RDD/CDD systems can be applied to PageRank computation.
For instance, consider Equation~\eqref{eqn:PPR_D} with $\t = \e_t$ for some $t \in V$, so the solution is $\vpi_{G,\alpha,\s}(t) / \dout(t)$.
Applying the ``CDD counterparts'' of Theorems~\ref{thm:MC_cubic} and \ref{thm:MC_square} to it corresponds to approximating $\vpi_{G,\alpha,\s}(t)$ using random-walk sampling from $\s$; applying the ``CDD counterpart'' of Theorem~\ref{thm:bidirectional_relative_average} to it corresponds to using a bidirectional method combining \push from $t$ and random-walk sampling from $\s$, and the obtained average-case complexity bound mirrors the result in \cite{lofgren2016personalized}; applying Theorem~\ref{thm:bidirectional_relative} to it when the system is RCDD mirrors the result on undirected graphs in \cite{lofgren2015bidirectional}.
We emphasize that our results are more general and in particular apply to weighted graphs.

However, directly applying our results to PageRank computation does not necessarily yield improvements, since the application of our techniques to PageRank computation has been extensively studied, and the special properties of PageRank enable simpler and more efficient algorithms than ours for general RDD/CDD systems.
In particular, PageRank algorithms can exploit non-truncated non-lazy random walks that terminate with probability $\alpha$ at each step and avoid partitioning the push process into multiple levels, which can remove the logarithmic factors in our setting of $L$ and potentially achieve better dependence on $\alpha$ than our general results.

Nevertheless, our framework can still provide new insights and results for PageRank computation, in particular when the involved system is RCDD.
Theorem~\ref{thm:PageRank_Eulerian} stated in the introduction is one such example, which shows that previous results for single-node PageRank computation on undirected graphs can be improved and generalized to Eulerian graphs.

To prove Theorem~\ref{thm:PageRank_Eulerian}, we investigate the case when the matrix $\D_G-(1-\alpha)\A_G^{\top}$ in Equation~\eqref{eqn:PPR_D} is RCDD.
This matrix is CDD, and it is also RDD if $\dout_G(v) \ge (1-\alpha)\din_G(v)$ holds for all $v \in V$.
In particular, this condition holds when $G$ is Eulerian or $\alpha$ is large enough.
Now, by applying Theorem~\ref{thm:MC_relative}, we directly obtain the following result.

\begin{theorem} \label{thm:PageRank_RDD}
	For any unweighted graph $G$ and decay factor $\alpha$, suppose that $\dout_G(v) \ge (1-\alpha)\din_G(v)$ holds for all $v \in V$.
	Then there exists a randomized algorithm that, given $t \in V$, $\deltaout_G$, and accuracy parameter $\eps$, computes an estimate of $\vpi_{G,\alpha}(t)$ within relative error $\eps$ with success probability at least $3/4$ in time
	\begin{align*}
		\tO\left(\frac{1}{\alpha\eps^2} \cdot \frac{\dout_G(t)}{\deltaout_G} \cdot \frac{1}{n \vpi_{G,\alpha}(t)}\right),
	\end{align*}
	where $\tO$ hides $\polylog\left(\frac{n}{\alpha\eps}\right)$ factors.
\end{theorem}

\begin{proof}
Consider applying Theorem~\ref{thm:MC_relative} to Equation~\eqref{eqn:PPR_D} with $\s = \frac{1}{n}\vone$ and $\t = \e_t$.
Note that the corresponding $\left\|\D_{\M}^{-1}\b\right\|_{\infty}$ equals $\frac{\alpha}{n \deltaout_G}$, $L = \tO(1/\alpha)$, $\frow(\M) = O(1)$, $\|\t\|_1 = 1$, and the obtained $(1 \pm \eps)$-multiplicative approximation of $\t^{\top}\x^{\ast} = \vpi_{G,\alpha}(t) / \dout_G(t)$ directly yields a $(1 \pm \eps)$-multiplicative approximation of $\vpi_{G,\alpha}(t)$.
Therefore, the time complexity is
\begin{align*}
	\tO\left(\frac{\alpha/(n \deltaout_G)}{\alpha^2 \eps^2 \cdot \vpi_{G,\alpha}(t) / \dout_G(t)}\right) = \tO\left(\frac{1}{\alpha\eps^2} \cdot \frac{\dout_G(t)}{\deltaout_G} \cdot \frac{1}{n \vpi_{G,\alpha}(t)}\right),
\end{align*}
as desired.
\end{proof}

To prove Theorem~\ref{thm:PageRank_Eulerian}, we establish some lower bounds on PageRank values in the next lemma, which may be of independent interest.
This lemma is partly inspired by \cite[Lemma 5.13]{bressan2023sublinear}, \cite{wang2023estimating,wang2024revisitinga}, and \cite[Theorem 1.1]{wang2024revisiting}.

\begin{lemma} \label{lem:PageRank_lower_bound}
	For any weighted directed graph $G$ and $t \in V$, we have
	\begin{align*}
		\vpi_{G,\alpha}(t) \ge \max\left(\frac{\alpha}{n}, \frac{\alpha(1-\alpha) \din_G(t)}{n\Deltaout_G}, \frac{\alpha(1-\alpha)\din_G(t)^2}{n\big\|\A_G(\cdot,t)\big\|_{\infty}\|\A_G\|_{1,1} }, \frac{\alpha(1-\alpha)\din_G(t)^2}{n\sqrt{n}\big\|\A_G(\cdot,t)\big\|_2\|\A_G\|_{\Fro}}\right),
	\end{align*}
	where $\|\A_G\|_{1,1} := \sum_{u,v \in [n]}\big|\A_G(u,v)\big|$ is the entrywise $1$-norm and $\|\A_G\|_{\Fro} := \sqrt{\sum_{u,v \in V}\A_G(u,v)^2}$ is the Frobenius norm.
	If $G$ is Eulerian, we further have $\vpi_{G,\alpha}(t) \ge \frac{d_G(t)}{n\Delta_G}$; if $G$ is unweighted Eulerian, we further have $\vpi_{G,\alpha}(t) \ge \frac{\sqrt{1-\alpha} \cdot d_G(t)}{n\sqrt{m}}$.
\end{lemma}

The proof of this lemma is given in Appendix~\ref{sec:deferred_proofs}.
Now we present the proof of Theorem~\ref{thm:PageRank_Eulerian}.

\begin{proof}[Proof of Theorem~\ref{thm:PageRank_Eulerian}]
By Theorem~\ref{thm:PageRank_RDD}, we obtain the complexity bound $\tO\left(\frac{1}{\alpha\eps^2} \cdot \frac{d_G(t)}{\delta_G} \cdot \frac{1}{n \vpi_{G,\alpha}(t)}\right)$.
By Lemma~\ref{lem:PageRank_lower_bound}, on unweighted Eulerian graphs, we have
\begin{align*}
	\vpi_{G,\alpha}(t) & \ge \max\left(\frac{\alpha}{n}, \frac{\alpha(1-\alpha)\din_G(t)^2}{n\big\|\A_G(\cdot,t)\big\|_{\infty}\|\A_G\|_{1,1} }, \frac{d_G(t)}{n\Delta_G}, \frac{\sqrt{1-\alpha} \cdot d_G(t)}{n\sqrt{m}}\right) \\
	& = \max\left(\frac{\alpha}{n}, \frac{\alpha(1-\alpha)d_G(t)^2}{nm}, \frac{d_G(t)}{n\Delta_G}, \frac{\sqrt{1-\alpha} \cdot d_G(t)}{n\sqrt{m}}\right).
\end{align*}
By plugging these lower bounds on $\vpi_{G,\alpha}(t)$ into the complexity bound, we obtain the desired results up to $\polylog\left(\frac{n}{\alpha\eps}\right)$ factors (where we omit the terms of $1/(1-\alpha)$ since we often consider the case when $\alpha \to 0$).
These $\polylog\left(\frac{n}{\alpha\eps}\right)$ factors can be removed by using non-truncated random walks for sampling (cf. \cite{wang2024revisitinga}), leading to the stated complexity bounds.
\end{proof}

Additionally, we can prove the following result.

\begin{theorem} \label{thm:PageRank_RDD_stronger}
	For any unweighted graph $G$ and constant $\alpha$, suppose that $(1-\alpha)\Deltain_G \le 1$.
	Then there exists a randomized algorithm that, given $t \in V$, and accuracy parameter $\eps$, computes an estimate of $\vpi_{G,\alpha}(t)$ within relative error $\eps$ with success probability at least $3/4$ in time
	\begin{align*}
		\tO\left(\frac{1}{\eps^2 n \vpi_{G,\alpha}(t)}\right),
	\end{align*}
	where $\tO$ hides $\polylog(n/\eps)$ factors.
\end{theorem}

Note that this complexity can also be bounded by $\tO\left(1/\eps^2\right)$ since $\vpi_{G,\alpha}(t) \ge \alpha/n$.
This matches a result given in a concurrent work~\cite{thorup2026pagerank}, which states that if $(1-\alpha)\Deltain_G \le 1$ for an unweighted graph $G$, then we can approximate single-node PageRank within relative error $\eps$ in $\tO\left(1\right)$ time, for constant $\alpha$ and $\eps$.
Compared to their algorithm, our algorithm is simpler and only performs random-walk sampling from $t$.

\begin{proof}[Proof of Theorem~\ref{thm:PageRank_RDD_stronger}]
If $\din_G(t) = 0$, then $\vpi_{G,\alpha}(t) = \frac{\alpha}{n}$, which can be computed in $O(1)$ time and the theorem holds trivially.
Now suppose that $\din_G(t) \ge 1$.
Recall that the PageRank equation~\eqref{eqn:PPR_I} is $\left(\I-(1-\alpha)\A_G^{\top}\D_G^{-1}\right)\vpi_{G,\alpha} = \frac{\alpha}{n} \vone$.
Define $\X_G$ as the diagonal matrix with $\X_G(v,v) := (1-\alpha)\max\big(\din_G(v),1\big)$ for each $v \in V$.
Then the PageRank equation is equivalent to
\begin{align}
	\left(\X_G-(1-\alpha)\A_G^{\top}\D_G^{-1}\X_G\right)\left(\X_G^{-1}\vpi_{G,\alpha}\right) = \frac{\alpha}{n} \vone. \label{eqn:PPR_X}
\end{align}
Now we show that the matrix $\M := \X_G-(1-\alpha)\A_G^{\top}\D_G^{-1}\X_G$ is RDDZ.
To this end, it suffices to show that each row sum of $\M$ is nonnegative.
For each $u \in V$, the $u$-th row sum of $\M$ is
\begin{align*}
	& \phantom{{}={}} \X_G(u,u) - (1-\alpha) \sum_{v \in V} \A_G(v,u) \cdot \frac{\X_G(v,v)}{\dout_G(v)} \\
	& = (1-\alpha)\max\big(\din_G(u),1\big) - (1-\alpha) \sum_{v \in V} \A_G(v,u) \cdot \frac{(1-\alpha)\max\big(\din_G(v),1\big)}{\dout_G(v)} \\
	& \ge (1-\alpha)\din_G(u) - (1-\alpha) \sum_{v \in V} \A_G(v,u) = 0,
\end{align*}
as desired.
Thus, we can apply Theorem~\ref{thm:MC_relative} to Equation~\eqref{eqn:PPR_X}.
In this case, we have $\left\|\D_{\M}^{-1}\b\right\|_{\infty} \le \frac{\alpha}{n(1-\alpha)}$, $L = \tO(1)$ since $\gamma_1(\M) = \frac{1}{2}\alpha = \Omega(1)$, $\frow(\M) = O(1)$ since $\Deltain_G = O(1)$, $\|\t\|_1 = 1$, and the obtained $(1 \pm \eps)$-multiplicative approximation of $\t^{\top}\x^{\ast} = \frac{\vpi_{G,\alpha}(t)}{(1-\alpha)\din_G(t)}$ directly yields a $(1 \pm \eps)$-multiplicative approximation of $\vpi_{G,\alpha}(t)$.
Therefore, the time complexity is
\begin{align*}
	\tO\left(\frac{1/n}{\eps^2 \vpi_{G,\alpha}(t) / \din_G(t)}\right) = \tO\left(\frac{1}{\eps^2 n \vpi_{G,\alpha}(t)}\right),
\end{align*}
as desired.
\end{proof}

\subsection{A Lower Bound on the Accuracy Parameter for SDD Solvers}

This subsection proves Theorem~\ref{thm:lower_bound_eps}.
To this end, we establish the following reduction from single-node PageRank computation on undirected graphs to solving SDD systems.

\begin{lemma} \label{lem:reduction_PageRank_to_SDD}
	Suppose that there exists a randomized algorithm that computes an estimate $\hat{x}_t$ such that $\Pr\left\{\big|\hat{x}_t-\x^{\ast}(t)\big| \le \eps\|\x^{\ast}\|_{\infty}\right\} \ge \frac{3}{4}$ for any SDD system $\S\x = \b$ in $O\left(\gamma^{-\nu}\eps^{-\tau}\right)$ time.
	Then there exists a randomized algorithm that, given $\delta_G$, estimates $\vpi_{G,\alpha}(t)$ on unweighted undirected graphs $G$ within constant relative error and with success probability at least $3/4$ in time $O\left(\big(d_G(t)/\delta_G\big)^{\tau}/\alpha^{\nu+\tau}\right)$.
\end{lemma}

To prove this lemma, we use the following upper bound on $\vpi_{G,\alpha}(t)$ on Eulerian graphs, whose proof is given in Appendix~\ref{sec:deferred_proofs}.

\begin{lemma} \label{lem:PageRank_upper_bound}
On any Eulerian graph $G$ and $v \in V$, we have $\vpi_{G,\alpha}(v) \le \frac{d_G(v)}{n \delta_G}$.
\end{lemma}

\begin{proof}[Proof of Lemma~\ref{lem:reduction_PageRank_to_SDD}]
Consider the PageRank equation~\eqref{eqn:PPR_D} with $\s = 1/n \cdot \vone$.
When $G$ is undirected, the matrix $\D_G-(1-\alpha)\A_G^{\top}$ is SDD.
By setting $\gamma := \alpha/2$ and $\eps := \Theta\big(\alpha \delta_G / d_G(t)\big)$, the supposed algorithm can compute an estimate $\hat{x}_t$ such that $\left|\hat{x}_t-\frac{\vpi_{G,\alpha}(t)}{d_G(t)}\right| \le \eps \cdot \max_{v \in V}\left\{\frac{\vpi_{G,\alpha}(v)}{d_G(v)}\right\}$ with probability at least $3/4$ in time $O\left(\gamma^{-\nu}\eps^{-\tau}\right) = O\left(\big(d_G(t)/\delta_G\big)^{\tau}/\alpha^{\nu+\tau}\right)$.
Using Lemma~\ref{lem:PageRank_upper_bound} and $\vpi_{G,\alpha}(t) \ge \alpha/n$, we have $\max_{v \in V}\left\{\frac{\vpi_{G,\alpha}(v)}{d_G(v)}\right\} \le \frac{1}{n \delta_G} \le \frac{\vpi_{G,\alpha}(t)}{\alpha \delta_G}$.
Thus, with probability at least $3/4$, $\big|d_G(t) \cdot \hat{x}_t-\vpi_{G,\alpha}(t)\big| \le \eps \cdot d_G(t) \cdot \frac{\vpi_{G,\alpha}(t)}{\alpha \delta_G} = \Theta\big(\vpi_{G,\alpha}(t)\big)$, so $d_G(t) \cdot \hat{x}_t$ is an estimate of $\vpi_{G,\alpha}(t)$ within constant relative error.
This completes the proof.
\end{proof}

\begin{proof}[Proof of Theorem~\ref{thm:lower_bound_eps}]
\cite{wang2024revisitinga} establishes a complexity lower bound of $\Omega\big(d_G(t)/\delta_G\big)$ for estimating $\vpi_{G,\alpha}(t)$ within constant relative error with constant success probability on unweighted undirected graphs, where $\alpha$ is constant and the bound holds for any possible combination of $\delta_G$ and $d_G(t)$.
This lower bound applies to the number of queries to the graph structure.
Therefore, combining this lower bound with Lemma~\ref{lem:reduction_PageRank_to_SDD} and noting that the reduction uses $\eps := \Theta\big(\alpha \delta_G / d_G(t)\big) = \Omega(1/n)$ yield the desired lower bound of $\Omega(1/\eps)$ for $\eps = \Omega(1/n)$.
\end{proof}

In fact, we can also prove the following similar reduction for the error bound $\eps\left\|\D_{\S}^{-1}\b\right\|_{\infty}$, which may be of independent interest.

\begin{lemma}
	Suppose that there exists a randomized algorithm that computes an estimate $\hat{x}_t$ such that $\Pr\left\{\big|\hat{x}_t-\x^{\ast}(t)\big| \le \eps\left\|\D_{\S}^{-1}\b\right\|_{\infty}\right\} \ge \frac{3}{4}$ for any SDD system $\S\x = \b$ in $O\left(\gamma^{-\nu}\eps^{-\tau}\right)$ time.
	Then there exists a randomized algorithm that, given $\delta_G$, estimates $\vpi_{G,\alpha}(t)$ on undirected graphs $G$ within constant relative error and with success probability at least $3/4$ in time $O\left(\big(d_G(t)/\delta_G\big)^{\tau}/\alpha^{\nu}\right)$.
\end{lemma}

\begin{proof}
Following the proof of Lemma~\ref{lem:reduction_PageRank_to_SDD}, by setting $\gamma := \alpha/2$ and $\eps := \Theta\big(\delta_G / d_G(t)\big)$, the supposed algorithm can compute an estimate $\hat{x}_t$ such that $\left|\hat{x}_t-\frac{\vpi_{G,\alpha}(t)}{d_G(t)}\right| \le \eps \cdot \frac{\alpha}{n \delta_G}$ with probability at least $3/4$ in time $O\left(\gamma^{-\nu}\eps^{-\tau}\right) = O\left(\big(d_G(t)/\delta_G\big)^{\tau}/\alpha^{\nu}\right)$.
Using $\vpi_{G,\alpha}(t) \ge \alpha/n$, we have $\frac{\alpha}{n \delta_G} \le \frac{\vpi_{G,\alpha}(t)}{\delta_G}$.
Thus, with probability at least $3/4$, $\big|d_G(t) \cdot \hat{x}_t-\vpi_{G,\alpha}(t)\big| \le \eps \cdot d_G(t) \cdot \frac{\vpi_{G,\alpha}(t)}{\delta_G} = \Theta\big(\vpi_{G,\alpha}(t)\big)$, so $d_G(t) \cdot \hat{x}_t$ is an estimate of $\vpi_{G,\alpha}(t)$ within constant relative error.
This completes the proof.
\end{proof}

\section{Connections with Effective Resistance Computation} \label{sec:effective_resistance}

This section justifies the relationship between our framework and effective resistance computation on graphs in Lemma~\ref{lem:ER_quadratic_form} and proves Corollary~\ref{cor:ER}.

Recall that in the context of computing effective resistances, we assume that $G$ is undirected and connected.
In our framework, we set $\M = \L_G$ and $\b = \t = \e_s-\e_t$.
By Theorem~\ref{thm:p_norm_gap}, $\gamma_{\max}(\L_G) = \gamma(\L_G)$, so a lower bound $\gamma$ on the spectral gap $\gamma(\L_G)$ serves as a lower bound on the maximum $p$-norm gap $\gamma_{\max}(\L_G)$.
The following lemma states that with this setting, the quantity $\t^{\top}\x^{\ast}$ that our algorithms approximate equals the effective resistance $R_G(s,t)$.

\begin{lemma} \label{lem:ER_quadratic_form}
	When $\M = \L_G$ and $\b = \t = \e_s - \e_t$, we have $\t^{\top}\x^{\ast} = R_G(s,t)$.
\end{lemma}

This lemma is nontrivial since although $\M = \L_G$ is SDD, we only have $\x^{\ast} = \D_{\M}^{-1/2}\tM^{+}\D_{\M}^{-1/2}\b$ by Theorem~\ref{thm:x*}, which may not equal $\M^{+}\b = \L_G^{+}\b$.
This result together with Theorem~\ref{thm:x*} provides the theoretical foundation for estimating effective resistance through multi-step random-walk probabilities~\cite{peng2021local,yang2023efficient,cui2025mixing,yang2025improved}.
This connection was established in prior work~\cite[Lemma 4.3]{peng2021local}, but although the conclusion of that lemma is correct, its proof relies on the incorrect assertion that $\D_{G}^{-1/2}\left(\D_{G}^{-1/2} \L_G \D_{G}^{-1/2}\right)^{+}\D_{G}^{-1/2} = \L_G^{+}$.
\cite{andoni2019solving} also does not establish this connection and only applies their algorithm to approximate effective resistances on regular graphs.

In fact, we can prove Lemma~\ref{lem:ER_quadratic_form} by using the results in \cite{bozzo2013moore}, but here we provide a different self-contained proof.
We use the following lemma, which shows that although $\D_{\M}^{-1/2}\tM^{+}\D_{\M}^{-1/2}$ may not equal $\M^{+}$, their quadratic forms of $\b$ are equal whenever $\M$ is SDD and $\b \in \range(\M)$.

\begin{lemma} \label{lem:quadratic_form}
    When $\M$ is SDD, for any $\b \in \range(\M)$, $\b^{\top}\D_{\M}^{-1/2}\tM^{+}\D_{\M}^{-1/2}\b = \b^{\top}\M^{+}\b$.
\end{lemma}

\begin{proof}
Let $\y := \M^{+}\b$ and $\z := \tM^{+}\D_{\M}^{-1/2}\b$.
Since $\b \in \range(\M)$, which implies that $\D_{\M}^{-1/2}\b \in \range(\tM)$, we have $\M\y = \b$ and $\tM\z = \D_{\M}^{-1/2}\b$.
Letting $\z' := \D_{\M}^{1/2}\y = \D_{\M}^{1/2}\M^{+}\b$, we have
\begin{align*}
    \tM\z' = \left(\D_{\M}^{-1/2}\M\D_{\M}^{-1/2}\right)\left(\D_{\M}^{1/2}\y\right) = \D_{\M}^{-1/2}\M\y = \D_{\M}^{-1/2}\b = \tM\z.
\end{align*}
Thus, $\z-\z' \in \ker(\tM)$, so $\z^{\top}\tM\z = \z'^{\top}\tM\z'$ since $\tM$ is symmetric.
Consequently,
\begin{align*}
    & \b^{\top}\D_{\M}^{-1/2}\tM^{+}\D_{\M}^{-1/2}\b = \left(\D_{\M}^{-1/2}\b\right)^{\top} \left(\tM^{+}\D_{\M}^{-1/2}\b\right)= \left(\tM\z\right)^{\top}\z = \z^{\top}\tM\z = \z'^{\top}\tM\z' \\
    = & \left(\D_{\M}^{1/2}\M^{+}\b\right)^{\top}\tM\left(\D_{\M}^{1/2}\M^{+}\b\right) = \b^{\top}\M^{+}\D_{\M}^{1/2}\tM\D_{\M}^{1/2}\M^{+}\b = \b^{\top}\M^{+}\M\M^{+}\b = \b^{\top}\M^{+}\b,
\end{align*}
as desired.
\end{proof}

\begin{proof}[Proof of Lemma~\ref{lem:ER_quadratic_form}]
Since $\e_s-\e_t \in \spansp(\vone)^{\perp} = \ker(\L_G)^{\perp}$, we have $\b = \t = \e_s-\e_t \in \range(\L_G)$.
Thus, Theorem~\ref{thm:x*} and Lemma~\ref{lem:quadratic_form} imply that $\t^{\top}\x^{\ast} = \b^{\top}\D_G^{-1/2}\tL_G^{+}\D_G^{-1/2}\b = \b^{\top}\L_G^{+}\b = R_G(s,t)$.
\end{proof}

Now we can directly apply Theorems~\ref{thm:MC_cubic} and \ref{thm:bidirectional_RCDD} to prove Corollary~\ref{cor:ER}.
The only remaining detail in the proof is to derive a better setting of $L$ for the case $\b = \t = \e_s-\e_t$.

\begin{proof}[Proof of Corollary~\ref{cor:ER}]
Following the proof of Theorem~\ref{thm:truncation_error}, we can derive that $\left|\t^{\top}\x^{\ast}_L - \t^{\top}\x^{\ast}\right| \le \frac{1}{2\gamma} \cdot e^{-\gamma L} \cdot \left\|\D_{\M}^{-1/2}\t\right\|_2 \left\|\D_{\M}^{-1/2}\b\right\|_2$.
As $\M = \L_G$ and $\b = \t = \e_s-\e_t$, we have $\left\|\D_{\M}^{-1/2}\t\right\|_2 \left\|\D_{\M}^{-1/2}\b\right\|_2 = \frac{1}{d_G(s)} + \frac{1}{d_G(t)}$.
Thus, setting $L := \Theta\left(\frac{1}{\gamma} \log\left(\frac{1}{\gamma\eps} \left(\frac{1}{d_G(s)} + \frac{1}{d_G(t)}\right)\right)\right)$ ensures that $\left|\t^{\top}\x^{\ast}_L - \t^{\top}\x^{\ast}\right| \le \frac{1}{2}\eps$.
The corollary then follows by applying Theorems~\ref{thm:MC_cubic} and \ref{thm:bidirectional_RCDD} with this setting of $L$ and noting that $\left\|\D_{\M}^{-1}\b\right\|_{\infty} = 1/\min\big(d_G(s),d_G(t)\big)$ and $\frow(\M) = O(1)$ in this case.
\end{proof}

\section{Acknowledgments}

This research was supported by National Natural Science Foundation of China (No. 92470128, No. U2241212).
We thank the anonymous reviewers for their valuable comments. Mingji Yang thanks Prof. Shang-Hua Teng for inspiring discussions and encouragement, and Guanyu Cui for helpful discussions on effective resistance.

\appendix

\section{Appendix}

\subsection{Concentration Bounds}

\begin{theorem}[The Hoeffding Bound (see, e.g., \protect{\cite[Theorem 4.12]{mitzenmacher2017probability}})] \label{thm:hoeffding}
	Let $X_1,\dots,X_k$ be independent random variables such that for all $j \in [k]$, $\E[X_j] = \mu$ and $\Pr\{a \le X_j \le b\} = 1$. Then
	\begin{align*}
		\Pr\left\{\left|\frac{1}{k}\sum_{j=1}^{k}X_j - \mu\right|\ge \epsilon\right\} \le 2e^{-2k\epsilon^2/(b-a)^2}.
	\end{align*}
\end{theorem}

\subsection{Deferred Proofs} \label{sec:deferred_proofs}

We will use the following lemma, which extends a known symmetry property of PPR on undirected graphs~\cite[Lemma 1]{arvachenkov2013choice} to weighted Eulerian graphs.

\begin{lemma} \label{lem:PPR_symmetry}
	On any weighted Eulerian graph $G$,
	\begin{align*}
		\frac{\vpi_{G,\alpha}(u,v)}{d_G(v)} = \frac{\vpi_{G^{\top},\alpha}(v,u)}{d_G(u)}, \quad \forall u,v \in V.
	\end{align*}
\end{lemma}

\begin{proof}
By the PPR equation~\eqref{eqn:PPR_D} and the Neumann series expansion of the solution $\x^{\ast}$ in Theorem~\ref{thm:x*}, we have
\begin{align*}
	\frac{\vpi_{G,\alpha}(u,v)}{d_G(v)} & = \frac{1}{2}\e_v^{\top} \sum_{\ell=0}^{\infty}\left(\frac{1}{2}\left(\I+(1-\alpha)\D_G^{-1}\A_G^{\top}\right)\right)^{\ell} \D_G^{-1} (\alpha \e_u) \\
	& = \frac{1}{2} \alpha \e_u^{\top} \D_G^{-1} \sum_{\ell=0}^{\infty}\left(\frac{1}{2}\left(\I+(1-\alpha)\A_G\D_G^{-1}\right)\right)^{\ell} \e_v \\
	& = \frac{1}{2} \e_u^{\top} \sum_{\ell=0}^{\infty}\left(\frac{1}{2}\left(\I+(1-\alpha)\D_G^{-1}\A_G\right)\right)^{\ell} \D_G^{-1} (\alpha \e_v) \\
	& = \frac{\vpi_{G^{\top},\alpha}(v,u)}{d_G(u)},
\end{align*}
where the last equality uses the fact that for Eulerian $G$, $\D_G$ is also the outdegree matrix of $G^{\top}$.
\end{proof}

\begin{proof}[Proof of Lemma~\ref{lem:PageRank_lower_bound}]
By Equation~\eqref{eqn:PPR_I}, we have $\vpi_{G,\alpha} = \frac{\alpha}{n}\vone + (1-\alpha)\A_G^{\top}\D_G^{-1}\vpi_{G,\alpha}$, so $\vpi_{G,\alpha}(t) = \frac{\alpha}{n} + (1-\alpha)\sum_{v \to t}\frac{\A_G(v,t)}{\dout_G(v)} \cdot \vpi_{G,\alpha}(v) \ge \frac{\alpha}{n}$.
This also leads to
\begin{align*}
	\vpi_{G,\alpha}(t) & = \frac{\alpha}{n}+(1-\alpha)\sum_{v \to t}\frac{\A_G(v,t)}{\dout_G(v)} \cdot \vpi_{G,\alpha}(v) \ge \frac{\alpha(1-\alpha)}{n}\sum_{v \to t}\frac{\A_G(v,t)}{\dout_G(v)} \\
	& \ge \frac{\alpha(1-\alpha)}{n}\sum_{v \to t}\frac{\A_G(v,t)}{\Deltaout_G} = \frac{\alpha(1-\alpha)\din_G(t)}{n\Deltaout_G},
\end{align*}
proving the first two lower bounds in the first claim.

Now, using the Cauchy-Schwarz inequality, we have
\begin{align*}
	\din_G(t)^2 & = \left( \sum_{v \to t}\A_G(v,t) \right)^2 = \left( \sum_{v \to t}\sqrt{\frac{\A_G(v,t)}{\dout_G(v)}} \cdot \sqrt{\A_G(v,t)\dout_G(v)} \right)^2 \\
	& \le \left( \sum_{v \to t}\frac{\A_G(v,t)}{\dout_G(v)} \right) \left( \sum_{v \to t}\A_G(v,t)\dout_G(v) \right) \le \left( \sum_{v \to t}\frac{\A_G(v,t)}{\dout_G(v)} \right) \left( \sum_{v \to t}\big\|\A_G(\cdot,t)\big\|_{\infty} \dout_G(v) \right) \\
	& \le \left( \sum_{v \to t}\frac{\A_G(v,t)}{\dout_G(v)} \right) \big\|\A_G(\cdot,t)\big\|_{\infty} \|\A_G\|_{1,1},
\end{align*}
which yields the third lower bound that
\begin{align*}
	\vpi_{G,\alpha}(t) & \ge \frac{\alpha(1-\alpha)}{n}\sum_{v \to t}\frac{\A_G(v,t)}{\dout_G(v)} \ge \frac{\alpha(1-\alpha)\din_G(t)^2}{n\|\A_G(\cdot,t)\|_{\infty}\|\A_G\|_{1,1} }.
\end{align*}
On the other hand, using the Cauchy-Schwarz inequality, we also have
\begin{align*}
	\din_G(t)^2 & \le \left( \sum_{v \to t}\frac{\A_G(v,t)}{\dout_G(v)} \right) \left( \sum_{v \to t}\A_G(v,t)\dout_G(v) \right) = \left( \sum_{v \to t}\frac{\A_G(v,t)}{\dout_G(v)} \right) \left( \sum_{v,w \in V}\A_G(v,t)\A_G(v,w) \right) \\
	& \le \left( \sum_{v \to t}\frac{\A_G(v,t)}{\dout_G(v)} \right) \left( \sum_{v,w \in V}\A_G(v,t)^2 \right)^{1/2} \left( \sum_{v,w \in V}\A_G(v,w)^2 \right)^{1/2} \\
	& = \left( \sum_{v \to t}\frac{\A_G(v,t)}{\dout_G(v)} \right) \sqrt{n} \big\|\A_G(\cdot,t)\big\|_2\|\A_G\|_{\Fro},
\end{align*}
which yields the fourth lower bound that
\begin{align*}
	\vpi_{G,\alpha}(t) & \ge \frac{\alpha(1-\alpha)}{n}\sum_{v \to t}\frac{\A_G(v,t)}{\dout_G(v)} \ge \frac{\alpha(1-\alpha)\din_G(t)^2}{n \sqrt{n} \big\|\A_G(\cdot,t)\big\|_2 \|\A_G\|_{\Fro}}.
\end{align*}

Next, assume that $G$ is Eulerian.
Using Lemma~\ref{lem:PPR_symmetry}, we have
\begin{align*}
	\vpi_{G,\alpha}(t) & = \frac{1}{n}\sum_{v\in V}\vpi_{G,\alpha}(v,t) = \frac{d_G(t)}{n}\sum_{v \in V}\frac{\vpi_{G^{\top},\alpha}(t,v)}{d_G(v)} \ge \frac{d_G(t)}{n\Delta_G}\sum_{v \in V}\vpi_{G^{\top},\alpha}(t,v) = \frac{d_G(t)}{n\Delta_G},
\end{align*}
proving the penultimate lower bound.
For the last lower bound for unweighted Eulerian $G$, using Lemma~\ref{lem:PPR_symmetry}, we have
\begin{align*}
	d_G(t) & = \sum_{v \in V}d_G(t)\vpi_{G^{\top},\alpha}(t,v) = \sum_{v \in V}d_G(v)\vpi_{G,\alpha}(v,t) \\
	& = \sum_{v \in V}\vpi_{G,\alpha}(v,t)\sum_{u \to v} 1 = \sum_{v \in V}\vpi_{G,\alpha}(v,t)\sum_{u \to v}\sqrt{d_G(u)} \cdot \sqrt{\frac{1}{d_G(u)}} \\
	& = \sum_{v \in V}\sum_{u \to v} \sqrt{\vpi_{G,\alpha}(v,t) d_G(u)} \cdot \sqrt{\vpi_{G,\alpha}(v,t) \cdot \frac{1}{d_G(u)}}.
\end{align*}
Applying the Cauchy-Schwarz inequality gives
\begin{align*}
	d_G(t) & \le \left(\sum_{v \in V}\sum_{u \to v}\vpi_{G,\alpha}(v,t) d_G(u)\right)^{1/2}\left(\sum_{v \in V}\sum_{u \to v}\vpi_{G,\alpha}(v,t) \cdot \frac{1}{d_G(u)}\right)^{1/2} \\
	& = \left(\sum_{v \in V} \vpi_{G,\alpha}(v,t) \sum_{u \to v} d_G(u)\right)^{1/2} \left(\sum_{u \in V}\frac{1}{d_G(u)}\sum_{u \to v}\vpi_{G,\alpha}(v,t)\right)^{1/2}.
\end{align*}
Observe that for any $u \in V$, by the PPR equation~\eqref{eqn:PPR_I}, we have
\begin{align*}
	\vpi_{G,\alpha}(u,t) = \alpha\indicator\{u = t\} + (1-\alpha)\sum_{u \to v}\frac{1}{d_G(u)} \cdot \vpi_{G,\alpha}(v,t),
\end{align*}
and thus
\begin{align*}
	\frac{1}{d_G(u)}\sum_{u \to v}\vpi_{G,\alpha}(v,t) = \frac{1}{1-\alpha}\big(\vpi_{G,\alpha}(u,t) - \alpha\indicator\{u = t\}\big) \le \frac{\vpi_{G,\alpha}(u,t)}{1-\alpha}.
\end{align*}
Consequently,
\begin{align*}
	d_G(t) & \le \left(\sum_{v \in V} \vpi_{G,\alpha}(v,t) \cdot m\right)^{1/2} \left(\sum_{u \in V}\frac{\vpi_{G,\alpha}(u,t)}{1-\alpha}\right)^{1/2} \\
	& \le \sqrt{n\vpi_{G,\alpha}(t) \cdot m} \cdot \sqrt{\frac{n \vpi_{G,\alpha}(t)}{1-\alpha}} = \frac{n \sqrt{m} \cdot \vpi_{G,\alpha}(t)}{\sqrt{1-\alpha}},
\end{align*}
which leads to
\begin{align*}
	\vpi_{G,\alpha}(t) \ge \frac{\sqrt{1-\alpha} \cdot d_G(t)}{n \sqrt{m}},
\end{align*}
as desired.
\end{proof}

\begin{proof}[Proof of Lemma~\ref{lem:PageRank_upper_bound}]
Since $G$ is Eulerian, by Lemma~\ref{lem:PPR_symmetry}, we have
\begin{align*}
	\vpi_{G,\alpha}(t) & = \frac{1}{n}\sum_{v \in V}\vpi_{G,\alpha}(v,t) = \frac{d_G(t)}{n}\sum_{v \in V}\frac{\vpi_{G^{\top},\alpha}(t,v)}{d_G(v)} \le \frac{d_G(t)}{n\delta_G}\sum_{v \in V}\vpi_{G^{\top},\alpha}(t,v) = \frac{d_G(t)}{n\delta_G},
\end{align*}
finishing the proof.
\end{proof}

\printbibliography

\end{document}